\documentclass{toc}

\tocdetails{%
  title = {Polynomial Identity Testing via Evaluation of Rational Functions},
  number_of_pages = {70},
  number_of_bibitems = {50},   %
  number_of_figures = {1},
  conference_version = {ITCS22}, %
  author = {Ivan Hu, Dieter van Melkebeek, and Andrew Morgan},
  authorlist = {Ivan Hu, Dieter van Melkebeek, Andrew Morgan},
  acmclassification = {Theory of computation $\to$ Algebraic complexity theory,Theory of computation $\to$ Pseudorandomness and derandomization},
  amsclassification = {68Q17, 68Q87, 68Q15},
  keywords = {polynomial identity testing, derandomization, pseudorandomness, lower bounds, vanishing ideal, Gröbner basis}
}   %

\tocdetails{%
volume = 20,      %
number = 1,       %
year = 2024,
received = {November 2, 2022}, %
revised = {April 7, 2024},
published = {July 18, 2024},
doi = {10.4086/toc.2024.v020a001},
}   %

\usepackage{packages}

\usepackage{aumacros}

\begin{document}

\begin{frontmatter}%

\title{Polynomial Identity Testing via \\ Evaluation of Rational Functions
    \titlefootnote{A conference version of this paper appeared in the
    \href{https://doi.org/10.4230/LIPIcs.ITCS.2022.119}{Proceedings of the 13th Innovations in Theoretical Computer Science Conference}~\cite{conf-version}.}}  %

\author[hu]{Ivan Hu}
\author[vanmelkebeek]{Dieter van Melkebeek}
\author[morgan]{Andrew Morgan}

\begin{abstract}
We introduce a hitting set generator for Polynomial Identity Testing based on evaluations of low-degree univariate rational functions at abscissas associated with the variables. We establish an equivalence up to rescaling with a generator introduced by Shpilka and Volkovich, which has a similar structure but uses multivariate polynomials.

We initiate a systematic analytic study of the power of hitting set generators by characterizing their vanishing ideals, \ie, the sets of polynomials that they fail to hit. We provide two such characterizations for our generator. 
First, we develop a small collection of polynomials that jointly produce the vanishing ideal. As corollaries, we obtain tight bounds on the minimum degree, sparseness, and partition class size of set-multilinearity in the vanishing ideal.
Second, inspired by a connection to alternating algebra, we develop a structured deterministic membership test for the multilinear part of the vanishing ideal. We present a derivation based on alternating algebra along with the required background, as well as one in terms of zero substitutions and partial derivatives, avoiding the need for alternating algebra.
 
As evidence of the utility of our analytic approach, we rederive known derandomization results based on the generator by Shpilka and Volkovich and present a new application in derandomization / lower bounds for read-once oblivious algebraic branching programs.
\end{abstract}

%

\end{frontmatter}

\section{Overview}
\label{s.intro}

Polynomial identity testing (PIT) is the fundamental problem of deciding
whether a given multivariate algebraic circuit formally computes the zero
polynomial.
PIT has a simple, efficient randomized algorithm that only needs blackbox
access to the circuit:
Pick a random point and check whether the circuit evaluates to zero on that
particular point.

Despite the fundamental nature of PIT and the simplicity of the randomized
algorithm,
no efficient deterministic algorithm is known---even in the white-box setting,
where the algorithm has access to the description of the circuit.
The existence of such an algorithm would imply long-sought circuit lower
bounds \cite{HeintzSchnorr1980, Agrawal2005, KabanetsImpagliazzo2004}.
Conversely, sufficiently strong circuit lower bounds yield blackbox
derandomization for all of BPP,
the class of decision problems admitting efficient randomized algorithms with
bounded error \cite{NisanWigderson1994,ImpagliazzoWigderson1997}.
Although the known results leave gaps between the two directions, they 
show that PIT constitutes an important stepping stone towards derandomizing BPP, and 
suggest that derandomizing BPP can be achieved in a blackbox fashion if at all.

Blackbox derandomization of PIT for a class of polynomials $\Class$ in the
variables \(x_1\), \ldots, \(x_n\) is equivalent to the efficient
construction of a substitution $G$ that replaces each \(x_i\) by a
low-degree polynomial in a small set of fresh variables
such that, for every nonzero polynomial $p$ from $\Class$,
$p(G)$ remains nonzero {\cite[Lemma 4.1]{ShpilkaYehudayoff2010}}.
We refer to \(G\) as a generator,
the fresh variables are its seed,
and we say that $G$ hits the class $\Class$.
If there are \(l\) seed variables,
and if $p$ and $G$ have degree at most $n^{O(1)}$,
then the resulting deterministic PIT algorithm for $\Class$ makes $n^{O(l)}$
blackbox queries.

Much progress on derandomizing PIT has been obtained by designing such substitutions and analyzing their hitting properties for interesting classes $\Class$.
Shpilka and Volkovich \cite{ShpilkaVolkovich2008} introduced a generator, now dubbed the Shpilka--Volkovich generator, or ``SV generator'' for short. It takes as an additional parameter a positive integer $l$ and can be viewed as an algebraic version of  $l$-wise independence in the sense that any selection of $l$ of the original variables can remain independent while the others are forced to zero. The property is realized using Lagrange interpolation with respect to $n$ distinct elements of the underlying field $\FF$, one element $a_i$ corresponding to each original variable $x_i$. We refer to the elements $a_i$ as \emph{abscissas}; they are also parameters of $\SV$.

\begin{definition}[SV generator]
    \label{def.sv}
    The \emph{Shpilka--Volkovich (SV) Generator} for \(\FF[x_1,\dots,x_n]\) is paramet\-rized by the following data:    %
    \begin{itemize}
        \item A positive integer \(l\).
        \item For each \(i \in [n]\), a distinct abscissa \(a_i\in \FF\).
    \end{itemize}
    The generator \(\SVl\) takes as seed \(l\) pairs of fresh variables
    \((y_1,z_1),\ldots,(y_l,z_l)\) and substitutes
    \begin{equation}\label{eq-sv-sub}
    x_i\gets\sum_{t=1}^l  z_t \cdot L_i(y_t),
    \end{equation}
    where the \emph{Lagrange interpolant} $L_i$ is the unique univariate polynomial of degree at most $n-1$ satisfying $L_i(a_i)=1$ and $L_i(a_j) = 0$ for $j \in [n] \setminus \{i\}$.
\end{definition}
$\SVparam{1}$ takes two seed variables, $y$ and $z$. For any $i \in [n]$, setting $y=a_i$ gets $x_i=z$ while the other variables are set to zero. For larger $l$, $\SVl$ is the sum of $l$ independent copies of $\SVparam{1}$. 

Shpilka and Volkovich proved that $\SVparam{1}$ hits sums of a bounded number of read-once formulas for
$l=O(\log n)$ \cite{ShpilkaVolkovich2008}, later improved to $l=O(1)$ \cite{MinahanVolkovich2017}.
The generator for $l=O(\log n)$ has also been shown to hit multilinear
depth-4 circuits with bounded top fan-in \cite{KMSV2009},
multilinear bounded-read formulas \cite{AvMV2011},
commutative read-once oblivious algebraic branching programs \cite{FSS2014},
\(\Sigma\mathrm{m}\bigwedge\Sigma\Pi^{O(1)}\) formulas (i.e., sums of terms that are the product of a monomial and a power of a bounded-degree polynomial) \cite{Forbes2015},
circuits with locally-low algebraic rank in the sense of
\cite{KumarSaraf2016},
and orbits of simple polynomial classes under invertible linear
transformations of the variables \cite{MediniShpilka2021}.
The generator is an ingredient in other hitting set constructions, as well,
notably constructions using the technique of low-support rank concentration
\cite{ASS2013,AGKS2014,GKST2015,GKS2016,SahaThankey2021,BhargavaGhosh2021}.
It also forms the core of a ``succinct'' generator that hits a variety of
classes, including depth-2 circuits \cite{FSV2018}.

\paragraph{Vanishing ideal.}
In this paper, we initiate a systematic study of the power of a generator $G$ through the set of polynomials $p$ such that $p(G)$ vanishes, which we denote by \(\VanIdeal{G}\). For any fixed generator $G$, \(\VanIdeal{G}\) is closed under addition, and for all \(q\in \FF[x_1,\dots,x_n]\) and \(p\in \VanIdeal{G}\), \(q\cdot p\in \VanIdeal{G}\). By definition, this means that the set \(\VanIdeal{G}\) has the algebraic structure of an ideal. From now on, we refer to \(\VanIdeal{G}\) as the \emph{vanishing ideal} of $G$. Our technical contributions can be understood as precisely characterizing the vanishing ideal of the SV generator.

Characterizations of the vanishing ideal facilitate two objectives:
\begin{description}

  \item[Derandomization.]
   A generator $G$ hits a class $\Class$ of polynomials if and only if $\Class$ and $\VanIdeal{G}$ have at most the zero polynomial in common. For a class $\Class$ defined by a resource bound, $G$ trivially hits $\Class$ if the characterization of the nonzero elements in $\VanIdeal{G}$ is incompatible with being computable within the resource bound. In other words, derandomization of PIT for $\Class$ reduces to proving lower bounds for $\VanIdeal{G}$. By developing explicit structure for polynomials in the ideal, lower bounds become more tractable. 

   More generally, given a characterization of $\VanIdeal{G}$, in order to derandomize PIT for a class \(\Class\) it suffices to design another generator $G'$ that hits merely the polynomials in $\Class \cap \VanIdeal{G}$. As $G$ hits the remainder of $\Class$, combining $G$ with $G'$ yields a generator for all of $\Class$. In this way, one may assume---for free---additional structure about the polynomials in $\Class$, namely that the polynomials moreover belong to $\VanIdeal{G}$. 

  \item[Lower bounds.]
    If we happen to know that $G$ hits the class $\Class$ of polynomials
    computable within some resource bound,
    then any expression for a nonzero polynomial in $\VanIdeal{G}$ yields an
    explicit polynomial that falls outside $\Class$.
    Such a statement is often referred to as hardness of representation,
    and it can be viewed as a lower bound in the model of computation underlying
    $\Class$ (assuming the polynomial can be computed in the model at all).
    Characterizing \(\VanIdeal{G}\) makes explicit the polynomials to which
    the lower bound applies.

\end{description}

We illustrate how to make progress on both objectives through our
characterizations of the SV generator's vanishing ideal.

\paragraph{Rational function evaluations.}
Another contribution of our paper is the development of an alternate view of
the SV generator,
namely as evaluations of univariate rational functions of low degree.
We would like to promote the perspective for its intrinsic appeal and 
applicability.
Among other benefits, it facilitates the study of the vanishing ideal.

The transition goes as follows.
Recall in \expref{Definition}{def.sv} that the SV generator takes as additional parameters a positive integer $l$
and an arbitrary choice of distinct abscissas $a_i \in \FF$ for each of the original variables $x_i$, \(i\in[n]\).
When \(l=1\), $\SVparam{1}$ takes as seed two fresh variables, $y$ and $z$, and can be described succinctly in terms of the Lagrange interpolants $L_i$ for the set of abscissas. Plugging in an explicit expression for the Lagrange interpolants, we have:
\begin{equation}\label{eq.intro.SV}
  x_i \gets z \cdot L_i(y) \doteq
  z \cdot \prod_{j\in [n]\setminus\{i\}} \frac{y - a_{j}}{a_i - a_{j}}.
\end{equation}
By rescaling,
the denominators on the right-hand side of \eqref{eq.intro.SV} can be cleared,
resulting in the following somewhat simpler substitution:
\begin{equation}\label{eq.intro.SVnonnorm}
  x_i \gets z \cdot
    \prod_{j\in [n]\setminus\{i\}} (y - a_{j}).
\end{equation}
The vanishing ideals of \eqref{eq.intro.SVnonnorm} and $\SVparam{1}$ are the
same up to rescaling each variable to match the rescaling from
\eqref{eq.intro.SV} to \eqref{eq.intro.SVnonnorm}.

More importantly,
we apply the change of variables
$z \gets z' / \prod_{j \in [n]} (y - a_j)$. The
resulting substitution now uses rational functions of the seed:
\begin{equation}\label{eq.intro.RFE}
  x_i \gets \frac{z'}{y-a_i}.
\end{equation}
The notion of vanishing ideal naturally extends to rational function
substitutions.
The change of variables from \eqref{eq.intro.SVnonnorm} to
\eqref{eq.intro.RFE} establishes that any polynomial vanishing on
\eqref{eq.intro.SVnonnorm} also vanishes on \eqref{eq.intro.RFE}.
The change of variables is invertible
(the inverse is $z' \gets z \cdot \prod_{j\in [n]} (y-a_j)$),
so any polynomial vanishing on \eqref{eq.intro.RFE} also vanishes on
\eqref{eq.intro.SVnonnorm}.
We conclude that the vanishing ideal of \eqref{eq.intro.RFE} is the same as
that of $\SVparam{1}$ up to rescaling the variables.

Note that, for fixed \(y\) and \(z'\), \eqref{eq.intro.RFE} may be interpreted
as first forming a univariate rational function
\(f(\alpha) = \frac{z'}{y-\alpha}\)
(depending on \(y\) and \(z'\) but independent of \(i\))
and then substituting \(x_i\gets f(a_i)\).
As \(y\) and \(z'\) vary,
\(f\) ranges over all rational functions in \(\alpha\) with numerator degree zero and
denominator degree one.
We denote \eqref{eq.intro.RFE} by \(\RFEparam{0}{1}\),
where $\RFE$ is a short-hand for {\em Rational Function Evaluation},
\(0\) bounds the numerator degree, and \(1\) bounds the denominator degree.

As a generator, \(\RFEparam{0}{1}\) naturally generalizes to $\RFEkl$ for
arbitrary $k, l \in \N$.
\begin{definition}[RFE generator]
  \label{def.rfegen}
  The \emph{Rational Function Evaluation Generator (RFE)} for
  \(\FF[x_1,\ldots,x_n]\) is
  parametrized by the following data:
  \begin{itemize}
    \item
      A non-negative integer \(k\), the \emph{numerator degree}.
    \item
      A non-negative integer \(l\), the \emph{denominator degree}.
    \item For each \(i\in [n]\), a distinct \emph{abscissa} \(a_i \in \FF\).
  \end{itemize}
  The generator \(\RFEkl\) takes as seed a rational function \(f \in \FF(\alpha)\)
  such that \(f\) can be written as \(g/h\) for some \(g,h\in\FF[\alpha]\)
  with \(\deg(g)\le k\), \(\deg(h)\le l\), and \(h(a_i)\ne0\)
  for all \(i\in[n]\).
  From \(f\),
  it generates the substitution \( x_i \gets f(a_i) \) for each \(i\in[n]\).
\end{definition}
There are multiple ways to parametrize the seed of \(\RFEkl\) using scalars; the flexibility to choose is a source of convenience. We refer to \expref{Section}{s.rfe-formal} for a discussion on different parametrizations, as well as on how large the underlying field $\FF$ must be.
As is customary in the context of blackbox derandomization of PIT, we assume that $\FF$ is sufficiently large, possibly by taking a field extension.

The connection between \(\RFEparam{0}{1}\) and $\SVparam{1}$ extends as
follows.
For higher values of $l$, $\SVl$ is defined as the sum of $l$ independent
instantiations of $\SVparam{1}$.
The same transformations as above relate $\SVl$ and the sum of $l$ independent
instantiations of $\RFEparam{0}{1}$. 
Partial fraction decomposition expresses a (non-degenerate) univariate rational function with numerator of
degree $l-1$ and denominator of degree $l$ as a sum of $l$ rational functions with numerators of degree 0 and denominators of degree 1. As a result, $\SVparam{l}$ is equivalent in power to $\RFElml$, up to variable rescaling. See \expref{Section}{s.rfe-formal} for a formal treatment. 

For parameter values $k \ne l-1$,
there is no SV generator that corresponds to $\RFEkl$,
but $\SVparam{\max(k+1,l)}$ encompasses $\RFEkl$ (up to rescaling)
and uses at most twice as many seed variables.
Thus, the $\RFE$-generator and the $\SV$-generator efficiently hit the same
classes of polynomials.
However, $\RFE$ induces simple linear dependencies on the seed variables---as opposed 
to the nonlinear dependencies produced by $\SV$---which enables our approach for determining the vanishing ideal.
The moral is that,
even though polynomial substitutions are sufficient for derandomizing PIT,
it nevertheless helps to consider rational substitutions.
Their use may simplify analysis
and arguably yield more elegant constructions.

As another indication of the power of rational substitutions, an alternate interpretation of the $\RFE$ generator is that it substitutes the ratio of two linear functions of the seed variables, where the coefficients of the linear functions are powers of the abscissas. A generator that only substitutes linear functions---as opposed to a ratio of linear functions---of the seed variables must have seed length $n$ in order to hit all linear polynomials. This is because if the seed length were less than $n$, then there exists a nontrivial linear combination of the $n$ variables that becomes zero after substitution. In contrast, the simplest nontrivial case of $\RFE$, $\RFEparam{0}{1}$, hits all linear polynomials and only needs a seed of length 2. 

\paragraph{Generating set.}
Our first result describes a small and explicit generating set for the
vanishing ideal of $\RFE$.
It consists of instantiations of a single determinant expression.
\begin{theorem}[generating set]\label{thm.ideal.generators}
  Let $k, l, n \in \N$ and let $a_i$ for $i \in [n]$ be distinct elements of $\FF$.
  The vanishing ideal of $\RFEkl$ in $\FF[x_1,\dots,x_n]$ for the
  given choice of abscissas $(a_i)_{i\in[n]}$ is generated by the following
  polynomials over all choices of $k+l+2$ indices
  $i_1, i_2, \dots, i_{k+l+2} \in [n]$:
  \begin{equation}\label{eq.def.gen}
    \VanGenkl[i_1,i_2, \dots, i_{k+l+2}]
    \doteq \det
      \begin{bmatrix}
        a_{i_j}^{k} & a_{i_j}^{k-1} & \dots & 1 &
        a_{i_j}^l x_{i_j} & a_{i_j}^{l-1}x_{i_j} & \dots & x_{i_j}
      \end{bmatrix}_{j=1}^{k+l+2}.
  \end{equation}
  Moreover,
  for any fixed set $C \subseteq [n]$ of size $k+1$, the polynomials $\VanGenkl[C \sqcup L]$
  form a generating set of minimum size when $L$ ranges over all $(l+1)$-subsets of $[n]$ that are disjoint from $C$, where 
  \[\VanGenkl[S] \doteq \VanGenkl[i_1,i_2,\dots,i_{|S|}]\] 
  for $S = \{i_1, \dots, i_{|S|}\} \subseteq [n]$ with $i_1 < i_2 < \dots < i_{|S|}$.
 \end{theorem}
The name ``\(\VanGen\)'' is a shorthand for ``Elementary Vandermonde
Circulation''.
Later in this overview and in \expref{Section}{s.simplicial-repr} we discuss a representation of polynomials using alternating algebra,
with connections to notions from network flow.
In this representation, polynomials in the vanishing ideal coincide with
circulations,
and instantiations of \(\VanGen\) are the elementary circulations.

We refer to the set $C$ in \expref{Theorem}{thm.ideal.generators} as a
\emph{core}.
The core $C$ plays a similar role as in a combinatorial sunflower except that,
unlike the petals of a sunflower,
the sets $L$ do not need to be disjoint outside the
core.

\begin{example}
  \label{ex.rfe01}
  Consider the special case where $k=0$ and $l=1$.
  The generator for $\VanIdeal{\RFEparam{0}{1}}$
  when \(i_1=1\), \(i_2=2\), and \(i_3=3\) is given by
  \[ \VanGenparam{0}{1}[1,2,3]
    \doteq
    \begin{vmatrix*}
      1 & a_1 x_1 & x_1 \\
      1 & a_2 x_2 & x_2 \\
      1 & a_3 x_3 & x_3 \\
    \end{vmatrix*}
    =
    (a_1-a_2)x_1x_2 + (a_2-a_3)x_2x_3 + (a_3-a_1)x_3x_1.
  \]
  For any fixed $i^* \in[n]$,
  the polynomials $\VanGenparam{0}{1}[S]$ form a generating set of
  minimum size when $S$ ranges over all subsets of $[n]$ of size 3 that contain $C = \{i^*\}$.
  As an aside, they also constitute minimal polynomials not computable by read-once formulas, which is consistent with the fact that $\SVparam{1}$ hits all read-once formulas (see \expref{Theorem}{thm.rof}).
\end{example}

In general, the generators $\VanGenkl$ are nonzero, multilinear, homogeneous
polynomials of degree $l+1$, and they have nonzero coefficients for all multilinear monomials of degree $l+1$.
Each generating set of minimum size in \expref{Theorem}{thm.ideal.generators}
yields a Gröbner basis with respect to every monomial order that
prioritizes the variables outside $C$.
A Gröbner basis is a special generating set that allows solving ideal-membership
queries more efficiently, among other problems in computational algebra \cite{CLO2013,AL1994}.
Computing Gröbner bases for general ideals is exponential-space complete \cite{KuhnleM1996,Mayr1997}.
\expref{Theorem}{thm.ideal.generators} represents a rare instance of a natural and
interesting ideal for which we know a small and explicit Gröbner
basis. See the end of \expref{Section}{s.gens} for more background on Gröbner bases.

To gain some intuition about dependencies between the generators $\VanGenkl$,
note that permuting the order of the variables used in the construction of
$\VanGenkl$ yields the same polynomial or minus that polynomial,
depending on the sign of the permutation.
This follows from the determinant structure of $\VanGenkl$
and is the reason why we need to fix the order of the variables in order to
obtain a generating set of minimum size.
More profoundly,
the following relationship holds for every choice of $k+l+3$ indices
$i_1, i_2, \dots, i_{k+l+3} \in [n]$
and every univariate polynomial $q$ of degree at most $k$:
\begin{equation}\label{eq.generator.dependencies}
  \det
  \begin{bmatrix}
    q(a_{i_j}) & 
    a_{i_j}^{k} & a_{i_j}^{k-1} & \dots & 1 & 
    a_{i_j}^l x_{i_j} & a_{i_j}^{l-1}x_{i_j} & \dots & x_{i_j}
  \end{bmatrix}_{j=1}^{k+l+3} = 0.
\end{equation}
The determinant in \eqref{eq.generator.dependencies} vanishes because the
first column of the matrix is a linear combination of the next $k+1$.
A minor expansion across the first column expresses the determinant
of the matrix as a linear combination of minors,
and each minor is an instantiation of $\VanGenkl$.
Since \eqref{eq.generator.dependencies} vanishes, the minor expansion yields a linear dependency for every nonzero polynomial $q$ of degree at most $k$.
In fact,
when \(\{i_1,\ldots,i_{k+l+3}\}\) varies over subsets of \([n]\) containing a
fixed core of size \(k+1\),
the equations \eqref{eq.generator.dependencies} generate \emph{all} linear
dependencies among instances of $\VanGenkl$.

As corollaries to \expref{Theorem}{thm.ideal.generators}
we obtain the following tight bounds on $\VanIdeal{\RFEkl}$.
The bounds hold for every way to choose the parameters in \expref{Definition}{def.rfegen}, including the abscissas.
\begin{corollary}
    \label{cor.rfe.degree}
    The minimum {\em degree} of a nonzero polynomial in $\VanIdeal{\RFEkl}$
    equals $l+1$.
\end{corollary}
\expref{Corollary}{cor.rfe.degree} proves a conjecture by Fournier and Korwar \cite{FournierKorwar2018}
(additional partial results reported in \cite{Korwar2021}) that there
exists a polynomial of degree $l+1$ in $n=2l+1$ variables that $\SVl$
fails to hit.
The conjecture follows because the generators for $\VanIdeal{\SVl}$ have
degree $l+1$ and use $2l+1$ variables.
See also \expref{Corollary}{cor-gen-sv} in \expref{Section}{s.gens}.

As none of the generators contain a monomial of support $l$ or less, the same holds for every nonzero polynomial in $\VanIdeal{\RFEkl}$.
This extends the known property that $\SVl$ hits every polynomial that contains a monomial of support $l$ or less \cite{ShpilkaVolkovich2008}. 
See \expref{Proposition}{prop.ideal.membership.degree-bound} and \expref{Theorem}{thm.ideal.membership} for a strengthening in the case of multilinear polynomials. 

\begin{corollary}
    \label{cor.rfe.sparseness}
    The minimum {\em sparseness}, \ie, number of monomials, of a nonzero
    polynomial in $\VanIdeal{\RFEkl}$ equals $\binom{k+l+2}{l+1}$.
\end{corollary}
The generators $\VanGenkl$ realize the bound in \expref{Corollary}{cor.rfe.sparseness} as they exactly contain all
multilinear monomials of degree $l+1$ that can be formed out of their
$k+l+2$ variables.
The claim that no nonzero polynomial in $\VanIdeal{\RFEkl}$ contains fewer
than $\binom{k+l+2}{l+1}$ monomials requires an additional combinatorial
argument (see \expref{Lemma}{lem.rfe-hits-sparse}). 
It is a (tight) quantitative strengthening of the known property that
$\SVl$ hits every polynomial with fewer than $2^l$ monomials
\cite{AvMV2011,GKST2015,Forbes2015,FSV2018}.
Note that for $k=l-1$ we have that
$\binom{k+l+2}{l+1} = \binom{2l+1}{l+1} = \Theta(2^{2l}/\sqrt{l})$.
One consequence is that for $\SVl$ to hit all polynomials with $m$ monomials, a seed length of $l=\Omega(\log m)$ is required. In particular, hitting sparse polynomials requires $l = \Omega(\log n)$. 

Another consequence deals with set-multilinearity, a common restriction in works on derandomizing PIT
and algebraic circuit lower bounds.
A polynomial $p$ of degree $l+1$ in a set of variables
$\{x_1,\ldots,x_n\}$ is said to be set-multilinear
if $[n]$ can be partitioned as
$[n] = X_1 \sqcup X_2 \sqcup \dots \sqcup X_{l+1}$
such that every monomial in $p$ is a product
$x_{i_1}\cdot x_{i_2} \cdot \dots \cdot x_{i_{l+1}}$,
where $i_j \in X_j$.
Note that set-multilinearity implies multilinearity but not the other
way around.
As the generators $\VanGenkl$ are not set-multilinear,
it is not immediately clear from \expref{Theorem}{thm.ideal.generators}
whether $\VanIdeal{\RFEkl}$ contains nontrivial set-multilinear
polynomials of any degree.
However, a variation on the construction of the generators $\VanGenkl$
yields explicit set-multilinear homogeneous polynomials in
\(\VanIdeal{\RFEkl}\) of degree $l+1$ where each \(X_j\) has size $k+2$ (see \expref{Definition}{def.smvan1}).
We denote them by $\SMVanOnekl$, where $\SMVanOne$ stands for
``Elementary Set-Multilinear Vandermonde Circulation''.
$\SMVanOnekl$ contains all monomials of the form
$x_{i_1}\cdot x_{i_2} \cdot \dots \cdot x_{i_{l+1}}$ with $i_j \in X_j$.
For any variable partition $(X_1, X_2, \dots, X_{l+1})$ with
$|X_1| = \cdots = |X_{l+1}| = k+2$,
$\SMVanOnekl$ is the only set-multilinear polynomial in $\VanIdeal{\RFEkl}$
with that variable partition, up to a scalar multiple, and exhibits the following extremal property. See also \expref{Theorem}{thm.smvan1}.
\begin{corollary}
    \label{cor.rfe.partition-class-size}
    The minimum {\em partition class size} of a nonzero set-multilinear
    polynomial of degree $l+1$ in $\VanIdeal{\RFEkl}$ equals $k+2$.
\end{corollary}

\paragraph{Membership test.}
Our second characterization of the vanishing ideal of $\RFE$ can be viewed as
a structured membership test.
Given a polynomial \(p\),
there is a generic way to test whether \(p\) belongs to the
vanishing ideal of a generator $G$,
namely by symbolically substituting \(G\) into \(p\) and verifying that the
result simplifies to zero.
When $G$ is a polynomial substitution,
the well-known transformation of a generator into a deterministic blackbox PIT
algorithm 
yields another test:
Verify \(p(G)=0\) for a sufficiently large set of substitutions into the seed
variables.
By clearing denominators,
the same goes for rational substitutions like $\RFEkl$.

While the generic test works, one cannot extract $G$-specific insight into whether or why $G$ hits any particular polynomial. In contrast, our membership test uses the specific structure of $G$ and provides useful insight. Building on the generating set of \expref{Theorem}{thm.ideal.generators},
we state our structured test for membership of multilinear polynomials in $\VanIdeal{\RFEkl}$ in terms of partial derivatives and zero substitutions. Several prior papers demonstrated the utility of those operations in the context of derandomizing PIT using the SV generator,
especially for syntactically multilinear models
\cite{ShpilkaVolkovich2008,KMSV2009,AvMV2011}.

\begin{theorem}[membership test for multilinear polynomials]\label{thm.ideal.membership}
  Let $k, l, n \in \N$ and let $a_i$ for $i \in [n]$ be distinct elements of $\FF$.
  A multilinear polynomial $p \in \FF[x_1,\dots,x_n]$ belongs to
  $\VanIdeal{\RFEkl}$ if and only if both of the following conditions hold:
  \begin{enumerate}
    \item\label{thm.ideal.membership.cond1}
      There are no monomials of degree $l$ or less,
      nor of degree $n-k$ or more, in \(p\).
    \item\label{thm.ideal.membership.cond2}
      For all disjoint subsets $K,L \subseteq [n]$ with $|K|=k$ and $|L|=l$,
      $\PartialZero{L}{K}{p}$ is zero upon the following
      substitution for each \(i\in[n]\setminus (K\cup L)\), where $z$ denotes a fresh variable:
      \begin{equation}\label{eq.thm.ideal.membership.sub}
        x_i \gets z\cdot \frac%
          { \prod_{j \in K} (a_i - a_{j}) }%
          { \prod_{j \in L} (a_i - a_{j}) }.
      \end{equation}
  \end{enumerate}
\end{theorem}

A few technical comments regarding the statement are in order.
The first part of \expref{condition}{thm.ideal.membership.cond1} in
\expref{Theorem}{thm.ideal.membership} generalizes the known property that $\SVl$ hits
every multilinear polynomial that contains a monomial of degree $l$ or less
\cite{ShpilkaVolkovich2008}. As for the second part, see \expref{Proposition}{prop.ideal.membership.degree-bound} for more discussion.  
The two parts together imply that all multilinear polynomials on $n \le k+l+1$ variables are hit by $\RFEkl$.

In \expref{condition}{thm.ideal.membership.cond2},
$\PartialZero{L}{K}{p}$ denotes the polynomial obtained by taking the partial derivative of $p$ with respect to every variable in $L$ and setting all the variables in $K$ to zero.
Because of the multilinearity, the order of the operations does not matter,
and the resulting polynomial depends only on variables in
\([n]\setminus (K\cup L)\). The substitution \eqref{eq.thm.ideal.membership.sub} can be viewed as $x_i \gets f(a_i)$, where 
\begin{equation*}
  f(\alpha) = z \cdot f_{K,L}(\alpha) \doteq z \cdot \frac{
      \prod_{j \in K} (\alpha - a_{j})
    }{
      \prod_{j \in L} (\alpha - a_{j})
    }
\end{equation*}
is a valid seed of \(\RFEkl\) for polynomials in the variables $x_i$, $i \in [n] \setminus (K \cup L)$.
Upon substitution, $\PartialZero{L}{K}{p}$ becomes a univariate polynomial $q$ of degree at most $n-k-l$ in the fresh variable $z$. In the case where $p$ is homogeneous, $q$ has at most one term, and $q$ is nonzero if and only $q$ is nonzero at $z=1$. In general, for any fixed set $Z$ of $n-k-l+1$ elements of $\FF$, $q$ is nonzero if and only if $q$ is nonzero at some $z \in Z$. 

\expref{Theorem}{thm.ideal.membership} can be understood as stating that a multilinear polynomial \(p\) is hit by \(\RFEkl\)
if and only if
\(p\) has a monomial supported on few or all-but-few variables,
or else there is a set of $k$ zero substitutions, \(K\),
and a set of $l$ partial derivatives, \(L\),
whose application to \(p\) leaves a polynomial that is nonzero after
substituting \(x_i \gets z\cdot f_{K,L}(a_i)\).
By judiciously choosing variables for the zero substitutions and partial derivatives, prior papers managed to simplify polynomials $p$ of certain types and reduce PIT for $p$ to PIT for simpler instances, resulting in efficient recursive algorithms.
In \expref{Section}{s.test},
we develop a general framework for such algorithms
and prove correctness directly from \expref{Theorem}{thm.ideal.membership}.
Moreover,
because \expref{Theorem}{thm.ideal.membership} is a precise characterization,
any argument that \(\SV\) or \(\RFE\) hits a class of multilinear polynomials
can be converted into one within our framework,
\ie, into an argument based on zero substitutions and partial derivatives.
Thus, \expref{Theorem}{thm.ideal.membership} shows that these tools harness the complete
power of $\SV$ and $\RFE$ for multilinear polynomials.

\paragraph{Applications.}
We illustrate the utility of our characterizations of the vanishing ideal of
$\RFE$ in the two directions mentioned before.

\subparagraph{Derandomization.}
To start,
we demonstrate how \expref{Theorem}{thm.ideal.membership} yields an alternate
proof of the result from \cite{MinahanVolkovich2017}
that \(\SVparam{1}\)---equivalently, \(\RFEparam{0}{1}\)---hits
every nonzero read-once formula $F$.
Whereas the original proof hinges on a clever ad-hoc argument,
our proof (described in  \expref{Section}{s.test}) is entirely systematic and amounts to
a couple straightforward observations in order to apply
\expref{Theorem}{thm.ideal.membership}.

As a proof of concept of the additional power of our characterization for derandomization,
we make progress in a well-studied model for algebraic computation, namely read-once oblivious algebraic
branching programs (ROABPs). An ROABP consists of a layered digraph, the width of which constitutes an important complexity parameter. We refer to \expref{Section}{s.roabp.background} for more background.

\begin{theorem}[ROABP hitting property]\label{thm.roabp}
  For any integer $l \ge 1$, $\SVl$ hits the class of polynomials computed by
  read-once oblivious algebraic branching programs of width less than
  $(l/3)+1$ that contain a monomial of degree at most $l+1$.
\end{theorem}
To the best of our knowledge,
\expref{Theorem}{thm.roabp} is incomparable to the known results for ROABPs
\cite{RazShpilka2005,JQS2009,JQS2010,ForbesShpilka2013,FSS2014,AGKS2014,AFSSV2015,GKST2015,GKS2016,GuoGurjar2020,SahaThankey2021,BhargavaGhosh2021}.
Without the restriction that the polynomial has a monomial of degree at most
\(l+1\),
\expref{Theorem}{thm.roabp} would imply a fully blackbox polynomial-time identity test
for the class of constant-width ROABPs.
No such test has been proven to exist at this time;
prior work requires either quasipolynomial time or else opening the
blackbox,
such as by knowing the order in which the variables are read.

With the restriction, 
hitting the class in \expref{Theorem}{thm.roabp} with $\SVl$ represents fairly specialized progress. This is because $\SVparam{l+1}$ is well-known to hit every polynomial containing a monomial of support $l+1$ or less, and thus hits the class in \expref{Theorem}{thm.roabp}, irrespective of the restriction on ROABP width. That said, the method of proof of \expref{Theorem}{thm.roabp} diverges
significantly from prior uses of the SV generator
and therefore may be of independent interest.
We elaborate on the method more when we discuss the techniques of this paper, but for now,
we point out that most prior uses of the SV generator rely on combinatorial
arguments,
\ie, arguments that depend only on which monomials are present in
the polynomials to hit.
\expref{Theorem}{thm.roabp} necessarily goes beyond this
because there is a polynomial in \(\VanIdeal{\SVl}\) of degree \(l+1\) that
has the same monomials as a polynomial computed by an ROABP of width 2, which by \expref{Theorem}{thm.roabp} is not in \(\VanIdeal{\SVl}\) for $l \ge 4$.
Namely, any instance of \(\SMVanOneparam{l-1}{l}\) contains exactly all the
monomials of the form \(x_{i_1}\cdot x_{i_2}\cdot \dots \cdot x_{i_{l+1}}\)
with \((i_1,\ldots,i_{l+1})\in X_1\times \cdots \times X_{l+1}\) for some
disjoint sets \(X_j\);
the same goes for \(\prod_j \sum_{i_j\in X_j} x_{i_j}\),
which is computed by an ROABP of width 2.

\subparagraph{Lower bounds.}
Our result for ROABPs also illustrates this direction. Our derandomization result for the class in \expref{Theorem}{thm.roabp} is equivalent to the following lower bound. 
\begin{theorem}[ROABP lower bound]
  \label{thm.roabp.lb}
  For any integer $l \ge 1$, any nonzero polynomial in $\VanIdeal{\SVl}$ that contains a monomial of degree at most $l+1$, requires ROABP width at least $(l/3)+1$.
\end{theorem}
Such a lower bound is interesting because there are appealing polynomials meeting the conditions, in particular the generators $\VanGenparam{l-1}{l}$ as well as $\SMVanOneparam{l-1}{l}$.
Other hardness of representation results follow in a similar manner from prior hitting properties of $\SV$ in the literature.
The following lower bounds apply to computing both
\(\VanGenparam{l-1}{l}\) and \(\SMVanOneparam{l-1}{l}\):
\begin{itemize}
  \item
    Any syntactically multilinear formula must have at least
    \(\Omega(\log(l)/\log\log(l))\) reads of some variable
    \cite[Theorem~6.3]{AvMV2011}.
  \item
    Any sum of read-once formulas must have at least \(\Omega(l)\) summands
    \cite[Corollary~5.2]{MinahanVolkovich2017}.
  \item
    There exists an order of the variables such that any ROABP with that order
    must have width at least \(2^{\Omega(l)}\)
    \cite[Corollary~4.3]{FSS2014}.
  \item
    Any \(\Sigma\mathrm{m}\bigwedge\Sigma\Pi^{O(1)}\) formula must have top
    fan-in at least \(2^{\Omega(l)}\)
    \cite{Forbes2015}; 
    see also \cite[Lemma~5.12]{FSV2018}.
  \item
    Lower bounds over characteristic zero for circuits with locally-low
    algebraic rank
    \cite[Lemma~5.2]{KumarSaraf2016}.
\end{itemize}

\paragraph{Techniques.}
Many of our results ultimately require showing that,
under suitable conditions,
\(\RFE\) hits a polynomial \(p\).
A recurring analysis fulfills this role in the proofs of
\expref{Theorems}{thm.ideal.generators},~\ref{thm.ideal.membership},~and~\ref{thm.roabp}.
We take intuition from the analytic setting (\eg, \(\FF = \R\))
and study the behavior of \(p(\RFE)\) as a function of the seed's zeroes and poles. When they are near the abscissas of chosen variables of
\(p\), the behavior is dominated by the contributions of the monomials of \(p\)
for which
the variables with abscissas near zeros have minimal degree and
the variables with abscissas near poles have maximal degree.
This allows us to analyze a first approximation to \(p(\RFE)\) by ``zooming in'' on
the contributions of the monomials in which the chosen
variables have extremal degrees.
If the first approximation is nonzero,
then we can conclude that \(\RFE\) hits \(p\).
We capture the technique in our Zoom Lemma (\expref{Lemma}{lem.multi-zsub-deriv}).
Formal Laurent series can express the analytic intuition purely algebraically. We provide a proof from first principles that does not require any background in Laurent series and works over all fields.

\expref{Theorem}{thm.ideal.generators} states the equality $I = \VanIdeal{\RFEkl}$ of two ideals,
where $I$ denotes the ideal generated by all instantiations of $\VanGenkl$,
and $\VanIdeal{\RFEkl}$ the vanishing ideal of $\RFEkl$.
\begin{itemize}
  \item
    The inclusion $\subseteq$ follows from linearizing the defining equations
    of $\RFEkl$ (\expref{Lemma}{lemma.vangen-vanishes}).
    The technique mirrors the use of resultants to compute implicit equations
    for rational plane curves.
    This is where the switch from $\SV$ to $\RFE$ helps. 
  \item
    To establish the inclusion $\supseteq$ we first show that the
    equivalence class of any polynomial $p$ modulo $I$
    contains a representative $r$ whose monomials exhibit the combinatorial
    structure of a core (\expref{Lemma}{lem.vangen-coredform}). If $p \not\in I$, $r$ is nonzero.
    The core structure of $r$ then allows us to apply the zooming-in mechanism such that the resulting first approximation to $r$ is nonzero, in which case the \hyperref[lem.multi-zsub-deriv]{Zoom Lemma} tells us that $\RFEkl$ hits $r$ (\expref{Lemma}{lem.rfe-hits-cored}). By the inclusion $\subseteq$, we conclude that $\RFEkl$ hits $p$.
\end{itemize}

The proof of \expref{Theorem}{thm.ideal.membership} also relies on the
\hyperref[lem.multi-zsub-deriv]{Zoom Lemma}. 
Membership to the ideal is equivalent to the vanishing of all coefficients of
the expansion of $p(\RFE)$.
The application of the Zoom Lemma can be viewed as determining a small number of coefficients
sufficient to guarantee that their vanishing implies all coefficients vanish.
The restriction to multilinear polynomials $p$ allows us to express the
zoomed-in contributions of $p$ as the result of applying partial derivatives and
zero-substitutions.

\expref{Theorem}{thm.roabp} makes use of the characterization of the minimum width
of a read-once oblivious algebraic branching program computing a polynomial
$p$ as the maximum rank of the monomial coefficient matrices of $p$ for
various variable partitions \cite{Nisan1991}.
The result is effectively about polynomials \(p\) that are homogeneous of degree $l+1$,
in which case the monomial coefficient matrices have a block-diagonal structure with $l+2$ blocks.
An application of the \hyperref[lem.multi-zsub-deriv]{Zoom Lemma} in the contrapositive yields linear equations
between elements of consecutive blocks under the assumption that $\SVl$ fails to hit $p$.
When some block is zero, the equations yield a Cauchy system on the rows or columns of its
neighboring blocks. Based on the fact that Cauchy systems have full rank and exploiting the specific structure, we deduce several constraints on the row-space/column-space of the neighboring blocks. A careful analysis and case analysis based on the number of zero blocks yields a rank lower bound of at least $(l/3)+1$ for a well-chosen partition of the variables.

We point out that, in the preceding application,
the \hyperref[lem.multi-zsub-deriv]{Zoom Lemma} is instantiated several times in parallel to form a large
system of equations on the coefficients of \(p\),
and the whole system is necessary for the analysis.
This stands in contrast to most prior work using \(\SV\),
which can be cast as using knowledge of how \(p\) is computed to guide a search for a
\emph{single} fruitful instantiation of the \hyperref[lem.multi-zsub-deriv]{Zoom Lemma}.

\paragraph{Alternating algebra representation.}
The inspiration for several of our results stems from expressing the
polynomials $\VanGenkl$ using concepts from alternating algebra
(also known as exterior algebra or Grassmann algebra).
In fact, the membership test for the ideal generated by the instantiations of $\VanGenkl$ in \expref{Theorem}{thm.ideal.membership} is based on the relationship \(\bdry^2=0\) from alternating algebra.
Our original statement and proof of the theorem made use of that framework,
but we managed to eliminate the alternating algebra afterwards.
Still, as we find the perspective insightful and potentially helpful for
future developments,
we describe the connection briefly here
and in more detail in \expref{Section}{s.simplicial-repr}.
We explain the intuition for the
simple case where the degree of the polynomial $p$ equals $l+1$.
In that setting, belonging to the ideal generated by the polynomials
$\VanGenkl$  is equivalent to being in their linear span.

The alternating algebra $\Simplices$ of a vector space $\SpanVertices$ over a field $\FF$ consists
of the closure of $\SpanVertices$ under an additional binary operation,
referred to as ``wedge'' and denoted $\wedge$,
which is bilinear, associative, and satisfies
\begin{equation}\label{eq.lindepzero}
  u \wedge u = 0
\end{equation}
for every $u \in \SpanVertices$.
This determines a well-defined algebra.
When the characteristic of $\FF$ is not 2,
\eqref{eq.lindepzero} can equivalently be understood as anti-commutativity:
\begin{equation}\label{eq.anti-commute}
  u_1 \wedge u_2 = -(u_2 \wedge u_1)
\end{equation}
for every $u_1, u_2\in \SpanVertices$.
For any characteristic and $u_1, u_2, \dots, u_t \in \SpanVertices$,
\begin{equation} \label{eq.wedge}
  u_1 \wedge u_2 \wedge \dots \wedge u_t
\end{equation}
is nonzero iff the $u_i$'s are linearly independent,
and any permutation of the order of the vectors in \eqref{eq.wedge} yields the
same element of $\Simplices$ up to a sign.
The sign equals the sign of the permutation,
whence the name ``alternating algebra.''
If $\SpanVertices$ has a basis $\Vertices=\{v_1,\dots,v_n\}$ of size $n$,
then a basis for $\Simplices$ can be formed by all $2^n$ expressions of the form
\eqref{eq.wedge},
where the $u_i$'s range over all subsets of $\Vertices$
and are taken in some fixed order.
Considering the elements of $\Vertices$ as vertices,
the basis elements of $\Simplices$ can be thought of as the oriented simplices of all
dimensions that can be built from $\Vertices$.

Anti-commutativity arises naturally in the context of network flow,
where $V$ denotes the vertices of the underlying graph,
and a wedge $v_1 \wedge v_2$  of level $t=2$ represents one unit of flow from
$v_1$ to $v_2$.
Equation~\eqref{eq.anti-commute} reflects the fact that one unit of flow
from $v_1$ to $v_2$ cancels with one unit of flow from $v_2$ to $v_1$.
The adjacent levels $t=1$ and $t=3$ also have natural interpretations in the
flow setting:
$v_1$ (the element of $\Simplices$ of the form \eqref{eq.wedge} with $t=1$) represents
one unit of surplus flow at $v_1$ (the vertex of the graph),
and $v_1 \wedge v_2 \wedge v_3$ abstracts a circulation of one
unit along the directed cycle $v_1 \to v_2 \to v_3 \to v_1$.

The different levels are related by so-called boundary maps, which are linear transformations that map a simplex to a linear
combination of its subsimplices of one dimension less.
The maps are parametrized by a linear weight function $w: \SpanVertices \to \FF$,
and defined on the vertices by
\begin{equation}\label{eq.boundary}
  \bdry_w:
  v_1 \wedge v_2 \wedge \dots \wedge v_t
  \mapsto
  \sum_{i=1}^t
    (-1)^{i+1}
    w(v_i) \;
    v_1 \wedge \dots \wedge v_{i-1} \wedge v_{i+1} \wedge \dots \wedge v_t,
\end{equation}
an expression resembling the minor expansion of a determinant along a column
$[ w(v_i) ]_{i=1}^t$.
In the flow setting,
using $w \equiv 1$,
applying \(\bdry_1\) to \(v_1 \wedge v_2\)
yields $v_2 - v_1$, the superposition of demand at $v_1$
and surplus at $v_2$ corresponding to one unit of flow from $v_1$ to $v_2$.
Likewise, $\bdry_1$ sends the abstract elementary circulation
\(v_1 \wedge v_2 \wedge v_3\) to the superposition of the three edge flows
that make up the directed 3-cycle $v_1\to v_2 \to v_3 \to v_1$. 
A linear combination $p$ of terms \eqref{eq.wedge} with $t=2$ represents a
valid circulation iff it satisfies conservation of flow at every vertex,
which can be expressed as $\bdry_1(p)=0$,
\ie, $p$ is in the kernel of $\bdry_1$.
An equivalent criterion is for $p$ to be the superposition of
circulations around directed 3-cycles,
which can be expressed as $p$ being in the image of $\bdry_1$.
The relationship $\im(\bdry_w) = \ker(\bdry_w)$ between the image and the kernel of a boundary map holds for any nonzero $w$, and generalizes to composed boundary maps:
For any linearly independent \(w_1, \ldots, w_{k+1}\), it holds that
\begin{equation}\label{eq.image.kernel}
  \im\left(
    \bdry_{w_{k+1}} \circ \bdry_{w_{k} }\circ \dots \circ \bdry_{w_1}
  \right)
  =
  \bigcap_{r=1}^{k+1} \ker\left( \bdry_{w_r} \right).
\end{equation}
When \(w_1, \ldots, w_{k+1}\) are linearly dependent, \(\bdry_{w_{k+1}} \circ \dots \circ \bdry_{w_1}\) is the zero map.

In the context of \(\RFE\),
the set $V$ consists of a distinct vertex $v_i$ for each variable $x_i$,
and simplices correspond to multilinear monomials.
The anti-commutativity of $\wedge$ coincides with the fact that swapping two
arguments to \(\VanGenkl\) means swapping two rows in \eqref{eq.def.gen},
which changes the sign of the determinant.
Using boundary maps, $\VanGenkl[i_1,i_2,\dots,i_{k+l+2}]$
can be viewed as
$\bdry_{\omega} (v_{i_1} \wedge v_{i_2} \wedge \dots \wedge v_{i_{k+l+2}})$, 
where
\(\bdry_{\omega} \doteq \bdry_{w_{k+1}} \circ \bdry_{w_k} \circ \dots \circ
\bdry_{w_1}\) and
$w_r(v_i) \doteq (a_i)^{r-1}$.
By \eqref{eq.image.kernel},
this means that $\VanGenkl$ is in the kernel of $\bdry_{w_r}$ for each
$r \in [k+1]$,
or equivalently,
in the kernel of $\bdry_w$ for each $w: \SpanVertices \to \FF$
of the form $w(v_i)=q(a_i)$
where $q$ is a polynomial of degree at most $k$.
In fact,
\eqref{eq.image.kernel} implies that the linear span of the generators
$\VanGenkl$ consists exactly of the polynomials of degree $l+1$ in this
kernel. The linear span coincides with the polynomials of degree $l+1$ in the ideal generated by the polynomials $\VanGenkl$. For multilinear polynomials, being in the kernel can be expressed in terms of zero substitutions and partial derivatives as in \expref{Theorem}{thm.ideal.membership}. This yields an alternate route for deriving our membership test for multilinear polynomials of degree $d=l+1$ in the ideal generated by the instantiations of $\VanGenkl$, which by \expref{Theorem}{thm.ideal.generators} agrees with $\VanIdeal{\RFEkl}$. In the basic case where $k=0$ and $l=1$, only the weight function $w \equiv 1$ needs to be considered and the kernel requirement coincides with flow conservation. We refer to \expref{Section}{s.simplicial-repr} for the general multilinear case of arbitrary degree. 

\paragraph{Related recent work and further research.}
We propose to systematically investigate the power of generators by characterizing their vanishing ideals. As we demonstrated for $\SV$ and $\RFE$, such characterizations can exhibit both strengths and weaknesses of the generator. 

Specific other generators of interest include Klivans--Spielman \cite{KS2001} and generators based on the matrix rank condenser by Gabizon and Raz \cite{GR2008,KS2011,FS2012}.  A related direction is figuring out how vanishing ideals are affected when manipulating generators. Examples include the $\RFE$ generator with pseudorandom abscissas, or work that relates the vanishing ideal of a combination of generators to the vanishing ideals of the constituent generators. In particular, a combination of $\SV$ with Klivans--Spielman appears in the literature \cite{KMSV2009, FSS2014, FSV2018}, where the latter is used to effectively hit sparse polynomials, which our results show that $\SV$ does not.

The generator $\SVl$ is the canonical example of an $l$-wise independent generator in the algebraic setting. Understanding the power of $l$-wise independent generators more broadly, \eg, as formalized in \cite{FSTW21, MediniShpilka2021}, could lead to useful insights for derandomizing PIT. This work demonstrates explicit polynomials like \(\VanGenparam{l-1}{l}\) and \(\SMVanOneparam{l-1}{l}\) that are not automatically hit by \(l\)-wise independence as they are not hit by \(\SVl\). Is there a deeper underlying reason related to $l$-wise independence?

A generator hits all polynomials from a resource-bounded class iff no nonzero polynomial in the vanishing ideal can be computed within those resources. Chatterjee and Tengse \cite{CT2023} recently showed the following generic limitation: The vanishing ideal of any generator computable by algebraic circuits of polynomial size in the number of variables contains a nonzero polynomial computable in VPSPACE. From this perspective, our results exhibit a weakness of $\SV$ and $\RFE$ in that their vanishing ideals contain nonzero polynomials from the presumably much smaller class VBP. In fact, $\VanGenkl$ is a polynomial depending on only $k+l+2$ variables and is computable by a branching program of size polynomial in the number of variables. Thus, in order to hit all branching programs of size $s$, $\SV$ and $\RFE$ require a seed length $k+l+2=s^{\Omega(1)}$. 

A related question is whether the generators we have identified have minimal (or approximately minimal) complexity in the vanishing ideal. Andrews and Forbes \cite{AndrewsForbes2022} recently established such a result for a generator that substitutes an $n\times m$ matrix of variables with the product of $n\times l$ and $l \times m$ matrices of variables for small $l$. The vanishing ideal of their generator is straightforwardly generated by $(l+1) \times (l+1)$ minors. For this vanishing ideal the authors manage to show that every nonzero element is at least as hard as computing $\Theta(l^{1/3})\times \Theta(l^{1/3})$ determinants (under simple reductions and in the sense of border complexity). 

Lastly, we list some avenues for improving specific aspects of our results. \expref{Theorem}{thm.ideal.membership} represents an elementary deterministic membership test in the vanishing ideal of $\RFEkl$ for \emph{multilinear} polynomials. Can the elementary test can be extended to all polynomials? From the alternating algebra perspective, the test relies an the convenient one-to-one correspondence between multilinear polynomials and elements of the alternating algebra. For general polynomials, this correspondence is no longer one-to-one, and the resulting membership test is nondeterministic.

Another target is eliminating degree restrictions for our characterizations of specialized classes of polynomials, in particular in \expref{Theorem}{thm.roabp} for ROABPs.
Removing the degree restriction for ROABPs would result in a full blackbox derandomization of constant-width ROABPs. An alternative possibility is that, through better analysis of the vanishing ideal, it turns out that $\RFE$ has limitations in derandomizing constant-width ROABPs.

\paragraph{Organization.}
We start in \expref{Section}{s.rfe-formal} with formal aspects of the $\RFE$ generator that have been omitted from the informal discussion thus far. 
We construct the generating set for the vanishing ideal
(\expref{Theorem}{thm.ideal.generators}) in \expref{Section}{s.gens},
followed by the \hyperref[lem.multi-zsub-deriv]{Zoom Lemma} in \expref{Section}{s.zoom}.
The ideal membership test (\expref{Theorem}{thm.ideal.membership}) is developed in
\expref{Section}{s.test}.
We present the results on sparseness in \expref{Section}{s.sparse},
and the ones on set-multilinearity in \expref{Section}{s.set-multiaffine}.
Background on ROABPs and our result on derandomizing PIT for ROABPs
(\expref{Theorem}{thm.roabp}) are covered in \expref{Section}{s.constant-width-roabp}.
We end our paper in \expref{Section}{s.simplicial-repr} with a further discussion
of the alternating algebra representation and an alternate derivation of the membership test for multilinear polynomials in the ideal generated by the instances of $\VanGenkl$.

\section{RFE Generator}
\label{s.rfe-formal}

We defined the $\RFE$ generator in the overview but omitted some of the formal details. In this section, we discuss different parametrizations of $\RFE$ as well as how to obtain deterministic blackbox PIT algorithms from a generator and how large the underlying field $\FF$ must be. We also state and establish the precise relationship between $\RFElml$ and $\SVl$.

In \expref{Definition}{def.rfegen},
we defined \(\RFE\) as a set of substitutions formed by varying the seed \(f\)
over certain rational functions with coefficients in \(\FF\).
Meanwhile,
our analyses proceed by parametrizing \(f\) by scalars,
abstracting the scalar parameters as fresh formal variables,
and calculating in the field of rational functions in those variables.
The approaches are equivalent over large enough fields, 
and the flexibility to choose is a source of convenience.
Here are some natural parametrizations of \(f\):
\begin{description}
  \item[Coefficients.]
    Select scalars \(g_0, \ldots, g_k, h_0,\ldots, h_l \in\FF\) and
    set
    \[
      f(\alpha) =
        \frac{
          g_k\alpha^k
          + g_{k-1} \alpha^{k-1}
          + \cdots
          + g_1 \alpha
          + g_0
        }{
          h_l\alpha^l
          + h_{l-1} \alpha^{l-1}
          + \cdots
          + h_1 \alpha
          + h_0
      },
    \]
    ignoring choices of \(h_0,\ldots,h_l\) for which the denominator
    vanishes at some abscissa.

  \item[Evaluations.]
    Fix two collections, \(B = \{b_1,\ldots,b_{k+1}\}\) and
    \(C = \{c_1,\ldots,c_{l+1}\}\), each of distinct scalars from \(\FF\).
    Then select scalars \(g_1, \ldots, g_{k+1}\) and \(h_1,\ldots,h_{l+1}\)
    and set
    \[
      f(\alpha) = \frac{ g(\alpha) }{ h(\alpha) }
    \]
    where \(g\) is the unique degree-\(k\) polynomial with
    \(g(b_1)=g_1\), \(g(b_2)=g_2\), \ldots, \(g(b_{k+1}) = g_{k+1}\),
    and \(h\) is defined similarly with respect to \(C\).
    Choices of \(h_1,\ldots,h_{l+1}\) that lead \(h\) to vanish at some
    abscissa are ignored.

    Note that an explicit formula for \(g\) and \(h\) in terms of the
    parameters can be obtained using the Lagrange interpolants with
    respect to \(B\) and \(C\).

  \item[Roots.]
    Select scalars \(z, s_1,\ldots,s_{k'}, t_1,\ldots,t_{l'} \in \FF\)
    for some \(k' \le k\) and \(l'\le l\) and set
    \[
      f(\alpha) = z \cdot
        \frac{ (\alpha - s_1) \cdot \cdots \cdot (\alpha - s_{k'}) }%
             { (\alpha - t_1) \cdot \cdots \cdot (\alpha - t_{l'}) },
    \]
    where \(\{t_1,\ldots,t_{l'}\}\) is disjoint from the set of abscissas.

    In fact, it is no loss of power to restrict to \(k'=k\) and
    \(l'=l\).
\end{description}
Hybrids are of course possible, too.
For example, \expref{Proposition}{prop.rfe-equiv-sv} below uses the evaluations
parametrization for the numerator and roots parametrization for the
denominator.

The following lemma justifies that,
for any polynomial \(p\),
as long as \(\FF\) is large enough,
\(p(\RFE)\) vanishes with respect to a particular parametrization of \(\RFE\)
if and only if
it vanishes with respect to \(\RFE\) as defined in \expref{Definition}{def.rfegen}.
The lemma is an %
immediate consequence %
of the well-known analogous result for polynomials, sometimes referred to as the
Polynomial Identity Lemma   %
\cite{Ore1922,DeMilloLipton1978,Zippel1979,Schwartz1980,arvind19}.

\begin{lemma}
  \label{lem.schwz}
  Let \(\FF\) be field,
  and \(f = g/h\in\FF(\tau_1,\ldots,\tau_l)\) be a rational function in \(l\)
  variables with \(\deg(g) \le d\) and \(\deg(h)\le d\).
  Let \(S\subseteq\FF\) be finite.
  Then the probability that \(f\) vanishes or is undefined when each
  \(\tau_i\) is substituted by a uniformly random element of \(S\)
  is at most \(2d/|S|\).
\end{lemma}
\begin{proof}
The rational function $f$ vanishes or is undefined if and only if the polynomial $p \doteq =gh$ vanishes, which happens with probability at most $\deg(p)/|S|$
according to the Polynomial Identity Lemma.
\end{proof}

In particular,
if \(\FF\) is infinite,
then, for all polynomials \(p\),
all the above parametrizations and \expref{Definition}{def.rfegen}
are equivalent for the purposes of hitting \(p\);
when \(p\) is fixed,
the equivalence holds provided \(|\FF| \ge \poly(n, \deg(p))\).
Quantitative bounds on the number of substitutions to perform
when testing whether \(\RFE\) hits \(p\) in the blackbox algorithm
likewise follow from \expref{Lemma}{lem.schwz}.
As is customary in the context of blackbox derandomization of PIT,
if $\FF$ is not large enough,
then one works instead over a sufficiently large extension of $\FF$.

\medskip

We now formally state and argue the close relationship between $\RFElml$ and $\SVl$ that we sketched in \expref{Section}{s.intro}.
\begin{proposition}
  \label{prop.rfe-equiv-sv}
  Let $l$ and $n$ be positive integers.
  There is an invertible diagonal transformation \(A : \FF^n\to\FF^n\)
  such that, for any polynomial \(p\in\FF[x_1,\ldots,x_n]\),
  \(p(\SVl) = 0\) if and only if \((p\circ A)(\RFElml) = 0\).
\end{proposition}
In particular, the vanishing ideals of $\RFElml$ and of $\SVl$ are the
same up to the rescaling of \expref{Proposition}{prop.rfe-equiv-sv}.

\begin{proof}[Proof of \expref{Proposition}{prop.rfe-equiv-sv}]
  Let \(\FFext\) be the field of rational functions in indeterminates
  \(\upsilon_1\), \ldots, \(\upsilon_l\), \(\zeta_1\), \ldots, \(\zeta_l\)
  over \(\FF\).
  A polynomial \(p\in\FF[x_1,\ldots,x_n]\) has \(p(\SVl)=0\)
  if and only if
  \(p\) vanishes at the point
  \begin{equation}
    \label{eq:prop.rfe-equiv-sv:sv}
    \left(
      \sum_{t=1}^l \zeta_t \prod_{j\in [n]\setminus\{i\}} \frac{\upsilon_t -
      a_{j}}{a_i - a_{j}}
      \;:\; i \in [n]
    \right) \in \FFext^n.
  \end{equation}
  Set \(A : \FF^n \to \FF^n\) to be the diagonal linear transformation
  that divides the coordinate for \(x_i\) by
  \(\prod_{j\in [n]\setminus\{i\}} (a_i - a_{j})\).
  $A$ is invertible.
  Applying \(A^{-1}\) to \eqref{eq:prop.rfe-equiv-sv:sv} yields the point
  \begin{equation}
    \label{eq:thm.rfe-equiv-sv:svnonnorm}
    \left(
      \sum_{t=1}^l
        \zeta_t
        \prod_{j\in [n]\setminus\{i\}}
          (\upsilon_t - a_{j})
      \;:\; i \in [n]
    \right)
    =
    \left(
      \sum_{t=1}^l
        \left(
          \zeta_t
          \prod_{j\in [n]} (\upsilon_t - a_{j})
        \right)
      \frac{1}{\upsilon_t - a_i}
      \;:\; i \in [n]
    \right).
  \end{equation}
  \(p\) vanishes at \eqref{eq:prop.rfe-equiv-sv:sv} if and only if
  \(p\circ A\) vanishes at \eqref{eq:thm.rfe-equiv-sv:svnonnorm}.
  Now let \(\FFext'\) be the field of rational functions in indeterminates
  \(\tau_1\), \ldots, \(\tau_l\), \(\sigma_1\), \ldots,
  \(\sigma_l\) over \(\FF\).
  After the change of variables
  \begin{equation*}
    \zeta_t \gets
      \frac{1}{\prod_{j\in [n]}(\tau_t-a_{j})}
      \cdot
      \frac{-\sigma_t}{\prod_{s \ne t}(\tau_t - \tau_{s})}
      \quad
      \text{and}
      \quad
    \upsilon_t \gets \tau_t
  \end{equation*}
  \eqref{eq:thm.rfe-equiv-sv:svnonnorm} becomes
  \begin{equation}
    \label{eq:prop.rfe-equiv-sv:sv-cov}
    \left(
      \sum_{t=1}^l
        \frac{\sigma_t}{\left(\prod_{s \ne t} \tau_t - \tau_{s}\right)}
        \frac{1}{a_i - \tau_t}
      \;:\; i \in [n]
    \right)
    =
    \left(
      \frac{
        \sum_{t=1}^l
          \sigma_t
          \prod_{s \ne t} \frac{a_i - \tau_{s}}%
                               {\tau_t - \tau_{s}}
      }{
        \prod_{t=1}^l a_i - \tau_t
      }
      \;:\; i \in [n]
    \right)
    \in \FFext'^n.
  \end{equation}
  Since the change of variables is invertible,
  \(p\circ A\) vanishes at \eqref{eq:thm.rfe-equiv-sv:svnonnorm}
  if and only if it vanishes at \eqref{eq:prop.rfe-equiv-sv:sv-cov}.

  Now, viewing \(\sigma_1,\ldots,\sigma_l, \tau_1,\ldots,\tau_l\) as
  seed variables,
  observe that the right-hand side of \eqref{eq:prop.rfe-equiv-sv:sv-cov}
  is \(\RFElml(g/h)\)
  where
  \(g\) is parametrized by evaluations (\(g(\tau_t)=\sigma_t\))
  and
  \(h\) is parametrized by roots (\(\tau_1,\ldots,\tau_l\)).
  It follows that \(p\circ A\) vanishes at
  \eqref{eq:prop.rfe-equiv-sv:sv-cov} if and only if \((p\circ A)(\RFElml)
  = 0\).
\end{proof}

\section{Generating Set}
\label{s.gens}

In this section,
we establish \expref{Theorem}{thm.ideal.generators},
our characterization of the vanishing ideal of $\RFE$ in terms of an explicit
generating set.
For every $k,l \in \N$,
we develop a template, $\VanGenkl$, for constructing polynomials that belong
to the vanishing ideal of $\RFEkl$
such that all instantiations collectively generate the vanishing ideal.

The template can be derived in the following fashion.
Fix any seed \(f\) of \(\RFEkl\),
and write it as \(f = g/h\)
where \(g(\alpha) = \sum_{d=0}^k g_d \alpha^d\) and
\(h(\alpha) = \sum_{d=0}^l h_d \alpha^d\) are respectively polynomials of
degree \(k\) and \(l\).
For any \(i \in [n]\),
the polynomial \(g(a_i)/h(a_i) - x_i \in \FF[x_1,\ldots,x_n]\) vanishes by
definition at \(\RFEkl(f)\).
While this polynomial varies with \(f\),
it does so uniformly.
Specifically,
after rescaling to \(g(a_i) - h(a_i)x_i\),
the polynomial depends only \emph{linearly} on the coefficients of \(g\) and
\(h\).
We exploit this uniformity to construct a polynomial
that vanishes at \(\RFEkl(f)\)
but that now is \emph{independent} of \(f\).
Since \(f\) is arbitrary,
the constructed polynomial belongs to the vanishing ideal of \(\RFEkl\).

The construction begins by expressing the vanishing of each
\(g(a_i) - h(a_i)x_i\) at \(\RFEkl(f)\)
as the following system of equations.
Abbreviating
\begin{align*}
  \vec{g} &\doteq
  \begin{bmatrix}
    g_k & g_{k-1} & \dots & g_1 & g_0
  \end{bmatrix}^\T
  \\
  \vec{h} &\doteq
  \begin{bmatrix}
    h_l & h_{l-1} & \dots & h_1 & h_0
  \end{bmatrix}^\T,
\end{align*}
we write
\begin{equation}\label{eq.vangen.cons}
  \begin{bmatrix}
    a_i^k     & a_i^{k-1}    & \dots & 1 &
    a_i^l x_i & a_i^{l-1}x_i & \dots & x_i
  \end{bmatrix}_{i \in [n]}
  \cdot
  \begin{bmatrix}
    \vec{g} \\
    -\vec{h}
  \end{bmatrix}
  = 0.
\end{equation}
Written this way, \eqref{eq.vangen.cons} has the form of a homogeneous system
of linear equations.
There is one equation for each \(i \in [n]\)
and one unknown for each of the \(k+l+2\) parameters of the seed \(f\).
The system's coefficient matrix has no dependence on \(f\), but 
for any \(f\), substituting \(\RFEkl(f)\) into \(x_1, \ldots, x_n\)
yields a system that has a nontrivial solution,
namely the vector in \eqref{eq.vangen.cons}.

Consider, then,
the determinant of the square subsystem of \eqref{eq.vangen.cons} formed by
any \(k+l+2\) rows.
It is a polynomial in \(\FF[x_1, \ldots, x_n]\).
Because the coefficient matrix in \eqref{eq.vangen.cons} is independent of
\(f\),
the determinant is independent of \(f\).
Because the subsystem has a nonzero solution after substituting \(\RFEkl(f)\)
for any \(f\),
the determinant vanishes after substituting \(\RFEkl(f)\) for any \(f\).
We conclude that the determinant belongs to the vanishing ideal of \(\RFEkl\).

Recalling that the determinant for the subsystem consisting of rows
\(i_1,\ldots,i_{k+l+2}\) is identically
$\VanGenkl[i_1,i_2, \dots, i_{k+l+2}]$,
we have established:
\begin{lemma}
  \label{lemma.vangen-vanishes}
  For every \(k,l \in \N\) and $i_1,i_2,\dots,i_{k+l+2} \in [n]$,
  \(\VanGenkl[i_1, \ldots, i_{k+l+2}] \in \VanIdeal{\RFEkl}\).
\end{lemma}

As we explain at the end of this section, the above derivation is where our use of $\RFE$ in lieu of $\SV$ plays a critical role. Before moving on, we also point a few elementary properties and an provide an explicit expression for the coefficients of $\VanGenkl$ as products of Vandermonde determinants in the abscissas and a sign term. We introduce the following notation for the underlying Vandermonde matrices.
\begin{definition}[Vandermonde matrix]
For $T=\{i_1, \dots, i_t\} \subseteq [n]$  with $i_1 < \dots < i_t$, we abbreviate the Vandermonde matrix built from $a_i$ for $i \in T$ in increasing order as
\begin{equation}
  \label{eq.vandermonde-rep}
  \Amat{T}
  \doteq
  \begin{bmatrix}
    a_{i_1}^{t-1} & \cdots & 1 \\
    \vdots & & \vdots \\
    a_{i_t}^{t-1} & \cdots & 1
  \end{bmatrix}.
\end{equation}
\end{definition}
The sign term makes use of the following standard combinatorial quantity.
\begin{definition}[cross inversions]
For $A,B \subseteq [n]$, $\XInv(A,B) \doteq |\{ (a,b)\in A\times B \mid a>b \}|$ denotes the number of \emph{cross inversions} between $A$ and $B$.
\end{definition}
\begin{proposition}\label{prop.vangen.basic}
  For every \(k,l \in \N\), 
  \(\VanGenkl\) is skew-symmetric in that,
  for any \(i_1,\ldots,i_{k+l+2} \in [n]\) and permutation \(\pi\) of $[k+l+2]$,
  \[
    \VanGenkl[i_1,\ldots, i_{k+l+2}]
    =
    \sgn(\pi)\cdot
    \VanGenkl[i_{\pi(1)},\ldots,i_{\pi(k+l+2)})].
  \]
  For any $S \subseteq [n]$ with $|S|=k+l+2$, $\VanGenkl[S]$ is a nonzero, multilinear, and homogeneous polynomial of total degree \(l+1\),
  and every multilinear monomial of degree $l+1$ in
  \(x_{i_1},\ldots,x_{i_{k+l+2}}\) appears with a nonzero coefficient. More specifically, for $S = \{i_1, \dots, i_{k+l+2}\} \subseteq [n]$ with $i_1 < i_2 < \dots, i_{k+l+2}$,
  \begin{equation}
    \label{eq.EVC.expanded}
      \VanGenkl[S] = \sum_{\substack{K\sqcup L= S \\ |L|=l+1}}
      \gamma_{K,L} \cdot \prod_{i \in L} x_i,
  \end{equation}
  where
  \begin{equation} \label{eq-EVC-coef}
    \gamma_{K,L} \doteq (-1)^{\XInv(K,L)} \cdot \Adet{K} \cdot \Adet{L}.
\end{equation}    
\end{proposition}
\begin{proof}
  All assertions follow from
  elementary properties of determinants, that Vandermonde
  determinants are nonzero unless they have duplicate rows,
  and the following computation. The coefficient $\gamma_{K,L}$ can be obtained by plugging in 0 for $x_i$ with $i \in K$, and 1 for $x_i$ with $i \in L$. For $K^*$ consisting of the first $k+1$ elements of $S$ and $L^*$ of the last $l+1$, this yields the determinant
  \begin{equation}\label{eq.det.ordered}
    \begin{vmatrix*}
      a_{i_1}^k     & \cdots & 1      && 0      & \cdots & 0 \\
      \vdots        & \ddots & \vdots && \vdots & \ddots & \vdots \\
      a_{i_{k+1}}^k & \cdots & 1      && 0      & \cdots & 0 \\
      \\
      *      & \cdots & *      && a_{i_{k+2}}^l   & \cdots & 1 \\
      \vdots & \ddots & \vdots && \vdots          & \ddots & \vdots \\
      *      & \cdots & *      && a_{i_{k+l+2}}^l & \cdots & 1
    \end{vmatrix*},
  \end{equation}
  which equals the product of the Vandermonde matrices $\Adet{K^*}$ and $\Adet{L^*}$, and confirms the expression for $\gamma_{K^*,L^*}$ as $\XInv(K^*,L^*)=0$. For general $K$ and $L$, we obtain a determinant with the same shape as \eqref{eq.det.ordered} after rearranging the rows such that the ones involving $a_i$ for $i \in K$ appear first and in order, and the ones involving $a_i$ for $i \in L$ appear last and in order. By skew-symmetry, the rearrangement induces an additional factor of $\sgn(\pi) = (-1)^{\Inv(\pi)}$, where $\pi$ denotes the underlying permutation of $[k+l+2]$ and $\Inv(\pi)$ denotes the number of inversions of $\pi$, which equals $\XInv(K,L)$. 
\end{proof}

\expref{Lemma}{lemma.vangen-vanishes} shows that the polynomials $\VanGenkl[i_1,\dots,i_{k+l+2}]$ belong to the vanishing ideal of $\RFEkl$. In fact, various subsets of them \emph{generate} the vanishing ideal. To prove that a certain subset does so, we establish the following two steps:
\begin{enumerate}
  \item
    Modulo the ideal $I$ generated by the subset, 
    every polynomial $p$ is equal
    to a polynomial $r$ with a particular combinatorial structure
    (\expref{Lemma}{lem.vangen-coredform}).
  \item
    Every nonzero polynomial $r$ with that structure is hit by
    $\RFEkl$ (\expref{Lemma}{lem.rfe-hits-cored}).
\end{enumerate}
Together,
these show that every polynomial in the vanishing ideal of $\RFEkl$ is equal to the zero polynomial modulo the ideal $I$. 
We conclude that the ideals coincide,
\ie, the vanishing ideal of $\RFEkl$ is generated by the subset of instantiations of \(\VanGenkl\) that define $I$.

For the subset of instantiations $\VanGenkl[C \sqcup L]$ where $C \subseteq [n]$ is any fixed subset of size $k+1$ and $L \subseteq [n]$ ranges over all subsets of size $l+1$ disjoint from $C$, the combinatorial structure bridging the two steps is that the polynomial is cored.

\begin{definition}[monomial support and cored polynomial]
  \label{def.sunform}
  The \emph{support} of a monomial $m \in \FF[x_1,\dots,x_n]$, denoted $\supp(m)$, is the set of indices $i \in [n]$ such that $m$ depends on $x_i$. 
  For \(c,t\in\N\),
  a polynomial \(p \in \FF[x_1,\dots,x_n]\) is said to be \emph{\((c,t)\)-cored} if there exists $C \subseteq [n]$, called the \emph{core}, such that $|C| \le c$ and for every monomial $m$ of \(p\), $|\supp(m) \setminus C| \le t$.
  For any subset $C \subseteq [n]$ and monomial $m = \prod_{i \in [n]} x_i^{d_i}$, we call $\prod_{i \in C} x_i^{d_i}$ the $C$-part of $m$, and $\prod_{i \in [n] \setminus C} x_i^{d_i}$ the non-$C$-part of $m$.
\end{definition}

The crux for the first step is the following property, which allows us to gradually get closer to a $(k+1,l)$-cored polynomial. 

  \begin{proposition}
    \label{prop.lem.vangen-coredform}
    Let \(k,l,n\in\N\), let \(C\) be a \((k+1)\)-subset of \([n]\), and let $I$ denote the ideal generated by the polynomials $\VanGenkl[C \sqcup L]$ where $L$ ranges over all $(l+1)$-subsets of $[n] \setminus C$. Consider a monomial $m \in \FF[x_1,\dots,x_n]$ such that $|\supp(m) \setminus C| > l$. Modulo $I$, $m$ is equal to a linear combination of monomials whose non-$C$-parts have lower degree than the non-$C$-part of $m$.
  \end{proposition}
\begin{proof}
  Let $L$ be a subset of $\supp(m) \setminus C$ of size $l+1$. Let $m'$ be the monomial such that $m = m' \cdot x^L$, where $x^L \doteq \prod_{i \in L} x_i$.
  By \expref{Proposition}{prop.vangen.basic},
  \(x^L\) is a monomial of $\VanGenkl[C \sqcup L]$,
  and every other monomial of $\VanGenkl[C \sqcup L]$ has non-$C$-part of degree at most $l$. 
  It follows that $m' \cdot \VanGenkl[C \sqcup L]$ can be written as $c \cdot m + r$, where $c$ in a nonzero element in $\FF$ and every monomial in $r$ has non-$C$-part of lower degree than $m$ does. Since ideals are closed under multiplication by any other polynomial, $m' \cdot \VanGenkl[C \sqcup L] \in I$. Thus, we have $0 \equiv c \cdot m +r \bmod I$, which can be rewritten as $m \equiv - c^{-1} \cdot r \bmod I$.
\end{proof}

\expref{Proposition}{prop.lem.vangen-coredform} leads to the following formalization of the first step of our approach.

\begin{lemma}
  \label{lem.vangen-coredform}
  Let \(k,l,n\in\N\),
  let \(C\) be a \((k+1)\)-subset of \([n]\), and
  let $I$ be the ideal generated by the
  polynomials \(\VanGenkl[C \sqcup L]\)
  where \(L\) ranges over all \((l+1)\)-subsets of \([n] \setminus C\).
  Modulo $I$,
  every polynomial is equal
  to a \((k+1,l)\)-cored polynomial with core \(C\).
\end{lemma}

\begin{proof}
For any polynomial $p \in \FF[x_1,\dots,x_n]$, \expref{Proposition}{prop.lem.vangen-coredform} allows us to 
systematically eliminate any monomial $m$ in \(p\) that violates the \((k+1,l)\)-cored condition, without changing $p$ modulo $I$. The process may introduce other monomials, but those monomials all have non-$C$-parts of degree lower than $m$ does. This means that the process cannot continue indefinitely. When it ends, the remaining polynomial is \((k+1,l)\)-cored with core \(C\) and is equivalent to $p$ modulo $I$. 
\end{proof}

The second step of our approach is formalized in \expref{Lemma}{lem.rfe-hits-cored}.

\begin{lemma}
  \label{lem.rfe-hits-cored}
  Let \(k,l,n\in\N\) and let $r \in \FF[x_1,\dots,x_n]$ be nonzero and \((k+1,l)\)-cored.
  Then \(\RFEkl\) hits \(r\).
\end{lemma}

We prove \expref{Lemma}{lem.rfe-hits-cored} from the \hyperref[lem.multi-zsub-deriv]{Zoom Lemma} in \expref{Section}{s.zoom}. Assuming it, we have all ingredients for the proof of \expref{Theorem}{thm.ideal.generators}. 

\begin{proof}[Proof of \expref{Theorem}{thm.ideal.generators}]
The combination of \expref{Lemma}{lemma.vangen-vanishes},
\expref{Lemma}{lem.vangen-coredform}, and \expref{Lemma}{lem.rfe-hits-cored} shows that,
for every core $C \subseteq [n]$ of size $k+1$, the vanishing ideal \(\VanIdeal\RFEkl\) is generated by the
polynomials \(\VanGenkl[C \sqcup L]\) where \(L\) ranges over all \((l+1)\)-subsets of \([n] \setminus C\). The generators are all homogeneous of minimum degree \(l+1\), and each generator has a monomial that occurs in none of the other generators (namely the product of the variables in $L$). Therefore, the generating set has minimum size
since it forms a vector space basis of the degree-\((l+1)\) part of \(\VanIdeal\RFEkl\).
\end{proof}

As an aside, we justify along the same lines the claim from  \expref{Section}{s.intro} that all linear dependencies among instances of $\VanGenkl$ are generated by the equations \eqref{eq.generator.dependencies} when $\{i_1, \dots, i_{k+l+3}\}$ ranges over all subsets of $[n]$ containing a core $C$ of size $k+1$. A similar reduction strategy modulo those equations allows us to rewrite 
\[
\sum_{\substack{S \subseteq [n] \\ |S|=k+l+2}} c_S \cdot  \VanGenkl[S] = 0
\]
such that the range of the subsets $S$ is reduced to a $(k+1,l+1)$-cored subclass with core $C$. By linear independence, the only equation of that form is the trivial one with all $c_S=0$.

\paragraph{Gr\"obner basis.} We end this section with a short discussion on Gr\"obner bases. This part is not essential for understanding the remainder of the paper; the reader may feel free to skip it. Readers who want to know more may refer to \cite{AL1994,CLO2013}.

Gr\"obner bases are useful for solving several computational problems involving ideals, including determining whether a given polynomial $p$ belongs to an ideal $I$ given by a finite set $G$ of generators. The setup presumes a total order $\ge$ on monomials with the following properties:
\begin{itemize}
    \item For all monomials $m$, we have $m \ge 1$, where $1$ denotes the empty monomial.
    \item For monomials $m_1,m_2,m$, we have that if $m_1\ge m_2$, then $m_1\cdot m \ge m_2\cdot m$.
\end{itemize}
Assuming such a monomial ordering $\ge$, every nonzero polynomial has a unique monomial that is maximal in $\ge$, which we call the leading monomial. 

For a given a polynomial $p$, we can compute a $G$-reduced form of $p$ by repeatedly applying the following reduction step, starting from $f=p$: Find $g\in G$ such that the leading monomial of $g$ divides some monomial $m$ of $f$, and then subtract a suitable multiple of $g$ from $f$ so as to produce a new value of $f$ that does not contain $m$ as a monomial. If several such $g$ and $m$ exist, pick any. The process continues until no suitable $g$ and $m$ can be found, which the properties of the ordering $\ge$ guarantee to happen at some point. The final $f$ is called a $G$-reduced form of $p$, which may or may not be unique.

A natural algorithm to determine membership of $p$ in $I$ is to compute a $G$-reduced form $f$ of $p$ and conclude that $p \in I$ if and only if $f=0$. The algorithm generalizes computing the remainder in univariate polynomial division, with some differences being that there are multiple choices of $g$, and the monomial $m$ does not need to be the leading monomial of $f$. A positive conclusion, $p \in I$, is always correct because reduction does not affect membership and $0$ trivially belongs to the ideal. However, the algorithm can have false negatives, namely when $p \in I$ has a nonzero $G$-reduced form, \ie, the reduction process reaches some $f \in I$ that does not have a monomial $m$ divisible by the leading monomial of some $g \in G$. 

A \emph{Gr\"obner basis} $G$ for $I$ is a finite set of generators satisfying the additional constraint that every nonzero element of $I$ has a monomial $m$ that is divisible by the leading monomial of some element of $G$. In this case, the above algorithm for deciding membership in $I$ is always correct. In fact, this gives another characterization of when a finite generating set $G$ is a Gr\"obner basis. Yet another characterization is that every polynomial $p$ has a unique $G$-reduced form $f$.

In the overview, we claimed that the set $G$ of polynomials $\VanGenkl[C \sqcup L]$ form a Gröbner basis for $\VanIdeal{\RFEkl}$, where $C \subseteq [n]$ is a fixed core of size $k+1$ and $L$ ranges over the $(l+1)$-subsets of $[n] \setminus C$. This holds with respect to any monomial ordering such that, for every $L$, $x^L \doteq \prod_{i \in L} x_i$ is the leading monomial of $\VanGenkl[C \sqcup L]$. Examples of such orderings include all lexicographic orderings where the variables outside $C$ have higher priority than the variables inside $C$. \expref{Lemma}{lem.rfe-hits-cored} implies that every nonzero polynomial in $\VanIdeal{\RFEkl}$ has a monomial with more than $l$ variables outside of $C$, which is to say that the monomial is divisible by $x^L$ for some $L \subseteq [n] \setminus C$ of size $l+1$. As $x^L$ is the leading monomial of $\VanGenkl[C \sqcup L]$, we conclude that every nonzero polynomial in $\VanIdeal{\RFEkl}$ has a monomial that is divisible by the leading term of some element of $G$, \ie, $G$ is a Gröbner basis. 

One can interpret \expref{Proposition}{prop.lem.vangen-coredform} as performing a reduction step of $p$ by $G$. \expref{Lemma}{lem.vangen-coredform} keeps performing this step until it is no longer possible, yielding a $G$-reduced form $f$ of $p$ that is $(k+1,l)$-cored with core $C$. \expref{Lemma}{lem.rfe-hits-cored} implies that any two $(k+1,l)$-cored representatives modulo $I$ of the same polynomial $p$ coincide, so every polynomial $p$ has a unique $G$-reduced form. This is another way to see that the set $G$ is a Gr\"obner basis. 

\paragraph{Instantiation for SV.}
By the connection between $\SV$ and $\RFE$, a generating set for $\VanIdeal{\RFElml}$ induces a generating set for $\VanIdeal{\SVl}$. We provide an explicit expression as an instantiation of \expref{Theorem}{thm.ideal.generators} and \expref{Proposition}{prop.vangen.basic}. 

\begin{corollary}\label{cor-gen-sv}
 Let $l, n \in \N$ and let $a_i$ for $i \in [n]$ be distinct elements of $\FF$.
 For any fixed set $C \subseteq [n]$ of size $l$, the polynomials $\VanGenSV{l}[C \sqcup L]$ form a generating set
 of minimum size for $\VanIdeal{\SVl}$ when $L$ ranges over all $(l+1)$-subsets of $[n]$ that are disjoint from $C$. Here, for any $S = \{i_1, \dots, i_{2l+1}\} \subseteq [n]$ with $i_1 < \dots < i_{2l+1}$,
 \[
 \VanGenSV{l}[S] \doteq \sum_{\substack{T \subseteq S \\ |T|=l+1}} \gamma'_{S \setminus T,T} \cdot \prod_{i \in
 T} x_i,
 \]
 where
 \begin{equation}\label{eq.SVgen.coef}
 \gamma'_{S \setminus T,T} \doteq (-1)^{\XInv(S \setminus T,T)} \cdot \left(\prod_{i \in T}\prod_{j \in [n]\setminus \{i\}} (a_i-a_{j})\right) \cdot 
 \Adet{S\setminus T} \cdot \Adet{T}.
\end{equation}
\end{corollary}

\begin{proof}
By \expref{Proposition}{prop.rfe-equiv-sv}, for any polynomial \(p\in\FF[x_1,\ldots,x_n]\), \(p(\SVl) = 0\) iff \((p\circ A)(\RFElml) = 0\), where $A:\FF^n\to \FF^n$ is the invertible transformation that divides each variable $x_i$ by $\prod_{j\in [n]\setminus \{i\}} (a_i-a_j)$. Since the vanishing ideal of $\RFElml$ coincides with the ideal generated by the polynomials $\VanGenparam{l-1}{l}[C \sqcup L]$, it follows that the vanishing ideal of $\SVl$ coincides with the ideal generated by the polynomials 
\begin{equation}\label{eq.gen.transfo}
\VanGenparam{l-1}{l}[C \sqcup L] \circ A^{-1}.
\end{equation} 
For $T \subseteq S \doteq C \sqcup L$ with $|T|=l+1$, the coefficient of $\prod_{i \in T} x_i$ in $\VanGenparam{l-1}{l}[S]$ is given by $\gamma_{S \setminus T, T}$. By the construction of $A$, $A^{-1}$ multiplies $x_i$ by $\prod_{j\in [n]\setminus \{i\}} (a_i-a_j)$. Thus, the coefficient of $\prod_{i \in T} x_i$ in \eqref{eq.gen.transfo} equals $\gamma'_{S\setminus T, T}$ given by \eqref{eq.SVgen.coef} and 
$\VanGenparam{l-1}{l}[C \sqcup L] \circ A^{-1} = \VanGenSV{l}[S]$.
\end{proof}

In comparison to the coefficients $\gamma_{S \setminus T, T}$ of the polynomial $\VanGenparam{l-1}{l}[S]$, the coefficients $\gamma'_{S \setminus T, T}$ of $\VanGenSV{l}[S]$ contain an additional term, namely the middle term on the right-hand side of \eqref{eq.SVgen.coef}. As a consequence, each coefficient of $\VanGenSV{l}[S]$ depends on \emph{all} abscissas $a_1,\dots,a_n$, whereas the coefficients of $\VanGenkl[S]$ only depend on the abscissas with indices in $S$. This reflects a difference in setup between the two generators: The substitution for a variable $x_i$ is a multivariate function of all abscissas in $\SV$ versus a univariate function of the abscissa $a_i$ only in $\RFE$. The difference represents one reason why $\RFE$ is more convenient to work with than $\SV$, even though both have essentially the same power.

A more important reason is our derivation of the generating set $\VanGenkl$ in \expref{Lemma}{lemma.vangen-vanishes}. Our approach hinges on the fact that the substitutions for a variable $x_i$ induce linear equations involving the seed variables $g_{k},\dots,g_{0},h_{l},\dots,h_{0}$, with coefficients being expressions in terms of the polynomial variables $x_1,\dots,x_n$ and abscissas $a_1,\dots,a_n$. Collecting $k+l+2$ of such equations yields as many linear constraints as unknowns, which suffices to derive a nontrivial element of the vanishing ideal. The substitutions \eqref{eq-sv-sub} for $x_i$ made by $\SVl$ similarly induce linear equations, though not between the mere seed variables $y_1, z_1, \dots, y_l, z_l$ but between monomials in the seed variables, namely the constant monomial and the monomials $z_t y_t^d$ for $t\in [l]$ and $d\in \{0,\dots,n-1\}$. In contrast to the setting of $\RFE$, even if we collect all of those equations, namely $n$ linear equations in $nl+1$ unknowns, this does not give us enough information to derive a nontrivial element of the vanishing ideal.

\section{Zoom Lemma}
\label{s.zoom}

Throughout the paper we make repeated use of a key technical tool,
the \hyperref[lem.multi-zsub-deriv]{Zoom Lemma}.
The lemma allows us to zoom in on the contributions of the monomials in a
polynomial $p$ that have prescribed degrees in a subset of the variables.
We introduce the following terminology for prescribing degrees.

\begin{definition}[degree pattern]\label{def.degree-pattern}
  Let \(J \subseteq [n]\).
  A \emph{degree pattern} with domain \(J\) is a \(J\)-indexed
  tuple \(d \in \N^J\) of nonnegative integers.
  A degree pattern \(d\) \emph{matches} a monomial
  $m \in \FF[x_1,\dots,x_n]$
  if,
  for every \(j \in J\),
  \(m\) has degree exactly $d_j$ in $x_j$.
  We say that \(d\) is \emph{in} \(p\) if \(d\) matches some monomial in \(p\).

  For any fixed \(J\),
  every polynomial \(p \in \FF[x_1,\ldots,x_n]\) can be written uniquely in
  the form
  \begin{equation*}
    p = \sum_{d \in \N^J}  p_d \cdot x^d
  \end{equation*}
  where \(x^{\smash{d}} \doteq \prod_{j\in J} x_j^{\smash{d_j}}\)
  and \(p_d\) depends only on variables not indexed by \(J\).
  We refer to \(p_d\) as the coefficient of \(x^d\) in \(p\).
\end{definition}

The notation $p_d$ can be viewed as a generalization of the common one for the
coefficient of degree $d$ of a univariate polynomial $p$.

Our technique allows us to zoom in on the contributions of the coefficients
$p_d$ of degree patterns $d$ that satisfy the
following additional constraint.
\begin{definition}[extremal degree pattern]\label{def.extremal}
  Let \(K, L \subseteq [n]\).
  A degree pattern $d^* \in \N^{K \cup L}$ is \emph{$(K,L)$-extremal} in
  a polynomial $p \in \FF[x_1,\dots,x_n]$
  if $d^*$ is the unique degree pattern $d \in \N^{K \cup L}$ in $p$ that satisfies both
  \begin{itemize}
    \item[(i)] $d_j \le d^*_j$ for all $j \in K$, and
    \item[(ii)] $d_j \ge d^*_j$ for all $j \in L$.
  \end{itemize}
\end{definition}
The notion of extremality in \expref{Definition}{def.extremal} is closely related to
standard notions of minimality and maximality of tuples of numbers.
A \(J\)-tuple $d^*$  is minimal in a set $D$ of such tuples
if the only tuple $d \in D$
that satisfies $d_j \le d^*_j$ for all $j \in J$,
is $d^*$ itself.
A maximal tuple is defined similarly by replacing $\le$ by $\ge$.
Minimality is equivalently \((J,\varnothing)\)-extremality,
and maximality is equivalently \((\varnothing,J)\)-extremality.

When \(K\) and \(L\) intersect,
note that only degree patterns \(d \in \N^{K\cup L}\)
with \(d_j = d^*_j\) for all \(j \in K \cap L\)
affect whether \(d^*\) is \((K,L)\)-extremal.

The above terminology lets us state our key technical lemma succinctly.
\begin{lemma}[Zoom Lemma]\label[zoom-lemma]{lem.multi-zsub-deriv}
  Let \(K,L \subseteq [n]\),
  let \(p \in \FF[x_1,\ldots,x_n]\),
  and let $d^* \in \N^{K\cup L}$ be a degree pattern that is $(K,L)$-extremal
  in $p$.
  If the coefficient $p_{d^*}$ is nonzero upon the substitution
  \begin{equation}
    \label{eq.multi-zsub-deriv.sub}
    x_i \gets
    z \cdot
      \frac{
        \prod_{j \in K\setminus L} ( a_i - a_j )
      }{
        \prod_{j \in L\setminus K} ( a_i - a_j )
      }
      \qquad \forall i\in [n]\setminus (K\cup L)
  \end{equation}
  where $z$ is a fresh variable, 
  then \(\RFEkl\) hits \(p\) for any \(k\ge|K|\) and \(l\ge|L|\).
\end{lemma}

Note that the result of substituting \eqref{eq.multi-zsub-deriv.sub} into $p_{d^*}$ is a univariate polynomial $q$ in $z$. In the case where $p$ is homogeneous, $q$ has a single monomial, so $q$ is nonzero iff $q$ is nonzero at $z=1$. In general, it suffices for $q$ to be nonzero at some point $z \in \FF$. As for the conclusion, the most interesting settings in \expref{Lemma}{lem.multi-zsub-deriv} are $k=|K|$ and $l=|L|$. This is because the range of $\RFEkl$ is contained in the range of $\RFEparam{k'}{l'}$ for $k'\ge k$ and $l'\ge l$. Also, whereas many or our instantiations of the \hyperref[lem.multi-zsub-deriv]{Zoom Lemma} have \(K\) and \(L\) disjoint,
this is not necessary for the lemma to hold.%
\footnote{%
  In fact, allowing \(K\) and \(L\) to overlap is useful in
  \expref{Section}{s.test} (see \expref{Proposition}{prop.zsub-deriv-power-variant}) and 
  \expref{Section}{s.constant-width-roabp} (see \expref{Proposition}{prop.application.roabp.coef-eqns}).%
}

Let us first see how the \hyperref[lem.multi-zsub-deriv]{Zoom Lemma} allows us to complete the proof of
\expref{Theorem}{thm.ideal.generators}. There are several ways to do so; we present a fairly generic way.

\begin{proof}[Proof of \expref{Lemma}{lem.rfe-hits-cored} from the
  \texorpdfstring{\hyperref[lem.multi-zsub-deriv]{Zoom Lemma}}{Zoom Lemma}]
  Let \(C \subseteq [n]\) denote a core of size at most $k+1$ for \(r\).
  We construct subsets $K, L \subseteq [n]$ with $|K| \le k$ and $|L| \le l$,
  and a degree pattern $d^*$ with domain $K \cup L$ that is $(K,L)$-extremal in
  $r$ such that $r_{d^*}$ is nonzero upon the substitution 
  \eqref{eq.multi-zsub-deriv.sub}.
  The \hyperref[lem.multi-zsub-deriv]{Zoom Lemma} then implies that $\RFEkl$ hits $r$.

  The construction consists of two steps.
  First, we pick $i^*\in C$ arbitrarily. (We can assume without loss of
  generality that $C$ is nonempty because if $C$ is a core, then so is $C$ with an additional element.) 
  We also set $K \doteq C \setminus \{i^*\}$,
  and let $d_+$ be a degree pattern with domain $K$ that matches a monomial in
  $r$ and that is minimal among all such degree patterns.
  The existence of $d_+$ follows from the fact that $r$ is nonzero.
  By construction, $|K| \le (k+1)-1 = k$ and \(r_{d_+}\) is nonzero.

  Second, we pick a degree pattern $d_-$ with domain $[n]\setminus C$ that
  matches a monomial in $r_{d_+}$ and that is maximal among all such degree
  patterns.
  The existence of $d_-$ follows from the fact that $r_{d_+}$ is nonzero.
  Let $L$ denote the set of indices $j \in [n]\setminus C$ on which $d_-$ is
  positive.
  The hypothesis that $C$ is a $(k+1,l)$-core for $r$ implies that
  $|L| \le l$.
  By construction, the restriction of $d_-$ to the domain $L$ is
  maximal among the degree patterns with domain $L$ in $r_{d_+}$.

  Note that $K$ and $L$ are disjoint,
  because $K \subseteq C$ and $L \subseteq [n]\setminus C$.
  We define $d^*$ as the degree pattern with domain $K \sqcup L$ that agrees
  with $d_+$ on $K$ and with $d_-$ on $L$.
  The minimality and maximality properties of $d_+$ and $d_-$ imply that $d^*$
  is $(K,L)$-extremal in $r$.
  As there is at least one monomial in $r$ that agrees with the degree pattern $d^*$, the coefficient $r_{d^*}$ is nonzero. Since $K$ includes all of $C$ but $i^*$, $r_{d^*}$ cannot depend on variables indexed by $C$ other than $x_{i^*}$. By the maximality of $d_-$ on $[n] \setminus C$ and the fact that $L$ contains all indices in $[n] \setminus C$ on which $d_-$ is positive, $r_{d^*}$ cannot depend on any variable in $[n] \setminus C$. Thus, $r_{d^*}$ is a nonzero polynomial that depends only on
  $x_{i^*}$. It follows that
  substituting \eqref{eq.multi-zsub-deriv.sub} into $r_{d^*}$ yields a nonzero polynomial in $z$.
\end{proof}

Before giving a formal proof of the \hyperref[lem.multi-zsub-deriv]{Zoom Lemma},
we provide some intuition for the mechanism behind it,
and we explain how the choice of the substitution
\eqref{eq.multi-zsub-deriv.sub} and the extremality requirement arise. We consider $k=|K|$ and $l=|L|$, and focus on the setting of homogeneous polynomials $p$, in which case we can set $z=1$ without loss of generality. 

We start with the special case where (i) $\ell=0$,
or equivalently $L=\emptyset$,
and (ii) the degree pattern $d^* \in \N^K$ is zero in every coordinate,
so \(x^{d^*}\) is the constant monomial 1.
We can zoom in on $p_{d^*}$ by setting all variables $x_j$ for $j \in K$ to
zero.
The generator $\RFEparam{k}{0}$ allows us to do so by picking a seed $f$ such
that $f(a_j)=0$ for all $j \in K$, namely
\begin{equation}\label{eq.zoom.seed.special}
  f(\alpha) \doteq \prod_{j\in K}(\alpha - a_j).
\end{equation}
The evaluation of $p$ at $\RFE(f)$ coincides with the evaluation of $p_{d^*}$
at $\RFE(f)$,
which is precisely \eqref{eq.multi-zsub-deriv.sub} with $z=1$.
If the evaluation is nonzero,
then evidently $\RFEparam{k}{0}$ hits $p$, as desired.

In order to handle more general degree patterns $d^* \in \N^K$, 
we introduce a fresh parameter $\xi_j$ for each $j \in K$,
and replace $a_j$ in \eqref{eq.zoom.seed.special} by $a_j - \xi_j$,
\ie, we consider the seeds
\begin{equation}\label{eq.zoom.seed.K}
  \hat{f}(\alpha) \doteq \prod_{j \in K} (\alpha - a_j + \xi_j),
\end{equation}
where the hat indicates a dependency on the fresh parameters.
For each \(i\),
\(\hat{f}(a_i)\) is a multivariate polynomial in $\xi_j$, $j \in K$,
and \(\RFE(\hat{f})\) applies the substitution \(x_i \gets \hat{f}(a_i)\)
for each \(i \in [n]\).
The critical property is that $\hat{f}(a_i)$ contains the factor \(\xi_i\) for $i \in K$ but not for $i \not\in K$.
More precisely,
\begin{equation*}
  \hat{f}(a_i) = \begin{cases}
    \hat{c}_i\cdot \xi_i & i \in K \\
    \hat{c}_i & i \not\in K
  \end{cases},
\end{equation*}
where
$\hat{c}_i \doteq \prod_{j \in K\setminus\{i\}} (a_i - a_j + \xi_j)$
is a multivariate polynomial in the parameters $\xi$ with nonzero constant
term,
namely $c_i \doteq \prod_{j \in K\setminus\{i\}} (a_i - a_j)$.

For any monomial \(m\) with matching degree pattern \(d \in \N^K\),
we have
\[
  m(\RFE(\hat{f}))
  =
  m(\hat{c}) \cdot \xi^{d}
  =
  m_{d}(\hat{c}) \cdot \hat{c}^{d} \cdot \xi^{d}.
\]
Here we see that,
when \(m(\RFE(\hat{f}))\) is expanded as a linear combination of monomials in
the \(\xi_j\),
the combination contains only monomials divisible by \(\xi^d\) and the coefficient of \(\xi^d\) is nonzero (namely $c^{d}$). 

In the expansion of \(p(\RFE(\hat{f}))\),
the coefficient of \(\xi^{d^*}\)
\begin{itemize}
  \item[(a)] has a contribution $m(c) = m_{d^*}(c) \cdot c^{d^*}$ from each
    monomial $m$ in $p$ that matches \(d^*\), and
  \item[(b)] may have contributions from other monomials $m$ in $p$
    but only from those whose degree pattern on $K$ is smaller than $d^*$,
    \ie, only if $\deg_j(m) \le d^*_j$ for all $j \in K$.
\end{itemize}
By adding the contributions of all monomials $m$ with degree pattern $d^*$ we obtain
\[ p_{d^*}(\RFE(\hat{f})) \cdot \RFE(\hat{f})^{d^*} = p_{d^*}(\hat{c}) \cdot \hat{c}^{d^*} \cdot \xi^{d^*}.\]
By properties (a) and (b) above,
we conclude that the coefficient of the monomial $\xi^{d^*}$ in
$p(\RFE(\hat{f}))$:
\begin{itemize}
  \item[(a')] has a contribution of $p_{d^*}(c) \cdot c^{d^*}$ from the
    monomials matching $d^*$, and
  \item[(b')] cannot have any additional contributions provided that there are
    no degree patterns on $K$ in $p$ that are smaller than $d^*$.
\end{itemize}
For a degree pattern $d^*$ in $p$,
condition (b') can be formulated as the minimality of $d^*$ among the degree
patterns on $K$ in $p$,
which is exactly the requirement that $d^*$ is $(K,L)$-extremal in $p$
for $L = \emptyset$.
Under this condition we conclude that the coefficient of the monomial
$\xi^{d^*}$ in $p(\RFE(\hat{f}))$ equals $p_{d^*}(c) \cdot c^{d^*}$.
Note that $c^{d^*}$ is nonzero.
Since $p_{d^*}(c)$ only depends on the components $c_i$
for $i \in [n] \setminus K$,
and those components agree with \eqref{eq.multi-zsub-deriv.sub} for $z=1$,
the coefficient of the monomial $\xi^{d^*}$ in $p(\RFE(\hat{f}))$ is nonzero
if and only if $p_{d^*}$ is nonzero at the point
\eqref{eq.multi-zsub-deriv.sub} with $z=1$.
Thus, for a homogeneous polynomial $p$, the hypotheses of the lemma imply that
$p(\RFE(\hat{f}))$ is a nonzero polynomial in the parameters $\xi$.
It follows that a random setting of the parameters $\xi$ yields a seed $f'$
for $\RFEparam{k}{0}$ such that $p(\RFE(f'))$ is nonzero.
This shows that $\RFEparam{k}{0}$ hits $p$.

The symmetric case $k=0$ can be obtained from the case $l=0$ by transforming
$x_i \mapsto x_i^{-1}$ for each $i \in [n]$.
The transformation maps a seed $f$ for $\RFEparam{0}{l}$ into a seed
$\tilde{f}$ for $\RFEparam{l}{0}$,
wherein the zeroes of \(\tilde{f}\) come from the poles of $f$.
Given a polynomial $p(x_1,\dots,x_n)$,
we similarly transform the variables and clear denominators to obtain the
polynomial
$\tilde{p}(x_1,\dots,x_n) \doteq p(x_1^{-1},\dots,x_n^{-1}) \cdot x^{g}$,
where $g$ is any degree pattern with domain $[n]$ for which
$g_i$ is at least the degree of \(x_i\) in \(p\) for every $i \in [n]$.
We apply the previous case of the \hyperref[lem.multi-zsub-deriv]{Zoom Lemma} to $\tilde{p}$ and obtain the new
case of the \hyperref[lem.multi-zsub-deriv]{Zoom Lemma} for $p$.
Note that a monomial with degree pattern $\tilde{d}$ in $\tilde{p}$
corresponds to a monomial with degree pattern $d = g - \tilde{d}$ in
$p$.
It follows that $\widetilde{d^*}$ is minimal in $\tilde{p}$ iff $d^*$ is
maximal in $p$,
which is exactly the $(K,L)$-extremality requirement of
the \hyperref[lem.multi-zsub-deriv]{Zoom Lemma} in the
case where $K=\emptyset$.

The above arguments for the special cases $l=0$ and $k=0$ carry through for arbitrary polynomials $p$ under the assumption that $p_{d^*}$ is nonzero upon the substitution \eqref{eq.multi-zsub-deriv.sub} with $z=1$, \ie, that the univariate polynomial $q(z)$ obtained by substituting \eqref{eq.multi-zsub-deriv.sub} in $p_{d^*}$ is nonzero at $z=1$. The homogeneity of $p$ was only used to conclude that if $q$ is nonzero, then $q$ is nonzero at $z=1$. To handle polynomials $p$ where $q$ may be nonzero but zero at $z=1$, we run the above argument with an arbitrary value of $z \in \FF$ where $q$ is nonzero. We can do so by including an additional factor of $z$ on the right-hand sides of \eqref{eq.zoom.seed.special} and \eqref{eq.zoom.seed.K}, \ie, by considering
$f(\alpha) \doteq z \cdot \prod_{j\in K}(\alpha - a_j)$ and 
$\hat{f}(\alpha) \doteq z \cdot \prod_{j \in K} (\alpha - a_j + \xi_j)$, respectively. Both expressions correspond to valid seeds for $\RFEparam{k}{0}$ in the roots parametrization. 

The case for general $k$ and $l$ follows in a similar fashion, introducing parameters for the zeroes as well as the poles of the seed $f$, considering the monomial in those parameters with degree pattern determined by \(d^*\), and clearing denominators. 

\begin{proof}[Proof of \texorpdfstring{\hyperref[lem.multi-zsub-deriv]{Zoom Lemma}}{Zoom Lemma}]
  Let \(K\), \(L\), \(p\), and \(d^*\) be as in the lemma statement. Fix $z$ to a value in $\FF$ such that $p_{d^*}$ is nonzero upon the substitution \eqref{eq.multi-zsub-deriv.sub}. Such a value exists by the hypothesis of the lemma (for large enough $\FF$). Since the range of $\RFEkl$ is contained in the range of $\RFEparam{k'}{l'}$ for $k'\ge k$ and $l'\ge l$, 
  it suffices to show that $\RFEkl$ hits $p$ for \(k = |K|\) and \(l = |L|\).
  Let \(\xi_j\) for each \(j \in K\)
  and \(\eta_j\) for each \(j \in L\)
  be fresh indeterminates.
  We denote by \(\FFext\) the field of rational functions in those
  indeterminates with coefficients in $\FF$, and by
  $V$ the subset of elements that, when written in lowest terms,
  have denominators with nonzero constant terms.
  Let \(\Phi : V \to \FF\) map each element of \(V\) to the result of
  substituting \(\xi_j \gets 0\) for each \(j \in K\) and \(\eta_j \gets 0\)
  for each \(j \in L\).
  The result is always well-defined.

  Define \(\hat{f} \in \FFext(\alpha)\) as follows:
  \begin{equation*}
    \hat{f}(\alpha)
    \doteq
    z\cdot \frac{\prod_{j \in K}(\alpha - a_j + \xi_j)}%
                {\prod_{j \in L}(\alpha - a_j + \eta_j)}.
  \end{equation*}
  The substitution \(\RFE(\hat{f})\) effects
  \(x_i \gets \hat{f}(a_i) \in \FFext\) for each $i \in [n]$.
  We claim that \(p(\RFE(\hat{f}))\) is nonzero.
  This suffices to conclude that \(\RFEkl\) hits \(p\),
  because substituting \(\xi_j\) and \(\eta_j\) by a random scalar from
  \(\FF\) transforms \(\hat{f}\) into a seed $f'$
  such that,
  with high probability,
  $f'$ is a valid seed for \(\RFEkl\)
  and \(p(\RFE(f')) \ne 0\).
  Henceforth we show that \(p(\RFE(\hat{f})) \ne 0\).

  For each \(i \in [n]\),
  there exists \(\hat{c}_i \in V\)
  with \(\Phi(\hat{c}_i) \ne 0\)
  such that
  \begin{equation} \label{eq.multi-zsub-deriv.pf.seed}
    \hat{f}(a_i) = \begin{cases}
      \hat{c}_i\cdot \frac{\xi_i}{\eta_i} & i \in K\cap L \\
      \hat{c}_i\cdot \xi_i & i \in K\setminus L \\
      \hat{c}_i\cdot \frac{1}{\eta_i} & i \in L \setminus K \\
      \hat{c}_i & i \not\in K \cup L
    \end{cases},
  \end{equation}
  namely
  \[
    \hat{c}_i =  z\cdot
      \frac{\prod_{j \in K\setminus\{i\}}(a_i - a_j + \xi_j)}%
           {\prod_{j \in L\setminus\{i\}}(a_i - a_j + \eta_j)}.
  \]
  For \(i \not\in K \cup L\),
  \(\Phi(\hat{c}_i)\) is moreover the value substituted into \(x_i\) by
  \eqref{eq.multi-zsub-deriv.sub}.

  Let $D$ denote the set of all degree patterns $d \in \N^{K \cup L}$ that
  match a monomial in $p$.
  We have that
  \begin{equation} \label{eq.multi-zsub-deriv.pf.p-expand-1}
    p = \sum_{d\in D} p_d \cdot x^d.
  \end{equation}
  For $d \in D$, define \(\hat{q}_d\) to be the result of substituting
  \(x_i \gets \hat{c}_i\) into $p_d$ for each $i \in [n]$.

  Combining \eqref{eq.multi-zsub-deriv.pf.seed} and
  \eqref{eq.multi-zsub-deriv.pf.p-expand-1},
  we obtain
  \begin{equation}\label{eq.multi-zsub-deriv.pf.p-eval-1}
    p(\RFE(\hat{f}))
    =
    \sum_{d\in D} \hat{q}_d \cdot \hat{c}^d \cdot
      \frac{\xi^{d|_K}}{\eta^{d|_L}},
  \end{equation}
  where \(d|_K\) and \(d|_L\) respectively are the restrictions of \(d\) onto the domains \(K\) and \(L\) respectively.
  Fix any function \(\psi : [k+l] \to K\cup L\) such that
  \(\psi\) establishes a bijection between \(\{1,\ldots,k\}\) and \(K\)
  and establishes a bijection between \(\{k+1,\ldots,k+l\}\) and \(L\).
  For $j \in \{1,\dots,k\}$, let \(\zeta_j\) be an alias for \(\xi_{\psi(j)}\),
  and for $j \in \{k+1,\dots,k+l\}$, let \(\zeta_j\) be an alias for
  \(\eta_{\psi(j)}\).
  For each \(d \in \N^{K\cup L}\),
  define a corresponding \(\delta \in \Z^{k+l}\) given by
  \(\delta_j = d_{\psi(j)}\) for $j \in \{1,\dots,k\}$ and
  \(\delta_j = -d_{\psi(j)}\) for $j \in \{k+1,\dots,k+l\}$.
  Let \(\Delta \subseteq \Z^{k+l}\) consist of the \(\delta\) corresponding
  to each \(d \in D\).
  Finally, for each \(d \in D\) with corresponding \(\delta \in \Delta\),
  define
  $\hat{c}_\delta \doteq \hat{q}_d \cdot \hat{c}^d$,
  capturing the first two factors in the \(d\)-th term of
  \eqref{eq.multi-zsub-deriv.pf.p-eval-1}.
  Rewritten in this notation,
  \eqref{eq.multi-zsub-deriv.pf.p-eval-1} becomes
  \begin{equation}
    \sum_{\delta\in \Delta}
      \hat{c}_\delta \cdot
      \prod_{j=1}^{k+l} \zeta_j^{\delta_j}.
  \end{equation}

  Our hypothesis that $d^*$ is \((K,L)\)-extremal in \(p\)
  says that the only \(d \in D\) such that $d_j \le d^*_j$ for every $j \in K$
  and $d_j \ge d^*_j$ for every $j \in L$, is $d = d^*$.
  Translated into a condition on the element $\delta^* \in \Delta$
  corresponding to $d^*$,
  the hypothesis says that $\delta^*$ is minimal in $\Delta$.
  Our other hypothesis states that $p_{d^*}$ does not
  vanish upon substituting \eqref{eq.multi-zsub-deriv.sub}.
  As \eqref{eq.multi-zsub-deriv.sub} equates to substituting
  \(x_i \gets \Phi(\hat{c}_i)\) for \(i \not\in K \cup L\),
  this hypothesis equivalently states that \(\Phi(\hat{q}_{d^*})\) is nonzero.
  Since for each \(j \in K \cup L\) we have \(\Phi(\hat{c}_j) \ne 0\),
  we conclude that \(\Phi(\hat{c}_{\delta^*}) \ne 0\).
  That \(p(\RFE(\hat{f}))\) is nonzero now follows from the next proposition.
\end{proof}

\begin{proposition} \label{prop.multi-zsub-deriv.helper}
  Let \(\FFext = \FF(\zeta_1,\ldots,\zeta_r)\) be the field of rational
  functions in indeterminates \(\zeta_1, \ldots, \zeta_r\),
  let \(V \subseteq \FFext\) consist of the rational functions whose
  denominator has nonzero constant term,
  and let \(\Phi : V \to \FF\) be the function that maps each rational
  function in \(V\) to its value after substituting \(\zeta_j \gets 0\)
  for all \(j \in [r]\).
  Let
  \[
    s = \sum_{\delta\in \Delta} \hat{c}_\delta \cdot \prod_{j=1}^r \zeta_j^{\delta_j}
  \]
  where \(\Delta \subseteq \Z^r\) is some finite set,
  and
  we have \(\hat{c}_\delta \in V\)
  for every \(\delta \in \Delta\).
  If there exists \(\delta^* \in \Delta\)
  that is minimal in $\Delta$
  and
  for which \(\Phi(\hat{c}_{\delta^*}) \ne 0\),
  then $s \ne 0$.
\end{proposition}
\begin{proof}
  By clearing denominators,
  we may assume without loss of generality
  that,
    for every \(\delta \in \Delta\) and every $j \in [r]$,
    \(\delta_j \ge 0\),
  and that,
    for every \(\delta \in \Delta\),
    \(\hat{c}_\delta\) is a \emph{polynomial} in \(\zeta_1, \ldots, \zeta_r\).
  In this case, all quantities in the sum for \(s\) are polynomials in
  \(\zeta_1,\ldots,\zeta_r\).
  The minimality hypothesis on \(\delta^*\) implies that
  the coefficient of \(\prod_{j=1}^r \zeta_j^{\smash{\delta^*_j}}\) in the monomial
  expansion of \(s\) is precisely the constant coefficient of
  \(\hat{c}_{\delta^*}\),
  and the hypothesis \(\Phi(\hat{c}_{\delta^*}) \ne 0\) asserts that this
  coefficient is nonzero.
\end{proof}

\section{Membership Test}
\label{s.test}

In this section we develop the structured membership test for the vanishing ideal $\VanIdeal{\RFEkl}$ given in \expref{Theorem}{thm.ideal.membership}.
We begin with some basic results regarding membership to $\VanIdeal{\RFEkl}$ and then develop a criterion for multilinear polynomials.

\paragraph{Basic properties.}

It is well-known that $\VanIdeal{\SV^l}$ does not contain any polynomial with a monomial of support at most $l$, \ie, a monomial involving at most $l$ variables. We generalize the lower bound on the support to $\VanIdeal{\RFEkl}$ and also establish an upper bound in the case of multilinear polynomials. Note that for multilinear monomials support conditions translate into degree conditions.
\begin{proposition}
    \label{prop.ideal.membership.degree-bound}
    If a polynomial $p$ contains a monomial of support at most $l$, then $\RFEkl$ hits $p$. 
    If a multilinear polynomial $p$ in the variables $x_1,\dots,x_n$ contains a monomial of support at least $n-k$, then $\RFEkl$ hits $p$.
\end{proposition}
The known proofs of the first part for $\SV^l$ make use of partial derivatives. We establish the generalization for $\RFEkl$ using our generating set in \expref{Theorem}{thm.ideal.generators}, whose analysis hinges on the \hyperref[lem.multi-zsub-deriv]{Zoom Lemma}. A similar argument works for the second part, but we opt to establish it via a black-box reduction to the first part for multilinear polynomials. The approach illustrates the utility of our generalization of $\SVl$ to $\RFEkl$ since even for $\SVl$ we need to consider settings of the parameters $k$ and $l$ other than $k=l-1$. 
\begin{proof}
    For the first part, by \expref{Proposition}{prop.vangen.basic} none of the polynomials $\VanGenkl$ contain a monomial of support $l$ or less. The same holds for the nonzero polynomials in the ideal generated by these polynomials, which by \expref{Theorem}{thm.ideal.generators} equals $\VanIdeal{\RFEkl}$. Thus, every polynomial that has a monomial of support at most $l$, is hit by $\RFEkl$. 

    For the second part, consider $q(x_1,\dots,x_n) \doteq x_1\cdot \dots \cdot x_n\cdot p(1/x_1,\dots,1/x_n)$. Note that if $p$ is a multilinear polynomial in the variables $x_1, \dots, x_n$, then so is $q$. If a multilinear $p$ has a monomial of support at least $n-k$, then $q$ has a monomial of support at most $k$. By the first part, $\RFEparam{l}{k}$ hits $q$. Since the mapping $x_i\gets 1/x_i$ transforms $\RFEparam{l}{k}$ into $\RFEkl$, we conclude that $\RFEkl$ hits $p$.
\end{proof}

Another feature of $\SV$ that generalizes to $\RFE$ is that the generator separates the homogeneous components of a given polynomial $p$. The feature allows us to reduce the general case of testing membership in $\VanIdeal{\RFEkl}$ to the homogeneous case, as was already effectively used in the proof of the \hyperref[lem.multi-zsub-deriv]{Zoom Lemma}.
\begin{proposition}
  \label{prop.rfe-homog}
  For any polynomial \(p\),
  \(p\) vanishes upon substituting \(\RFE\)
  if and only if 
  every homogeneous component of \(p\) vanishes upon substituting \(\RFE\).
\end{proposition}
\begin{proof}
  For any seed \(f\) for \(\RFE\)
  and any scalar \(z\),
  the rescaled substitution \(z \cdot \RFE(f)\)
  is in the range of \(\RFE\),
  namely as \(\RFE(z\cdot f)\).
  It follows (provided that \(\FF\) is sufficiently large)
  that \(p(\RFE)\) vanishes
  if and only if
  \(p(\zeta\cdot\RFE)\) vanishes,
  where \(\zeta\) is a fresh indeterminate.
  We now consider the expansion of \(p(\zeta\cdot\RFE)\)
  as a polynomial in \(\zeta\).
  With \(p^{(d)}\) as the degree-\(d\) homogeneous component of \(p\),
  we have
  \begin{equation*}
    p(\zeta \cdot \RFE)
      = \sum_d p^{(d)}(\zeta \cdot \RFE)
      = \sum_d \zeta^d \cdot p^{(d)}(\RFE).
  \end{equation*}
  The coefficient of \(\zeta^d\), \(p^{(d)}(\RFE)\),
  has no dependence on \(\zeta\).
  We deduce that
  \(p(\zeta\cdot\RFE)\) is the zero polynomial
  if and only if
  \(p^{(d)}(\RFE)\) vanishes for every \(d\).
\end{proof}

\paragraph{Criterion for multilinear polynomials.}

We now develop the full membership test for multilinear polynomials given in \expref{Theorem}{thm.ideal.membership}. \expref{Condition}{thm.ideal.membership.cond2} in \expref{Theorem}{thm.ideal.membership} is closely related to the \hyperref[lem.multi-zsub-deriv]{Zoom Lemma}. Note that for multilinear polynomials and disjoint \(K\) and \(L\), \(\PartialZero{L}{K}{p}\) coincides with the coefficient \(p_{d^*}\) where \(d^*\) is the degree pattern with domain \(K \sqcup L\), \(0\) in the positions of \(K\), and \(1\) in the positions of \(L\).
Moreover, since \(p\) is multilinear,
the condition that \(d^*\) be \((K,L)\)-extremal in \(p\) is
automatically satisfied:
The only multilinear monomial \({m}\) with support in 
\(K\sqcup L\) with
\(\deg_{x_i}({m}) \le d^*_i = 0\) for all \(i\in K\)
and
\(\deg_{x_i}({m}) \ge d^*_i = 1\) for all \(i\in L\)
is \({m} = x^{d^*}\).
This leads to the following specialization of the \hyperref[lem.multi-zsub-deriv]{Zoom Lemma} for multilinear
polynomials with disjoint \(K\) and \(L\).

\begin{lemma}[Zoom Lemma for multilinear polynomials]\label{lem.multi-zsub-deriv.mlin}
  Let \(K,L \subseteq [n]\) be disjoint,
  and
  let \(p\in\FF[x_1,\ldots,x_n]\) be a multilinear polynomial.
  If
  \(\PartialZero{L}{K}{p}\)
  is nonzero upon the substitution
  \begin{equation}
    \label{eq.multi-zsub-deriv.mlin.sub}
    x_i \gets
    z \cdot
      \frac{
        \prod_{j\in K} ( a_i - a_{j} )
      }{
        \prod_{j\in L} ( a_i - a_{j} )
      }
      \qquad \forall i\in [n]\setminus (K\sqcup L),
  \end{equation}
  where $z$ is a fresh variable,
  then \(\RFEkl\) hits \(p\) for any \(k\ge|K|\) and \(l\ge|L|\).
\end{lemma}

Observe that the substitution \eqref{eq.multi-zsub-deriv.mlin.sub} in \expref{Lemma}{lem.multi-zsub-deriv.mlin} coincides with \eqref{eq.thm.ideal.membership.sub} in \expref{Theorem}{thm.ideal.membership}. Thus, for a multilinear polynomial $p$, \expref{condition}{thm.ideal.membership.cond2} in \expref{Theorem}{thm.ideal.membership} expresses that there is no way to show that $\RFEkl$ hits $p$ via an application of \expref{Lemma}{lem.multi-zsub-deriv.mlin}. The necessity of this condition for membership of $p$ in $\VanIdeal{\RFEkl}$ is clear. 
The necessity of \expref{Condition}{thm.ideal.membership.cond1} in \expref{Theorem}{thm.ideal.membership} is just the special case of \expref{Proposition}{prop.ideal.membership.degree-bound} for multilinear polynomials. 

What remains to argue is that the combination of \expref{condition}{thm.ideal.membership.cond1} and \expref{condition}{thm.ideal.membership.cond2} is sufficient for membership of $p$ in $\VanIdeal{\RFEkl}$. Equivalently, it remains to argue the following property for every multilinear polynomial $p \in \FF[x_1,\dots,x_n]$ that only contain monomials of degrees between $l+1$ and $n-k-1$: If $p$ is hit by $\RFEkl$ then there is an application of \expref{Lemma}{lem.multi-zsub-deriv.mlin} that exhibits this fact. We actually prove the property for every multilinear $p$ of degree $d$ with $l \le d \le n-k$. We do so with a two-step strategy similar to one we used for the part of \expref{Theorem}{thm.ideal.generators} that $\RFEkl$ hits every polynomial outside of the ideal generated by instantiations of $\VanGenkl$:
\begin{enumerate}
  \item
    Modulo the ideal $I$ generated by a certain subset of the instantiations of $\VanGenkl$, 
    $p$ is equal
    to a cored polynomial $r$ with certain parameters. Since $p \not\in \VanIdeal{\RFEkl}$ and $I \subseteq \VanIdeal{\RFEkl}$, $r$ needs to be nonzero. 
  \item
    For every such cored polynomial $r$ that is nonzero, we can apply \expref{Lemma}{lem.multi-zsub-deriv.mlin} to prove that $\RFEkl$ hits $r$, \ie, \expref{condition}{thm.ideal.membership.cond2} fails for $r$.
\end{enumerate}
By linearity and the necessity of \expref{condition}{thm.ideal.membership.cond2} for multilinear polynomials in $\VanIdeal{\RFEkl}$), we conclude that the condition fails for $p$, as well.

The crux for the first step in the context of \expref{Theorem}{thm.ideal.generators} is the transformation in \expref{Proposition}{prop.lem.vangen-coredform}, which gradually gets closer to a cored polynomial with the desired parameters. In general, the transformation in \expref{Proposition}{prop.lem.vangen-coredform} does not maintain multilinearity. We show how to tweak the transformation and preserve multilinearity at the expense of an increase in the size of the core.

\begin{proposition}
\label{prop.membership-coredform}
Let \(k,l,n,d \in\N\), let \(C\) be a \((k+d-l)\)-subset of \([n]\), and let $I$ denote the ideal generated by the polynomials $\VanGenkl[K \sqcup L]$ where $K$ ranges over all $(k+1)$-subsets of $C$ and $L$ ranges over all $(l+1)$-subsets of $[n] \setminus K$. Consider a multilinear monomial $m \in \FF[x_1,\dots,x_n]$ of degree at most $d$ such that $|\supp(m) \setminus C| > l$. Modulo $I$, $m$ is equal to a linear combination of multilinear monomials of the same degree as $m$ but whose non-$C$-parts have lower degree than the non-$C$-part of $m$.
\end{proposition}
\begin{proof}
Consider the subset $L \subseteq \supp(m) \setminus C$ of size $l+1$ in the proof of \expref{Proposition}{prop.lem.vangen-coredform}, and $x^L \doteq \prod_{i \in L} x_i$. Since $m$ is multilinear, so is $m' \doteq m/x^L$, and $|\supp(m')| \le d-|L| = d-l-1$. Provided $|C| \ge (k+1) + (d-l-1) = k+d-l$, there exists a subset $K \subseteq C$ of size $k+1$ that is disjoint from $\supp(m')$. We substitute $\VanGenkl[C \sqcup L]$ by $\VanGenkl[K \sqcup L]$ in the proof. $\VanGenkl[K \sqcup L]$ is homogeneous and multilinear. By construction $K \sqcup L$ is disjoint from $\supp(m')$, so $\VanGenkl[K \sqcup L]$ does not depend on any variables that $m'$ depends on. It follows that $m' \cdot \VanGenkl[K \sqcup L]$ is multilinear and homogeneous of the same degree as $m$, and so is $r$ in the proof.
\end{proof}

Applying \expref{Proposition}{prop.membership-coredform} repeatedly in a similar way as \expref{Proposition}{prop.lem.vangen-coredform} in the proof of \expref{Lemma}{lem.vangen-coredform} yields the following formalization of the first step in the setting of \expref{Theorem}{thm.ideal.membership}.
\begin{lemma}
  \label{lem.membership-coredform}
  Let \(k,l,n,d \in\N\), let \(C\) be a \((k+d-l)\)-subset of \([n]\), and let $I$ denote the ideal generated by the polynomials $\VanGenkl[K \sqcup L]$ where $K$ ranges over all $(k+1)$-subsets of $C$ and $L$ ranges over all $(l+1)$-subsets of $[n] \setminus K$.
  Modulo $I$,
  every multilinear polynomial $p $ of degree at most $d$ in $\FF[x_1,\dots,x_n]$ is equal to a \((k+d-l,l)\)-cored multilinear polynomial with core \(C\) that is either zero or else has the same degree as $p$.
\end{lemma}

The following refinement of \expref{Lemma}{lem.rfe-hits-cored} from the context of \expref{Section}{s.gens} represents the corresponding second step in the context of \expref{Theorem}{thm.ideal.membership}. This is where the degree constraint comes into play.

\begin{lemma}\label{lemma.rfe-hits-multiaffine-coreform}
Let \(k,l,n,d \in\N\) with $l \le d \le n-k$. 
Let \(r\) be a nonzero multilinear polynomial of degree $d$ in $\FF[x_1,\dots,x_n]$ that is \((d+k-l, l)\)-cored.
There are disjoint sets \(K,L\subseteq [n]\) with \(|K|=k\) and
\(|L|=l\) so that \(\PartialZero{L}{K}{r}\) is nonzero upon the substitution \eqref{eq.multi-zsub-deriv.mlin.sub}.
\end{lemma}
\begin{proof}
  Let $C$ denote the core of size at most $d+k-l$, and let $m$ be a monomial of $r$ of degree $d$ that maximizes $|\supp(m) \setminus C|$.
  Let $K$ be a subset of $[n]\setminus \supp(m)$ of size $k$ that 
  contains all of $C \setminus \supp(m)$.
  Such a set $K$ exists because $|C \setminus \supp(m)| = |C| - |\supp(m) \cap C| \le |C| - (d-l) \le k$, and $|[n] \setminus \supp(m)| = n-d \ge k$. 
  Let $L$ be a subset of $\supp(m)$ of size $l$ that contains all of $\supp(m) \setminus C$.
  Such a set $L$ exists because $|\supp(m) \setminus C| \le l$ and $|\supp(m)| = d \ge l$. Note that $K$ and $L$ are disjoint. 
  
  The monomial $m$ has a nonzero contribution to $\PartialZero{L}{K}{r}$. In general, a monomial $m'$ has a nonzero contribution to $\PartialZero{L}{K}{r}$ if and only if $\supp(m')$ is disjoint from $K$ and contains $L$. The disjointness requirement implies that $\supp(m') \cap C \subseteq C \setminus K = \supp(m) \cap C$, where the  equality follows from the choice of $K$. The inclusion requirement implies that $\supp(m) \setminus C = L \setminus C \subseteq \supp(m') \setminus C$, where the equality follows from the choice of $L$. In combination with the maximality of $|\supp(m) \setminus C|$ among the monomials of $r$ of degree $d$, this means that either $m'$ does  not have degree $d$ or else $\supp(m')\setminus C = \supp(m) \setminus C$. It follows that the only monomials $m'$ of $r$ of degree $d$ that contribute to $\PartialZero{L}{K}{r}$ satisfy $\supp(m') \subseteq \supp(m)$. As $r$ only contains multilinear monomials, $r$ has exactly one monomial of degree $d$ that has a nonzero contribution to $\PartialZero{L}{K}{r}$, namely the monomial $m$. We conclude that the polynomial $q(z)$ that results from  substituting \eqref{eq.multi-zsub-deriv.mlin.sub} into 
  \(\PartialZero{L}{K}{r}\) has a nonzero term of degree $d-l$.
\end{proof}

We now have all ingredients to establish 
\expref{Theorem}{thm.ideal.membership}. 

\begin{proof}[Proof of \expref{Theorem}{thm.ideal.membership}]
The necessity of \expref{condition}{thm.ideal.membership.cond1} and \expref{condition}{thm.ideal.membership.cond2} for the membership of a 
 multilinear polynomial $p \in \FF[x_1,\dots,x_n]$ in $\VanIdeal{\RFEkl}$ immediately follows from \expref{Proposition}{prop.ideal.membership.degree-bound} and \expref{Lemma}{lem.multi-zsub-deriv.mlin}, respectively. For sufficiency, we need to show that if $p$ only contains monomials of degrees between $l+1$ and $n-k-1$ and is hit by $\RFEkl$, then there exist disjoint sets \(K,L\subseteq [n]\) with \(|K|=k\) and \(|L|=l\) so that \(\PartialZero{L}{K}{p}\) is nonzero upon the substitution \eqref{eq.multi-zsub-deriv.mlin.sub}.
 
 By \expref{Lemma}{lem.membership-coredform}, we can write $p$ as $p=q+r$, where $q \in \VanIdeal{\RFEkl}$ and $r$ is a multilinear \((k+d-l,l)\)-cored polynomial that is either zero or else has the same degree as $p$. Since $p \not\in \VanIdeal{\RFEkl}$, the case of zero $r$ is ruled out. Thus $r$ is a multilinear \((k+d-l,l)\)-cored polynomial of degree $d$, where $l+1 \le d \le n-k-1$. \expref{Lemma}{lemma.rfe-hits-multiaffine-coreform} then yields disjoint sets \(K,L\subseteq [n]\) with \(|K|=k\) and \(|L|=l\) so that \(\PartialZero{L}{K}{r}\) is nonzero upon the substitution \eqref{eq.multi-zsub-deriv.mlin.sub}. As both $p$ and $r$ are multilinear, so is $q = p-r$. The contrapositive of \expref{Lemma}{lem.multi-zsub-deriv.mlin} implies that \(\PartialZero{L}{K}{q}\) is zero upon the substitution \eqref{eq.multi-zsub-deriv.mlin.sub}. It follows that \(\PartialZero{L}{K}{p} = \PartialZero{L}{K}{q} + \PartialZero{L}{K}{r}\) is nonzero upon the substitution \eqref{eq.multi-zsub-deriv.mlin.sub}.
\end{proof}

We conclude this section by detailing the connection between
\expref{Theorem}{thm.ideal.membership} and some prior applications of the SV generator.

\paragraph{Application to read-once formulas.}
We start with the theorem that \(\SVparam{1}\) hits read-once formulas.
The original proof in \cite{MinahanVolkovich2017} goes by induction on the
depth of $F$. The critical part is the inductive step for the case where the top gate is an addition, say $F = F_1 + F_2$. The argument in \cite{MinahanVolkovich2017} involves a clever analysis that uses the variable-disjointness of $F_1$ and $F_2$ to show that $F_1(\SVparam{1})$ and $F_2(\SVparam{1})$ cannot cancel each other out. We present an alternate proof that has a similar inductive outline but follows a more structured, principled approach based on \expref{Theorem}{thm.ideal.membership} for the critical part.

\begin{theorem}[\cite{MinahanVolkovich2017}]
    \label{thm.rof}
    $\SV^1$ hits read-once formulas.
\end{theorem}

\begin{proof}[Alternate proof]
We show by induction on the depth the formula $F$ that if $F$ is nonconstant, then so is $F(\SVparam{1})$. This suffices because it implies that nonconstant formulas are hit by $\SVparam{1}$, and nonzero constant formulas are hit as the range of $\SVparam{1}$ is nonempty.

The inductive step consists of two cases, depending on whether
the top gate is a multiplication gate or an addition gate. The case of a multiplication gate follows from the general property that the product of a nonconstant polynomial with any nonzero polynomial is nonconstant. It remains to consider the case of an addition gate.

For a nonconstant formula $F$, $F(\SVparam{1})$ is nonconstant iff $\SVparam{1}$ hits the variable part of $F$ (which is a nonzero polynomial). By \expref{Theorem}{thm.ideal.membership} with $k=0$ and $l=1$, the latter is the case iff at least one of the following two conditions hold:
\begin{enumerate}
\item $F$ has a homogeneous component of degree 1 or at least $n$.
\item For some $L = \{i\} \subseteq [n]$, the derivative $\partial_{x_i} F$ is nonzero upon the substitution \eqref{eq.thm.ideal.membership.sub}.
\end{enumerate}
Consider a read-once formula $F$ with an addition gate on top: $F = F_1 + F_2$. The variable-disjointness of $F_1$ and $F_2$ implies that if condition 1 holds for at least one of $F_1$ or $F_2$, then it holds for $F$. The same is true for condition 2. The inductive step in the case of an addition gate at the top follows. 
\end{proof}

The case of an addition gate in the above proof has a clean geometric interpretation along the lines of the alternating algebra representation that we discussed in \expref{Section}{s.intro} for polynomials that are multilinear (which polynomials computed by read-once formulas are).  Recall that we can think of the variables as vertices, and multilinear monomials as simplices made from those vertices.\footnote{In this setting the orientation of the simplices does not matter.} A multilinear polynomial is a weighted collection of such simplices with weights from \(\FF\).
In this view, \expref{Theorem}{thm.ideal.membership} translates to the following characterization: a weighted collection of simplices corresponds to a polynomial in the vanishing ideal of \(\RFEparam{0}{1}\) iff there are no simplices of zero, one, or all vertices (\expref{condition}{thm.ideal.membership.cond1}), and the remaining weights satisfy a certain system of linear equations (\expref{condition}{thm.ideal.membership.cond2}). 
Crucially, for each equation in the system, there is a vertex such that the equation only involves weights of the simplices \emph{that contain that vertex}, namely the vertex corresponding to the variable $x_i$ where $L = \{i\}$. 
Meanwhile, the sum of two variable-disjoint polynomials corresponds to taking the vertex-disjoint union of two weighted collections of simplices.
It follows directly that if either of the two polynomials violates a requirement besides the ``no simplex of zero vertices'' requirement, then their sum violates the same requirement. The ``no simplex of zero vertices'' requirement holds automatically when considering the variable parts, and maps to the special handling of the constant term in the formal proof. 

\paragraph{Zero-substitutions and partial derivatives.}
As mentioned in the overview,
several prior papers demonstrated the utility of partial derivatives and zero
substitutions in the context of derandomizing PIT using the SV generator,
especially for syntactically multilinear models.
By judiciously choosing variables for those operations,
these papers managed to simplify $p$ and reduce PIT for $p$ to PIT for simpler
instances, resulting in an efficient recursive algorithm.
Such recursive arguments can be wrapped into a general framework, similar to the one presented in \cite{MediniShpilka2021} for generic $l$-independent generators. Whereas the power of the framework in the generic setting remains open, thanks to \expref{Theorem}{thm.ideal.membership}, we can prove that our framework captures the full power of the specific $l$-independent generator $\SVl$. More generally, we exhibit a natural reformulation within the framework of any argument that $\RFE$ hits a certain class of multilinear polynomials, such as those computable with some bounded complexity in some syntactic model.

For the argument, we assume that we can break up the class in the following way.
\begin{definition}[grading hypothesis]
A class $\Class = \bigcup_{k,l \in \N} \Class_{k,l}$ of polynomials satisfies the grading hypothesis if for every \(k,l \in \N\) and \(p\in \Class_{k,l}\),
at least one of the following holds:
\begin{itemize}
  \item \(k=l=0\) and \(p\) is nonzero.
  \item
    \(k > 0\) and there is a zero substitution such that the result is in
    \(\Class_{k-1,l}\).
  \item
    \(l > 0\) and there is a partial derivative such that the result is in
    \(\Class_{k,l-1}\).
\end{itemize}
\end{definition}
Under the additional mild assumption of closure under variable rescaling, we obtain a parameter-efficient framework through
direct applications of \expref{Theorem}{thm.ideal.membership}.
\begin{proposition}\label{prop.zsub-deriv-power}
Let $\Class = \bigcup_{k,l \in \N} \Class_{k,l}$ be a class of polynomials that satisfies the grading hypothesis and such that each $\Class_{k,l}$ is closed under variable rescaling. If $\RFEparam{0}{0}$ hits $\Class_{0,0}$ then 
\(\RFEkl\) hits \(\Class_{k,l}\) for every \(k,l \in \N\).
\end{proposition}
\begin{proof}
  The proof is by induction on \(k\) and \(l\).
  The base case is \(k=l=0\), where the claim is immediate.
  When \(k>0\) or \(l > 0\), our hypotheses are such that $p \in \Class_{k,l}$ either simplifies under a zero substitution
  or a partial derivative. In either case, we show how a violation of the conditions in \expref{Theorem}{thm.ideal.membership} for a simpler polynomial $p' \in \FF[x'_1,\dots,x'_n]$ translates into a corresponding violation of the conditions for $p \in \FF[x_1,\dots,x_n]$, where each variable $x'_i$ is a rescaling of $x_i$. More specifically, by \expref{condition}{thm.ideal.membership.cond1} of
  \expref{Theorem}{thm.ideal.membership}, we may assume that $p$ only has homogeneous components with degrees in the range
  $l+1,\ldots,n-k-1$. We argue in both cases that $p'$ similarly satisfies \expref{condition}{thm.ideal.membership.cond1} of
  \expref{Theorem}{thm.ideal.membership}. By the induction hypothesis and closure under variable rescaling, it follows that $\PartialZeroPrime{L'}{K'}{p'}$ (where the prime in $\partial'$ indicates that the partial derivatives are with respect to the primed variables $x'_i$) is nonzero for some $K'$ and $L'$ under a particular substitution. Out of $K'$ and $L'$ we then construct $K$ and $L$ such that $\PartialZero{L}{K}{p}$ is nonzero upon the substitution in \expref{condition}{thm.ideal.membership.cond2} of
  \expref{Theorem}{thm.ideal.membership}, where variable rescaling between $x'_i$ and $x_i$ enables us to match the substitutions for $\PartialZeroPrime{L'}{K'}{p'}$ and $\PartialZero{L}{K}{p}$. We provide the remaining details for each case separately.
  \begin{itemize}
    \item
      If $p$ simplifies under a zero substitution $x_{j^*}\gets 0$, then
      write $p$ as $p = q x_{j^*} + r$ where $q$ and $r$ are
      polynomials that do not depend on $x_{j^*}$,
      and set $p'(\dots,x'_i,\dots)=r(\dots,x_i,\dots)$ with $x_i=x'_i \cdot (a_i-a_{j^*})$.
      By closure under rescaling, $p'\in \Class_{k-1,l}$, so by induction $p'$ is hit by \(\RFEparam{k-1}{l}\).
      We apply \expref{Theorem}{thm.ideal.membership} to $p'$ with respect to the set of variables \(\{x'_1,\ldots,x'_{j^*-1},x'_{j^*+1},\ldots,x'_n\}\) and $k$ replaced by $k-1$.
      As $p$ only has homogeneous components with degrees in the range $l+1,\ldots,n-k-1$, so does $p'$,
      and \expref{condition}{thm.ideal.membership.cond1} of
      \expref{Theorem}{thm.ideal.membership} holds for $p'$.
      This means that \expref{condition}{thm.ideal.membership.cond2} does not hold for $p'$.
      Thus, there must be disjoint $K',L' \subseteq [n]\setminus\{j^*\}$
      with $|K'|=k-1$ and $|L'|=l$ so that $\PartialZeroPrime{L'}{K'}{p'}$ is nonzero upon the substitution
      \begin{equation}\label{eq-framework-sub}
        x'_i \gets z\cdot \frac%
          { \prod_{j \in K'} (a_i - a_{j}) }%
          { \prod_{j \in L'} (a_i - a_{j}) }.
      \end{equation}
      Setting $K = K' \cup \{j^*\}$ and $L=L'$, we have
      \[ \PartialZero{L}{K}{p} = \PartialZero{L'}{K'}{r} = \PartialZeroPrime{L'}{K'}{p'} \Big/ \prod_{i \in L'} (a_i-a_{j^*}) \]
      and the substitution \eqref{eq-framework-sub} induces the substitution \eqref{eq.thm.ideal.membership.sub}.
      \item
      If $p$ simplifies under a partial derivative
      \(\partial_{x_{j^*}}\),
      then write $p$ as $p = q x_{j^*} + r$ where $q$ and $r$ are polynomials that do not depend on $x_{j^*}$,
      and set
      $p'(\dots, x'_i, \dots) \doteq q(\dots, x_i, \dots)$ with $x_i=x'_i / (a_i-a_{j^*})$. By closure under rescaling, $p'\in \Class_{k,l-1}$, so by induction $p'$ is hit by \(\RFEparam{k}{l-1}\).
      We apply \expref{Theorem}{thm.ideal.membership} to $p'$ with respect to the set of
      variables \(\{x'_1,\ldots,x'_{j^*-1},x'_{j^*+1},\ldots,x'_n\}\) and $l$ replaced by $l-1$.
      As $p'$ has homogeneous components of degrees one less than $p$ does,
      \expref{condition}{thm.ideal.membership.cond1} of
      \expref{Theorem}{thm.ideal.membership} holds for $p'$, so \expref{condition}{thm.ideal.membership.cond2} must fail. Thus, 
      there are disjoint $K',L' \subseteq [n]\setminus\{j^*\}$ with
      $|K'|=k$ and $|L'|=l-1$ so that $\PartialZeroPrime{L'}{K'}{p'}$ is nonzero upon the substitution
      \eqref{eq-framework-sub}.
      Setting $K = K'$ and $L=L' \cup \{j^*\}$, we have 
      \[ \PartialZero{L}{K}{p} = \PartialZero{L'}{K'}{q} = \PartialZeroPrime{L'}{K'}{p'} \cdot \prod_{i \in L'} (a_i-a_{j^*}) \]
      and the substitution \eqref{eq-framework-sub} induces the substitution \eqref{eq.thm.ideal.membership.sub}. 
  \end{itemize}
  In both cases we conclude that $\PartialZero{L}{K}{p}$ is nonzero upon the substitution \eqref{eq.thm.ideal.membership.sub}, which is the sought violation of \expref{condition}{thm.ideal.membership.cond2} of \expref{Theorem}{thm.ideal.membership}.
\end{proof}

We remark that the mild requirement of closure under variable rescaling in \expref{Proposition}{prop.zsub-deriv-power} can be dropped completely at the cost of reduced efficiency in parameters.\footnote{This is a setting where we exploit the possibility of the sets $K$ and $L$ in the \hyperref[lem.multi-zsub-deriv]{Zoom Lemma} to overlap.}
\begin{proposition}\label{prop.zsub-deriv-power-variant}
Let $\Class = \bigcup_{k,l \in \N} \Class_{k,l}$ be a class of polynomials that satisfies the grading hypothesis. If $\RFEparam{0}{0}$ hits $\Class_{0,0}$ then 
\(\RFEparam{k+l}{k+l}\) hits \(\Class_{k,l}\) for every \(k,l \in \N\).
\end{proposition}
\begin{proof}[Proof sketch]
The strategy is the same as in the proof of \expref{Proposition}{prop.zsub-deriv-power}, but in the inductive step the index $i^*$ is added to both $K'$ and $L'$ instead of just one of the two sets. This obviates the need for rescaling to ensure that the substitutions match. Note that the resulting sets $K$ and $L$ are no longer disjoint, but the general Zoom Lemma accommodates overlapping sets $K$ and $L$.
\end{proof}

\expref{Theorem}{thm.ideal.membership} tells us that derivatives and zero substitutions
suffice to witness when a multilinear polynomial \(p\) is hit by \(\SV\) or
\(\RFE\).
One can ask,
if we know more information about \(p\),
can we infer \emph{which} derivatives and zero substitutions form a witness?
In some cases we know.
For example,
if \(p\) has a low-support monomial \(x_1\cdots x_l\),
then it suffices to take derivatives with respect to each of
\(x_1,\ldots,x_l\).
On the other hand,
consider that whenever two polynomials \(p\) and \(q\) are hit by \(\SV\),
then so is their product \(pq\).
Given explicit witnesses for \(p\) and \(q\),
we do not know how to obtain an explicit witness for the product \(pq\).

\section{Sparseness}
\label{s.sparse}

By \expref{Proposition}{prop.vangen.basic},
the generators $\VanGenkl$ contain exactly $\binom{k+l+2}{l+1}$ monomials.
The following result shows that no nonzero polynomial in the vanishing ideal
of $\RFEkl$ has fewer monomials.
\expref{Corollary}{cor.rfe.sparseness} follows.

\begin{lemma}
  \label{lem.rfe-hits-sparse}
  Suppose \(p \in \FF[x_1,\ldots,x_n]\)
  is nonzero and has only \(s\) monomials with nonzero coefficients.
  Then, for any \(k,l\) such that \(\binom{k+l+2}{l+1} > s\),
  \(\RFEkl\) hits \(p\).
\end{lemma}
\noindent
The tactic here is to show that,
if \(p\) has too few monomials appearing in it,
then there is a way to instantiate the \hyperref[lem.multi-zsub-deriv]{Zoom Lemma}
wherein \(p_{d^*}\) is a single monomial
and therefore is nonzero upon the substitution \eqref{eq.multi-zsub-deriv.sub}.
\begin{proof}
  For \(i\in [n]\), we define two operations,
  \(\downarrow_i\) and \(\uparrow_i\),
  on nonempty sets of monomials.
  Applying $\downarrow_i$ to such a set \({M}\) yields the subset of \({M}\) consisting of the monomials in which
  \(x_i\) appears with its least degree among all the monomials in
  \({M}\).
  We define \(\uparrow_i\) similarly,
  except we select the monomials in which \(x_i\) appears with its highest
  degree.
  We make the following claim:

  \begin{claim}
    \label{claim.lem.rfe-hits-sparse}
    For any nonempty set of monomials with fewer than
    \(\binom{k+l+2}{l+1}\) monomials,
    there is a sequence of \(\downarrow\) and \(\uparrow\) operations,
    with at most \(k\) \(\downarrow\) operations and at most \(l\)
    \(\uparrow\) operations,
    such that the resulting set of monomials has exactly one element.
  \end{claim}

  The claim implies the lemma as follows.
  Let \({M}\) be the set of monomials with nonzero coefficient in
  \(p\).
  Apply the claim to \({M}\) to get a sequence of \(\downarrow\)
  and \(\uparrow\) operations resulting in a single monomial
  \({m}_0\).
  Let \(K\) denote the indices used for the \(\downarrow\)
  operations and \(L\) the indices used for the \(\uparrow\)
  operations.
  Let \(d^*\) be the degree pattern with domain \(K \cup L\)
  that matches \({m}_0\).
  By how the operators are defined,
  every monomial \({m}\) in \({M}\) satisfies either
  \begin{itemize}
    \item
      \(\deg_{x_i}({m}) > d^*_i\) for some \(i\in K\)
      (\({m}\) was removed by \(\downarrow_i\)),
    \item
      \(\deg_{x_i}({m}) < d^*_i\) for some \(i\in L\)
      (\({m}\) was removed by \(\uparrow_i\)),
      or
    \item \(\deg_{x_i}({m}) = d^*_i\) for every \(i\in
      K\cup L\), in which case \({m} = {m}_0\).
  \end{itemize}
  Accordingly,
  $d^*$ is $(K,L)$-extremal in $p$ and the \hyperref[lem.multi-zsub-deriv]{Zoom Lemma} applies.
  As \(p_{d^*}\) is a single monomial,
  it is nonzero upon the substitution \eqref{eq.multi-zsub-deriv.sub}.
  As $|K|\le k$ and $|L|\le l$, we conclude that \(p\) is hit by \(\RFEkl\).

  \medskip

  It remains to prove \expref{Claim}{claim.lem.rfe-hits-sparse}.
  We do this by induction on \(|{M}|\).
  In the base case, \(|{M}| = 1\), in which case the empty
  sequence suffices.
  Otherwise, \(|{M}| > 1\), in which case there is a variable
  \(x_i\) that appears with at least two distinct degrees among monomials in
  \({M}\).
  The sets \({\downarrow_i}({M})\) and \({\uparrow_i}({M})\) are
  nonempty and disjoint.
  Since \({M}\) has size less than
  \(\binom{k+l+2}{l+1} = \binom{k+l+1}{l+1} + \binom{k+l+1}{l}\),
  either \({\downarrow_i}({M})\) has size less than
  \(\binom{k+l+1}{l+1}\),
  or \({\uparrow_i}({M})\) has size less than \(\binom{k+l+1}{l}\).
  Whichever is the case, the claim follows by applying the inductive
  hypothesis to it.
\end{proof}

\section{Set-Multilinearity}
\label{s.set-multiaffine}

Although the generators $\VanGenkl$ provided by
\expref{Theorem}{thm.ideal.generators} are not set-multilinear,
the vanishing ideal of $\RFEkl$ does contain set-multilinear polynomials.
In this section,
we construct some of degree $l+1$ with partition classes of size $k+2$.
In fact, we argue that all set-multilinear polynomials in $\VanIdeal{\RFEkl}$
of degree $l+1$ are in the linear span of the ones we construct.

Our construction is a modification of the one for $\VanGenkl$.
\begin{definition}
  \label{def.smvan1}
  Let \(k,l,n \in \N\),
  and let \(S_1,\ldots,S_{l+1}\subseteq [n]\) be \(l+1\) disjoint
  subsets of \(k+2\) indices each.
  The polynomial \(\SMVanOnekl\) is an \((l+1)\times(l+1)\) determinant
  where each entry is itself a \((k+2)\times(k+2)\) determinant.
  We index the rows in the outer determinant by \(i=1,\ldots,l+1\),
  and the columns by \(d=l,\ldots,0\).
  In each \((i,d)\)-th inner matrix, there is one row per \(j\in S_i\); it is
  \[
    \begin{bmatrix}
        a_j^k & a_j^{k-1} & \cdots & a_j^1 & a_j^0 & a_j^d x_j
    \end{bmatrix}.
  \]
\end{definition}
\noindent
The name ``\(\SMVanOne\)'' is a shorthand for ``Elementary Set-Multilinear
Vandermonde Circulation''.
Similar to \(\VanGen\), the precise instantiation of \(\SMVanOne\)
requires one to pick an order for the sets \(S_1,\ldots,S_{l+1}\) and an order within each set. 

\begin{example}
  \label{ex.smvan1.k1l2}
  When \(k=1\) and \(l=2\), \(\SMVanOne\) uses three sets of three
  variables each.
  To help convey the structure of the determinant,
  we name the variable-sets
  \(S_1 = \{x_1,x_2,x_3\}\),
  \(S_2 = \{y_1,y_2,y_3\}\), and
  \(S_3 = \{z_1,z_2,z_3\}\),
  and denote the abscissa of \(x_i\) by \(a_i\),
  the abscissa of \(y_i\) by \(b_i\),
  and the abscissa of \(z_i\) by \(c_i\).
  With this notation and using the index ordering, \(\SMVanOne\) is the following:
  \begingroup\renewcommand{\arraystretch}{1.5}
  \[
    \begin{vmatrix*}
      \begin{vmatrix*}
        a_1^1 & a_1^0 & a_1^2 x_1 \\
        a_2^1 & a_2^0 & a_2^2 x_2 \\
        a_3^1 & a_3^0 & a_3^2 x_3
      \end{vmatrix*}
      &
      \begin{vmatrix*}
        a_1^1 & a_1^0 & a_1^1 x_1 \\
        a_2^1 & a_2^0 & a_2^1 x_2 \\
        a_3^1 & a_3^0 & a_3^1 x_3 
      \end{vmatrix*}
      &
      \begin{vmatrix*}
        a_1^1 & a_1^0 & a_1^0 x_1 \\
        a_2^1 & a_2^0 & a_2^0 x_2 \\
        a_3^1 & a_3^0 & a_3^0 x_3 
      \end{vmatrix*}
      \\
      \begin{vmatrix*}
        b_1^1 & b_1^0 & b_1^2 y_1 \\
        b_2^1 & b_2^0 & b_2^2 y_2 \\
        b_3^1 & b_3^0 & b_3^2 y_3 
      \end{vmatrix*}
      &
      \begin{vmatrix*}
        b_1^1 & b_1^0 & b_1^1 y_1 \\
        b_2^1 & b_2^0 & b_2^1 y_2 \\
        b_3^1 & b_3^0 & b_3^1 y_3 
      \end{vmatrix*}
      &
      \begin{vmatrix*}
        b_1^1 & b_1^0 & b_1^0 y_1 \\
        b_2^1 & b_2^0 & b_2^0 y_2 \\
        b_3^1 & b_3^0 & b_3^0 y_3 
      \end{vmatrix*}
      \\
      \begin{vmatrix*}
        c_1^1 & c_1^0 & c_1^2 z_1 \\
        c_2^1 & c_2^0 & c_2^2 z_2 \\
        c_3^1 & c_3^0 & c_3^2 z_3 
      \end{vmatrix*}
      &
      \begin{vmatrix*}
        c_1^1 & c_1^0 & c_1^1 z_1 \\
        c_2^1 & c_2^0 & c_2^1 z_2 \\
        c_3^1 & c_3^0 & c_3^1 z_3 
      \end{vmatrix*}
      &
      \begin{vmatrix*}
        c_1^1 & c_1^0 & c_1^0 z_1 \\
        c_2^1 & c_2^0 & c_2^0 z_2 \\
        c_3^1 & c_3^0 & c_3^0 z_3 
      \end{vmatrix*}
    \end{vmatrix*}.
  \]
  \endgroup
\end{example}

The elementary properties of $\VanGenkl$ from \expref{Proposition}{prop.vangen.basic} extend as follows to $\SMVanOnekl$.
\begin{proposition}
  \label{prop.smvan1-basic}
  For any \(k,l\in \N\) and index sets \(S_1,\ldots,S_{l+1}\) as in \expref{Definition}{def.smvan1}, 
  \(\SMVanOnekl\) is skew-symmetric with respect to the order of the sets
  \(S_1,\ldots,S_{l+1}\), and the choice of order within each set, in that any permutation thereof changes the construction by merely multiplying by the sign of the permutation. For any order,
  \(\SMVanOnekl\) is nonzero, homogeneous of degree
  \(l+1\), and set-multilinear with respect to the partition \(S_1, \dots, S_{l+1}\), and
  every monomial consistent with the partitions appears with a nonzero coefficient.
  When the sets are ordered as \(S_1,\ldots,S_{l+1}\)
  and the variables associated with $S_i$ are labeled and ordered as
  \((x_{i,1},\ldots,x_{i,k+2})\) for \(i=1,\ldots,l+1\),
  the coefficient of \(x_{1,1}\cdot\cdots\cdot x_{l+1,1}\)
  equals
  \begin{equation}\label{eq.smvan.coef} 
  (-1)^{(k+1)(l+1)} \cdot 
    \begin{vmatrix*}
      a_{1,1}^l & \cdots & a_{1,1}^0 \\
      \vdots & \ddots & \vdots \\
      a_{l+1,1}^l & \cdots & a_{l+1,1}^0
    \end{vmatrix*}
    \;\cdot\;
    \prod_{i=1}^{l+1}
      \begin{vmatrix*}
        a_{i,2}^k & \cdots & a_{i,2}^0 \\
        \vdots & \ddots & \vdots \\
        a_{i,k+2}^k & \cdots & a_{i,k+2}^0
      \end{vmatrix*}.
  \end{equation}
\end{proposition}
\begin{proof}
  All assertions to be proved follow from elementary properties of
  determinants, that Vandermonde determinants are nonzero unless they
  have duplicate rows, and the following computation for the coefficient of \(x_{1,1}\cdot\cdots\cdot x_{l+1,1}\):
  Plug 1 into \(x_{i,1}\) for \(i=1,\ldots,l+1\) and 0 into the
  remaining variables, and minor expand along the last column each of the inner determinants. Due to the minor expansions, the elements in the $i$-th row of the outer determinant have a common factor of $(-1)^{k+1}$ times the $(k+1) \times (k+1)$ determinant for that value of $i$ in the product on the right-hand side of \eqref{eq.smvan.coef}. After removing those common factors from all $l+1$ rows, the remaining $(l+1) \times (l+1)$ outer determinant equals the determinant in the middle of \eqref{eq.smvan.coef}.
\end{proof}

The following theorem formalizes the role \(\SMVanOne\) plays among the degree-\((l+1)\) polynomials with respect to \(\VanIdeal{\RFEkl}\).
\begin{theorem}
  \label{thm.smvan1}
  Let \(k,l \in \N\)
  and let \(X_1, \ldots, X_{l+1}\) be \(l+1\) disjoint sets of
  indices (of any size).
  The linear span of $\SMVanOnekl[S_1,\dots,S_{l+1}]$, over all choices of
  $S_i\subseteq X_i$ with $|S_i|=k+2$, equals the set-multilinear 
  polynomials in $\VanIdeal{\RFEkl}$ with variable partition $(X_1, \dots, X_{l+1})$.
\end{theorem}
\expref{Theorem}{thm.smvan1} and \expref{Proposition}{prop.smvan1-basic} imply \expref{Corollary}{cor.rfe.partition-class-size} that there are no set-multilinear polynomials of degree $l+1$ in $\VanIdeal{\RFEkl}$ that have at least one partition $X_i$ of size less than $k+2$.

The proof of \expref{Theorem}{thm.smvan1} follows the same outline as the one of \expref{Theorem}{thm.ideal.generators} in \expref{Section}{s.gens}. We start by showing that all instantiations of $\SMVanOnekl$ are contained in $\VanIdeal{\RFEkl}$ using a similar argument as that for $\VanGenkl$.

\begin{lemma}
  \label{lemma.smvanone-vanishes}
  For every \(k,l \in \N\) and every choice of \(l+1\) disjoint sets
  \(S_1,\dots,S_{l+1}\) of \(k+2\) indices each,
  \(\SMVanOnekl[S_1,\dots,S_{l+1}]\) vanishes at \(\RFEkl\).
\end{lemma}
\begin{proof}
  Let \(g/h\) be a seed for \(\RFEkl\).
  Let \(A\) be the \((l+1)\times(l+1)\) outer matrix defining
  \(\SMVanOne\), so that \(\SMVanOne \doteq \det(A)\).
  Recall that the columns of \(A\) are indexed by \(d=l,\ldots,0\).
  Let \(\vec{h}\in\FF^{l+1}\) be the column vector where the row indexed by
  \(d\) is the coefficient of \(\alpha^d\) in \(h(\alpha)\).
  We show that, after substituting \(\RFEkl(g/h)\),
  the matrix-vector product \(A\vec{h}\in \FF^{l+1}\) yields the zero
  vector.
  It follows that evaluating \(\SMVanOne\) at \(\RFEkl(g/h)\) vanishes,
  as it is the determinant of a singular matrix.

  Fix \(i\in\{1,\ldots,l+1\}\), and focus on the \(i\)-th coordinate of
  \(A\vec{h}\).
  The \((i,d)\) entry of \(A\) is a determinant;
  let \(B_{i,d}\) be the inner matrix as in \expref{Definition}{def.smvan1}.
  As \(d\) varies, only the last column of \(B_{i,d}\) changes.
  Thus, by multilinearity of the determinant,
  the \(i\)-th entry of \(A\vec{h}\) is itself a determinant.
  Recalling that the rows of \(B_{i,l},\ldots,B_{i,0}\) are indexed by
  \(j\in X_i\),
  the \(j\)-th row of this determinant is
  \[
    \begin{bmatrix}
      a_j^k & \cdots & a_j^0 & h(a_j) x_j 
    \end{bmatrix}.
  \]
  After substituting \(\RFEkl(g/h)\), it becomes
  \[
    \begin{bmatrix}
      a_j^k & \cdots & a_j^0 & g(a_j) 
    \end{bmatrix}.
  \]
  Since \(g\) is a degree-\(k\) polynomial,
  the columns of $B_{i,d}$ are linearly dependent,
  so the determinant is zero.
\end{proof}

Next, we argue that every polynomial in $\VanIdeal{\RFEkl}$ that is set-multilinear with respect to the variable partition $(X_1, \dots, X_{l+1})$ is in the ideal $I$ generated by the instantiations of $\SMVanOnekl$ in the statement of \expref{Theorem}{thm.smvan1}. We use a similar two-step approach as for 
\expref{Theorem}{thm.ideal.generators} in \expref{Section}{s.gens}.
\begin{enumerate}
  \item Modulo the ideal $I$, every polynomial $p$ is equal to a polynomial $r$ (depending on $p$) with a certain structure (\expref{Lemma}{lemma.smvanone-multicoredform}).
  \item Every nonzero polynomial $r$ that has the structure and is is set-multilinear with respect to the variable partition $(X_1, \dots, X_{l+1})$ is hit by $\RFEkl$ (\expref{Lemma}{lemma.rfe-hits-multicored}).
\end{enumerate}

For step~1, we need a suitable replacement for being \((c,t)\)-cored.
The following adaptation to the set-multilinear setting
suffices.
\begin{definition}
  \label{def.smcored}
  Let \(X_1, \dots, X_d \subseteq [n]\) be disjoint sets of
  indices.
  A polynomial that is
  set-multilinear with respect to the partition $(X_1, \dots, X_d)$ is
  \emph{\((c,t)\)-multi-cored} if
  there exists a set $C\doteq C_1\sqcup\dots\sqcup C_d$, with $C_i\subseteq X_i$, $|C_i|\le c$,
  such that every monomial $m$ of the polynomial satisfies $|\supp(M)\setminus C|\le t$.
\end{definition}

We refer to the set $C$ in \expref{Definition}{def.smcored} as a multi-core.

\begin{lemma}
    \label{lemma.smvanone-multicoredform}
    Let $k,l\in\N$ and let $X_1, \dots, X_{l+1}\subseteq [n]$
    be disjoint sets of indices. Suppose $C\doteq C_1\sqcup\cdots\sqcup C_{l+1}$ is a set of indices
    such that $C_i\subseteq X_i$ and $|C_i| = k+1$. Let $I$ be the ideal generated by
    the polynomials $\SMVanOnekl[C_1\sqcup \{j_1\},\dots,C_{l+1}\sqcup \{j_{l+1}\}]$, 
    where $j_i\in X_i\setminus C_i$. Modulo $I$, every set-multilinear polynomial with respect to the variable partition
    \(X_1, \dots,X_{l+1}\) equals a \((k+1,l)\)-multi-cored polynomial with multi-core $C$.
\end{lemma}

\begin{proof}
    By linearity it suffices to establish the result for any 
    monomial $m$ that is set-multilinear with respect to the partition $(X_1, \dots, X_{l+1})$. 
    If $\supp(m) \cap C$ is nonempty, then $m$ is already $(k+1,l)$-multi-cored with multi-core $C$ because $m$ only has $l+1$ variables in its support. Otherwise, let $m=x_{j_1}\cdots x_{j_{l+1}}$. By \expref{Proposition}{prop.smvan1-basic}, the polynomial $\SMVanOnekl[C_1\sqcup\{j_1\},\dots, C_{l+1}\sqcup\{j_{l+1}\}]$ can be written as $c \cdot m + r$ where $c \in \FF$ is nonzero and $r$ is a linear combination of monomials $m'$ that are set-multilinear with respect to the partition $(X_1, \dots, X_{l+1})$ and such that $\supp(m') \cap C$ is nonempty. The result for $m$ follows by writing $m \equiv -c^{-1} \cdot r \bmod I$.
\end{proof}

Step 2 is another application of the Zoom Lemma. We make use of the version geared towards multilinear polynomials, namely \expref{Lemma}{lem.multi-zsub-deriv.mlin}.

\begin{lemma}
  \label{lemma.rfe-hits-multicored}
  Let \(k,l\in \N\) 
  and let \(X_1, \dots, X_{l+1} \subseteq [n]\) be disjoint sets of indices.
  Every nonzero polynomial that is set-multilinear with respect to
  the partition $(X_1,\dots,X_{l+1})$ and that is \((k+1,l)\)-multi-cored
  is hit by \(\RFEkl\).
\end{lemma}

\begin{proof}
  Let $r$ satisfy the hypotheses of the lemma with multi-core $C$. Let $m^*$ be a monomial in $r$ for which $\supp(m^*) \setminus C$ is maximal with respect to inclusion. Such a monomial exists because $r$ is nonzero. Let $j^* \in \supp(m^*) \cap C$. Such an index exists since $|\supp(m^*)|=l+1$ and $|\supp(m^*) \setminus C| \le l$ by the multi-core property. Let $i^* \in [l+1]$ be such that $j^* \in X_{i^*}$. Set $K \doteq C \cap X_{i^*} \setminus \{j^*\}$ and $L \doteq \supp(m^*) \setminus \{j^*\}$. Note that $|K| \le (k+1)-1=k$ and $|L| \le (l+1)-1=l$. By set-multilinearity, monomials $m$ in $r$ for which \(\PartialZero{L}{K}{m}\) is nonzero need to have the form $x_j \cdot x^L$ where $j \in X_{i^*} \setminus K$. The monomial $m^*$ is of the form with $j=j^*$. By the maximality of $m^*$, any monomial in $r$ of the form has to have $j \in C$. Since $C \cap (X_{i^*} \setminus K) = \{j^*\}$, it follows that $m^*$ is the only monomial in $r$ that contributes to $\PartialZero{L}{K}{r}$. Since $\PartialZero{L}{K}{m^*} = x_{j^*}$, it follows that 
  $\PartialZero{L}{K}{r}$ is nonzero upon the substitution
  \eqref{eq.multi-zsub-deriv.mlin.sub}. We conclude that \(\RFEkl\) hits $r$ by \expref{Lemma}{lem.multi-zsub-deriv.mlin}.
\end{proof}

We now have all ingredients to establish \expref{Theorem}{thm.smvan1}. 

\begin{proof}[Proof of \expref{Theorem}{thm.smvan1}]
    Let $\mathcal{S} \doteq (S_1, \dots, S_{l+1})$ range as in the statement. The linear span of the polynomials $\SMVanOnekl[\mathcal{S}]$ is set-multilinear with respect to the variable partition $(X_1, \dots, X_{l+1})$ by \expref{Proposition}{prop.smvan1-basic}, and in \VanIdeal{\RFEkl} by 
    \expref{Lemma}{lemma.smvanone-vanishes}. 
    In the other direction, the combination of \expref{Lemma}{lemma.smvanone-multicoredform} and \expref{Lemma}{lemma.rfe-hits-multicored} imply that every polynomial $p \in \VanIdeal{\RFEkl}$ that is set-multilinear with respect to the variable partition $(X_1, \dots, X_{l+1})$ falls inside the ideal $I$ generated by the polynomials $\SMVanOnekl[\mathcal{S}]$, \ie, $p = \sum_{\mathcal{S}} q_{\mathcal{S}} \SMVanOnekl[\mathcal{S}]$ for some polynomials $q_{\mathcal{S}}$.
    As all polynomials $\SMVanOnekl[\mathcal{S}]$ as well as $p$ are homogeneous of degree $l+1$, it follows that each $q_{\mathcal{S}}$ can be replaced by its constant term. 
\end{proof}

\section{Read-Once Oblivious Algebraic Branching Programs}
\label{s.constant-width-roabp}

In this section we provide some background on ROABPs and establish
\expref{Theorem}{thm.roabp}.

\subsection{Background}
\label{s.roabp.background}

Algebraic branching programs are a syntactic model for algebraic
computation.
One forms a directed graph with a designated source and sink.
Each edge is labeled by a polynomial that depends on at most one variable
among \(x_1,\ldots,x_n\).
The branching program computes a polynomial in \(\FF[x_1,\ldots,x_n]\) by
summing,
over all source-to-sink paths,
the product of the labels on the edges of each path.

A special subclass of algebraic branching programs are read-once oblivious
algebraic branching programs (ROABPs).
In this model, the vertices of the branching program are organized in layers.
The layers are totally ordered,
and edges exist only from one layer to the next.
For each variable,
there is at most one consecutive pair of layers between which that variable
appears,
and for each pair of consecutive layers,
there is at most one variable that appears between them.
In this way, every source-to-sink path reads each variable at most once
(the branching program is \emph{read-once}),
and the order in which the variables are read is common to all paths (the
branching program is \emph{oblivious}).
We can always assume that the number of layers equals one plus the number of
variables under consideration.

The number of vertices comprising a layer is called its \emph{width}.
The width of an ROABP is the largest width of its layers.
The minimum width of an ROABP computing a given polynomial can be
characterized in terms of the rank of coefficient matrices constructed as
follows.
\begin{definition}
  \label{def.var-split-poly-matrix}
  Let \(U\sqcup V = [n]\) be a partition of the variable indices,
  and let \({M}_U\) and \({M}_V\) be the sets of monomials
  that are supported on variables indexed by \(U\) and \(V\), respectively.
  For any polynomial \(p\in\FF[x_1,\ldots,x_n]\) define the matrix
  \[
    \CMat{U}{V}(p)
    \in
    \FF^{{M}_U \times {M}_V}
  \]
  by setting the \(({m}_U,{m}_V)\) entry to equal
  the coefficient of \({m}_U{m}_V\) in \(p\).
\end{definition}
\(\CMat{U}{V}(p)\) is formally an infinite matrix,
but it has only finitely many nonzero entries.
When \(p\) has degree at most \(d\),
one can just as well truncate \(\CMat{U}{V}(p)\) to include only rows and
columns indexed by monomials of degree at most \(d\).

\begin{lemma}[\cite{Nisan1991}]
  \label{lem.roabp-char}
  Let \(p\in\FF[x_1,\ldots,x_n]\) be any polynomial.
  There is an ROABP of width \(w\) computing \(p\) in the variable order
  \(x_1,\ldots,x_n\) if and only if,
  for every \(s \in \{0,\ldots,n\}\),
  with respect to the partition
  \(U = \{1,\ldots,s\}\) and \(V = \{{s+1},\ldots,n\}\),
  we have
  \[
    \rank( \CMat{U}{V}(p) ) \le w.
  \]
\end{lemma}

We group the monomials in $M_U$ and $M_V$ by their degrees
and order the groups by increasing degree.
This induces a block structure on $\CMat{U}{V}(p)$ with one block for every
choice of $r, c \in \N$;
the \((r,c)\) block is the submatrix with rows indexed by degree-\(r\)
monomials in \(M_U\) and columns indexed by degree-\(c\) monomials in \(M_V\).
In the case where $p$ is homogeneous,
the only nonzero blocks occur for $r+c$ equal to the degree of $p$.
In this case the rank of \(\CMat{U}{V}(p)\) is the sum of the ranks of its
blocks.

In general,
the rank of $\CMat{U}{V}(p)$ is at least the rank of
$\CMat{U}{V}(p^{(\min)})$,
where $p^{(\min)}$ denotes the homogeneous component of $p$ of the lowest degree,
$d_{\min}$.
This follows because the submatrix of $\CMat{U}{V}(p)$ consisting of the rows
and columns indexed by monomials of degree at most $d_{\min}$ has a block
structure that is triangular with the blocks of $\CMat{U}{V}(p^{(\min)})$ on
the hypotenuse.
The observation yields the following folklore consequence of
\expref{Lemma}{lem.roabp-char}.
\begin{proposition}
  \label{prop.roabp-homog}
  \def\pdn{p^{(\min)}}
  Let \(p\in\FF[x_1,\ldots,x_n]\) be any nonzero polynomial,
  and let \(\pdn\) be the nonzero homogeneous component of \(p\) of least degree.
  If \(p\) can be computed by an ROABP of width \(w\),
  then so can \(\pdn\).
\end{proposition}

\subsection{Hitting property / lower bound}
\label{s.roabp.hitting-lower}

We now prove the ROABP hitting property of $\SV$ given in \expref{Theorem}{thm.roabp} and the equivalent ROABP lower bound given in \expref{Theorem}{thm.roabp.lb}. Both theorems follow from the next statement in a standard way. 
\begin{theorem}
  \label{thm.application.roabp.hardpart}
  For any integer $l \ge 1$,
  every nonzero multilinear homogeneous polynomial of degree $l+1$ in the vanishing ideal of $\SVl$ requires ROABP width at least $(l/3)+1$.
\end{theorem} 
For completeness, before proving \expref{Theorem}{thm.application.roabp.hardpart}, we argue how our ROABP hitting property and lower bound follow.
\begin{proof}[Proof of \expref{Theorem}{thm.roabp} and \expref{Theorem}{thm.roabp.lb}]
  The theorems are equivalent by complementation. We explain how \expref{Theorem}{thm.roabp} follows from \expref{Theorem}{thm.application.roabp.hardpart}.

  \def\pdn{p^{(\min)}}
  Fix \(p\) satisfying the hypotheses of \expref{Theorem}{thm.roabp}.
  We show that \(\RFElml\) hits \(p\);
  this implies \(\SVl\) hits \(p\) because \(\RFElml\) and \(\SVl\) are
  equivalent up to variable rescaling, and rescaling variables does not affect
  ROABP width.

  If $p$ contains a monomial depending on at most $l$ variables, then \expref{Proposition}{prop.ideal.membership.degree-bound} implies that $\RFElml$ hits $p$. The remaining case is when the homogeneous component $\pdn$ of the least degree is multilinear of degree $l+1$. By \expref{Proposition}{prop.roabp-homog}, $\pdn$ has ROABP width less than $(l/3)+1$. By \expref{Theorem}{thm.application.roabp.hardpart}, $\pdn$ is hit by $\RFElml$, and by \expref{Proposition}{prop.rfe-homog} so is $p$.
\end{proof}

In the remainder of this section we establish
\expref{Theorem}{thm.application.roabp.hardpart}. We do not try to optimize the dependence of the bound on $l$.

Fix a positive integer \(l\),
and fix an arbitrary variable order, say \(x_1,\ldots,x_n\).
We show that, for every polynomial \(p\) that is nonzero, multilinear, and homogeneous of
degree \(l+1\),
and that belongs to the vanishing ideal of \(\RFElml\),
there exists some \(s\in\{0,\ldots,n\}\)
so that,
with respect to the partition
\(U=\{1,\ldots,s\}\), \(V = \{{s+1},\ldots,n\}\),
it holds that \(\rank(\CMat{U}{V}(p)) \ge (l/3)+1\).
\expref{Theorem}{thm.application.roabp.hardpart} then follows by
\expref{Lemma}{lem.roabp-char}.

Let $C \doteq \CMat{U}{V}(p)$.
As $p$ is homogeneous of degree $l+1$,
\(C\) is block diagonal,
with a block \(C_d\) for each \(d \in \{0,\ldots,l+1\}\)
consisting of the rows indexed by monomials of degree \(d\)
and the columns indexed by monomials of degree \(l+1-d\).
The block diagonal structure implies $\rank(C) = \sum_{d=0}^{l+1} \rank(C_d)$.

Via \expref{condition}{thm.ideal.membership.cond2} of \expref{Theorem}{thm.ideal.membership}, the hypothesis that \(p\) belongs to \(\VanIdeal{\RFElml}\) induces linear equations on the entries in the blocks \(C_d\). For homogeneous polynomials like $p$, the condition stipulates that for all disjoint subsets $K,L \subseteq [n]$ with $|K| = k = l-1$ and $|L| = l$, \(\PartialZero{L}{K}{p}\) vanishes at the point \eqref{eq.thm.ideal.membership.sub} with $z=1$. This is a linear equation in the coefficients of \(\PartialZero{L}{K}{p}\), which are entries in the blocks $C_d$ of $C$. In fact, each of these equations only reads entries from two \emph{adjacent} blocks, \ie, blocks $C_d$ and $C_{d'}$ with \(|d - d'| = 1\). This is because \(L\) has size \(l\), one less than the degree of \(p\), so the only monomials that contribute to $\PartialZero{L}{K}{p}$ are those that are one variable \(x_i\) times the product of the variables indexed by \(L\). It follows that the corresponding linear equation on \(C\) reads only entries that reside in the blocks $C_{|L \cap U|+1}$ (for $i \in U$) and $C_{|L\cap U|}$ (for $i \in V$).

We exploit the structure of these equations and argue that, for an appropriate choice of the partition index $s$, $\rank(C)$ is high. 

\paragraph{Ingredients.}
Our analysis has four ingredients.
The first ingredient is the fact that $\rank(C)$ is at least the number of nonzero blocks $C_d$. This is because a nonzero block has rank at least 1, and $\rank(C)$ is the sum of the ranks of the blocks. This simple observation means we can focus on situations where relatively few of the blocks are nonzero.

The second ingredient establishes an alternative lower bound on \(\rank(C)\) in terms of the minimum distance between a nonzero block \(C_d\) and either extreme (\(d=0\) or \(d=l+1\)). Another way to think about this distance is as the maximum \(\minsplitdeg\) such that every monomial in \(p\) depends on at least \(\minsplitdeg\) variables indexed by \(U\) and at least \(\minsplitdeg\) variables indexed by \(V\).

\begin{lemma}
  \label{lem.application.roabp.core}
  Let \(p\in\VanIdeal{\RFEparam{l-1}{l}}\) be nonzero, multilinear, and
  homogeneous of degree \(l+1\),
  let \(U\sqcup V\) be a partition of \([n]\),
  and let \(C \doteq \CMat{U}{V}(p)\).
  If every monomial in p depends on at least $\minsplitdeg$ variables indexed
  by $U$ and at least $\minsplitdeg$ variables indexed by $V$,
  then \(\rank(C) \ge \minsplitdeg+1\).
\end{lemma}

The proof involves revisiting the equations from \expref{condition}{thm.ideal.membership.cond2} of \expref{Theorem}{thm.ideal.membership} and modifying%
\footnote{This is a setting where we exploit the possibility of the sets $K$ and $L$ in the \hyperref[lem.multi-zsub-deriv]{Zoom Lemma} to overlap.}
the underlying instantiations of \expref{Lemma}{lem.multi-zsub-deriv} to obtain a system of linear equations with a simple enough structure that we can analyze.

The remaining ingredients allow us to reduce to situations where either the
first or second ingredient applies.
The third ingredient lets us fix any two zero blocks
and zero out all the blocks that are not between them.

\begin{proposition}\label{prop:ROABP:zero}
  Let $p\in\VanIdeal{\RFEparam{l-1}{l}}$ be multilinear
  and homogeneous of degree $l+1$.
  Let $U\sqcup V$ be a partition of $[n]$,
  and let $C \doteq \CMat{U}{V}(p)$.
  Suppose that for some $d_1,d_2 \in \{-1,\dots,l+2\}$ with $d_1 \le d_2$,
  we have $C_{d_1} = 0$ and $C_{d_2} = 0$,
  where \(C_{-1} \doteq 0\) and \(C_{l+2} \doteq 0\).
  Let $p'$ be the polynomial obtained from $p$ by zeroing out the blocks
  $C_d$ with $d < d_1$ or $d > d_2$.
  Then $p'$ belongs to $\VanIdeal{\RFElml}$.
\end{proposition}

As zeroing out blocks does not increase the rank of \(C\),
our lower bound for \(\rank(C)\) reduces to the same lower bound for the rank
of \(\CMat{U}{V}(p')\).
This effectively extends the scope of the second ingredient:
Alone, the second ingredient requires that \emph{all} nonzero blocks of \(C\)
be far from the extremes;
with the third ingredient,
it suffices that there exists a subinterval of nonzero blocks that is surrounded by zero blocks and that is far from the extremes.
The proof hinges on the adjacent-block property of the equations from \expref{condition}{thm.ideal.membership.cond2} of \expref{Theorem}{thm.ideal.membership}.

The ingredients thus far suffice provided there exists a nonzero block far from the extremes:
Such a block belongs to some subinterval of nonzero blocks that is surrounded by zero blocks, say $C_{d_1}$ to the left and $C_{d_2}$ to the right, and the subinterval either is large and therefore has many nonzero blocks such that the first ingredient applies, or else it is small and therefore stays far from the extremes such that the combination of the second
and third ingredients applies. See \expref{Figure}{fig.roabp} for an illustration.
The fourth and final ingredient lets us ensure there is a nonzero block far from the extremes by setting the partition index $s$ appropriately.
In fact, it lets us guarantee a zero-to-nonzero transition at a position of our choosing.

\begin{proposition} \label{prop:ROABP:dselect}
  For every $d \in \{-1,\dots,l\}$, there is $s \in \{0,\dots,n\}$ such that $C_d=0$ and $C_{d+1} \ne 0$
  with respect to the partition \(U=\{1,\ldots,s\}\),
  \(V=\{{s+1},\ldots,n\}\), where $C_{-1} \doteq 0$.
\end{proposition}

\paragraph{Combining ingredients.}
Let us find out what lower bound on $\rank(C)$ the prior ingredients give us as a function of the position $d = d_1$ in the interval where we have a guaranteed zero-to-nonzero transition as in \expref{Proposition}{prop:ROABP:dselect}. Starting from position $d_1$, keep increasing the position index until we hit the next zero block, say at position $d_2$, where we use $C_{l+2} \doteq 0$ as a sentinel. See \expref{Figure}{fig.roabp}.
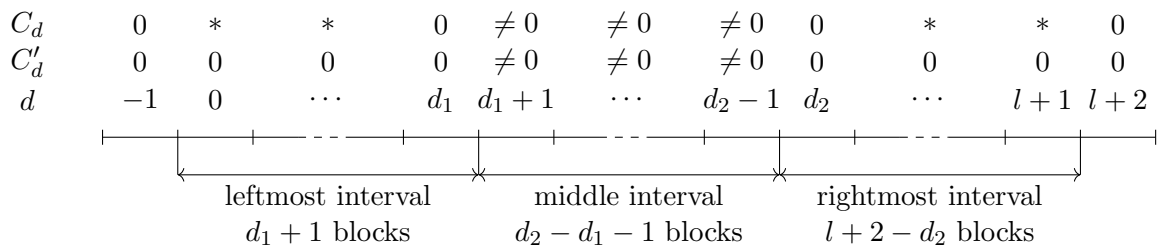
\begin{figure}
\begin{center}
\begin{tikzpicture}[scale=1]
  \draw (-1,-1.5) -- (1.6,-1.5);
  \draw[dashed] (1.6,-1.5) -- (2.25,-1.5);
  \draw (2.25,-1.5) -- (5.6,-1.5);
  \draw[dashed] (5.6,-1.5) -- (6.25,-1.5);
  \draw (6.25,-1.5) -- (9.6,-1.5);
  \draw[dashed] (9.6,-1.5) -- (10.25,-1.5);
  \draw (10.25,-1.5) -- (13,-1.5);  
  \draw (-1,-1.6) -- +(0,.2);
  \draw (0,-2) -- +(0,.6);
  \draw (1,-1.6) -- +(0,.2);
  \draw (3,-1.6) -- +(0,.2);
  \draw (4,-2) -- +(0,.6);
  \draw (5,-1.6) -- +(0,.2);
  \draw (7,-1.6) -- +(0,.2);
  \draw (8,-2) -- +(0,.6);
  \draw (9,-1.6) -- +(0,.2);
  \draw (11,-1.6) -- +(0,.2);
  \draw (12,-2) -- +(0,.6);
  \draw (13,-1.6) -- +(0,.2);
  \node at ($(-2, 0.0) + (0,-1)$) {$d$};
  \node at ($(-0.5, 0.0) + (0,-1)$) {$-1$};
  \node at ($(0.5, 0.0) + (0,-1)$) {$0$};
  \node at ($(2.0, 0.0) + (0,-1)$) {$\cdots$};
  \node at ($(3.5, 0.0) + (0,-1)$) {$d_1$};
  \node at ($(4.5, 0.0) + (0,-1)$) {$d_1+1$};
  \node at ($(6.0, 0.0) + (0,-1)$) {$\cdots$};
  \node at ($(7.5, 0.0) + (0,-1)$) {$d_2-1$};
  \node at ($(8.5, 0.0) + (0,-1)$) {$d_2$};
  \node at ($(10.0, 0.0) + (0,-1)$) {$\cdots$};
  \node at ($(11.5, 0.0) + (0,-1)$) {$l+1$};
  \node at ($(12.5, 0.0) + (0,-1)$) {$l+2$};
  \node at ($(-2, 0.0) + (0,0)$) {$C_d$};
  \node at ($(-0.5, 0.0) + (0,0)$) {$0$};
  \node at ($(0.5, 0.0) + (0,0)$) {$*$};
  \node at ($(2.0, 0.0) + (0,0)$) {$*$};
  \node at ($(3.5, 0.0) + (0,0)$) {$0$};
  \node at ($(4.5, 0.0) + (0,0)$) {$\ne 0$};
  \node at ($(6.0, 0.0) + (0,0)$) {$\ne 0$};
  \node at ($(7.5, 0.0) + (0,0)$) {$\ne 0$};
  \node at ($(8.5, 0.0) + (0,0)$) {$0$};
  \node at ($(10.0, 0.0) + (0,0)$) {$*$};
  \node at ($(11.5, 0.0) + (0,0)$) {$*$};
  \node at ($(12.5, 0.0) + (0,0)$) {$0$};
  \node at ($(-2, 0.0) + (0,-.5)$) {$C'_d$};
  \node at ($(-0.5, 0.0) + (0,-.5)$) {$0$};
  \node at ($(0.5, 0.0) + (0,-.5)$) {$0$};
  \node at ($(2.0, 0.0) + (0,-.5)$) {$0$};
  \node at ($(3.5, 0.0) + (0,-.5)$) {$0$};
  \node at ($(4.5, 0.0) + (0,-.5)$) {$\ne 0$};
  \node at ($(6.0, 0.0) + (0,-.5)$) {$\ne 0$};
  \node at ($(7.5, 0.0) + (0,-.5)$) {$\ne 0$};
  \node at ($(8.5, 0.0) + (0,-.5)$) {$0$};
  \node at ($(10.0, 0.0) + (0,-.5)$) {$0$};
  \node at ($(11.5, 0.0) + (0,-.5)$) {$0$};
  \node at ($(12.5, 0.0) + (0,-.5)$) {$0$};
  \draw[<->] (0,-2) -- node[auto,swap,align=center] {leftmost interval \\ $d_1+1$ blocks} (4.0,-2);
  \draw[<->] (4.0,-2) -- node[auto,swap,align=center] {middle interval \\ $d_2-d_1-1$ blocks} (8.0,-2);
  \draw[<->] (8.0,-2) -- node[auto,swap,align=center] {rightmost interval  \\ $l+2-d_2$ blocks} (12.0,-2);
\end{tikzpicture}
\end{center}
\caption{Rank lower bound analysis in terms of the blocks $C_d$ of $p$ and $C'_d$ of $p'$ (\expref{Proposition}{prop:ROABP:zero})}\label{fig.roabp}
\end{figure}
\begin{enumerate}
  \item By the first ingredient, since the middle interval consists of nonzero blocks only, $\rank(C) \ge d_2-d_1-1$.
  \item By the combination of the second and the third ingredient, we have that $\rank(C) \ge \minsplitdeg+1$
    where $\minsplitdeg = \min(d_1+1,l+2-d_2)$ is the minimum length of the leftmost and rightmost intervals.
    Indeed, let $p'$ be the polynomial obtained from $p$ by zeroing out the blocks $C_d$ with $d < d_1$ or $d > d_2$. By \expref{Proposition}{prop:ROABP:zero} $p' \in \VanIdeal{\RFEparam{l-1}{l}}$. The polynomial $p'$ is nonzero as it contains the original block $C_{d+1}$, which is nonzero. It is homogeneous of degree $l+1$ and multilinear as all of its monomials also occur in the homogeneous multilinear polynomial $p$ of degree $l+1$. By construction, every monomial in $p'$ contains at least $d_1+1$ variables indexed by $U$, and at least $l+2-d_2$ variables indexed by $V$. As such, $p'$ satisfies the conditions of \expref{Lemma}{lem.application.roabp.core} with $\minsplitdeg = \min(d_1+1,l+2-d_2)$. It follows that $\rank(C) \ge \rank(\CMat{U}{V}(p')) \ge \minsplitdeg+1$.
\end{enumerate}
If the rightmost interval has length at least the leftmost interval ($l+2-d_2 \ge d_1+1$),
then item 2 yields $\rank(C) \ge d_1+2$.
Otherwise, the rightmost interval is strictly shorter than the leftmost interval
($d_1+1 > l+2-d_2$);
this implies that the middle interval has length at least \(l - 2d_1 + 1\),
which by item 1 yields $\rank(C) \ge l-2d_1+1$.
In any case, the bound $\rank(C) \ge \min(d_1+2,l-2d_1+1)$ holds.
Taking $d_1= \floor{\frac{l-1}{3}}$ optimizes this expression,
achieving $\rank(C) \ge \floor{\frac{l-1}{3}} + 2 \ge (l/3)+1$.
This completes the proof of \expref{Theorem}{thm.application.roabp.hardpart} modulo the
proofs of ingredients two through four.

\paragraph{Proofs.}
We conclude by proving ingredients two through four.
We start with the one that requires the least specificity (ingredient 4, \expref{Proposition}{prop:ROABP:dselect}), then do ingredient 3 (\expref{Proposition}{prop:ROABP:zero}), and end with the one that involves the most structure (ingredient 2, \expref{Lemma}{lem.application.roabp.core}).

\begin{proof}[Proof of \expref{Proposition}{prop:ROABP:dselect}]
  When $s=0$, $C_0$ contains all entries. 
  As $s$ increases by $1$,
  some entries move from their current block $C_{d'}$ to the next block
  $C_{d'+1}$.
  Finally, when $s=n$, $C_{l+1}$ contains all entries.
  For $d \ge 0$, it follows that every nonzero entry moves from $C_d$ to $C_{d+1}$ at some time. If we stop increasing $s$ right after the last nonzero entry of $C$ moves
  out of \(C_d\), we have $C_d=0$ and $C_{d+1} \ne 0$. For $d=-1$, we can pick $s=0$ as $C_{-1}=0$ and $C_0=C \ne 0$. 
\end{proof}

\begin{proof}[Proof of \expref{Proposition}{prop:ROABP:zero}]
  It suffices to show that whenever $p$ satisfies the two conditions in
  \expref{Theorem}{thm.ideal.membership}, then so does $p'$. Both $p$ and $p'$ are homogeneous. 
  \expref{Condition}{thm.ideal.membership.cond1} holds for $p'$ as
  \(p'\) either is zero or else has the same degree as \(p\).
  Regarding \expref{condition}{thm.ideal.membership.cond2}, as
  mentioned, the condition is equivalent to a system of homogeneous
  linear equations on \(C' \doteq \CMat{U}{V}(p')\), each involving
  only an adjacent pair of blocks in $C'$.
  Those that involve only blocks $C'_d$ with $d \le d_1$ are met as the
  equations are homogeneous and the involved blocks are all zero.
  The same holds for the equations that involve only blocks $C'_d$
  with $d \ge d_2$.
  The remaining equations involve only blocks $C'_d$ with
  $d \in \{d_1, \dots, d_2\}$,
  on which $p$ and $p'$ agree.
  As the equations hold for \(C\), they also hold for $C'$.
\end{proof}

It remains to argue \expref{Lemma}{lem.application.roabp.core}. Our proof makes use of linear equations that are closely related to those given by \expref{Theorem}{thm.ideal.membership}, which in turn come from the \hyperref[lem.multi-zsub-deriv]{Zoom Lemma}. We revisit the application of the \hyperref[lem.multi-zsub-deriv]{Zoom Lemma} so as to obtain a simpler coefficient matrix---ultimately a Cauchy matrix---that enables a deeper analysis. To facilitate the discussion, we utilize the following notation. As $p$ is multilinear, we only need to consider rows indexed by monomials of the form $\prod_{i \in I} x_i$ for $I \subseteq U$ and columns indexed by monomials of the form $\prod_{j \in J} x_j$ for $J \subseteq V$. This allows us to index rows by subsets $I \subseteq U$ and columns by subsets $J \subseteq V$. For $I \subseteq U$ and $J \subseteq V$ we denote by $C(I,J)$ the corresponding entry of $C$. The following proposition describes the linear equations we use.
\begin{proposition}
  \label{prop.application.roabp.coef-eqns}
  Let \(p\in\VanIdeal{\RFEparam{l-1}{l}}\) be multilinear,
  and homogeneous of degree \(l+1\),
  let \(U\sqcup V\) be a partition of \([n]\),
  and let $C \doteq \CMat{U}{V}(p)$.
  For every \(I\subseteq U\) and \(J\subseteq V\)
  with \(|I|+|J|=l\),
  and for every \(j^*\in I\cup J\),
  \begin{equation}
    \label{eq.application.roabp.coef-eqns}
    \sum_{i\in U\setminus I}
      \frac{C(\{i\}\cup I, J)}%
        {a_i - a_{j^*}}
    +
    \sum_{i\in V\setminus J}
      \frac{C(I,\{i\}\cup J)}%
        {a_i - a_{j^*}}
    =0.
  \end{equation}
\end{proposition}
\begin{proof}
  Set
  $L \doteq I \cup J$ and
  $K \doteq L \setminus \{j^*\}$,
  and note that \(K \subseteq L\).
  Let \(d^* \in \N^L\) be the all-1 degree pattern with domain \(L\),
  and let \({m}^* \doteq \prod_{i\in L} x_i\)
  be the monomial supported on \(L\) that matches \(d^*\).
  As $p$ is multilinear,
  $d^*$ is $(K,L)$-extremal in $p$.
  Since \(p\) is in \(\VanIdeal{\RFEparam{l-1}{l}}\),
  the contrapositive of the \hyperref[lem.multi-zsub-deriv]{Zoom Lemma} tells us that the
  coefficient \(p_{d^*}\) of \(p\) vanishes at the point
  \eqref{eq.multi-zsub-deriv.sub} with $z=1$.

  The multilinear monomials $m$ of degree $l+1$ that match
  \(d^*\) have the form $m = x_i \cdot {m}^*$,
  where $i \in [n] \setminus L$.
  Thus, we can write the coefficient \(p_{d^*}\) as
  \begin{align}\label{eq:application:projection}
    p_{d^*}
    & =
    \sum_{i \in U\setminus I} C(\{i\} \cup I,J) \cdot x_i +
    \sum_{i \in V\setminus J} C(I,\{i\} \cup J) \cdot x_i.
  \end{align}
  For each \(i \in [n]\setminus L\),
  \eqref{eq.multi-zsub-deriv.sub} with $z=1$ substitutes $1/(a_i - a_{j^*})$
  into \(x_i\).
  Plugging this into \eqref{eq:application:projection} yields
  \eqref{eq.application.roabp.coef-eqns}.
\end{proof}

\begin{proof}[Proof of \expref{Lemma}{lem.application.roabp.core}]
  The proof goes by induction on \(\minsplitdeg\).
  The base case is \(\minsplitdeg = 0\), where the lemma holds because the rank of a nonzero matrix is always at least 1.
  For the inductive step, where \(\minsplitdeg \ge 1\), we zoom in on the contributions of the monomials that contain a particular variable. More precisely, for $j^* \in [n]$,
  let $p_{j^*}$ denote the partial derivative $p_{j^*} \doteq \partial_{x_{j^*}} p$.
  Consider any $j^* \in [n]$ such that $p_{j^*}$ is nonzero.
  As $p$ is multilinear and homogeneous of degree $l+1$,
  \(p_{j^*}\) is multilinear and homogeneous of degree \(l\). As every monomial in $p$ depends on at least $\minsplitdeg$ variables indexed by \(U\) and at least \(\minsplitdeg\) variables indexed by \(V\), every monomial in \(p_{j^*}\) depends on at least \(\minsplitdeg-1\) variables indexed by \(U\) and at least \(\minsplitdeg-1\) variables indexed by \(V\).
  In a moment, we argue that for every \(j^* \in [n]\),
  \(p_{j^*}\in\VanIdeal{\RFEparam{l-2}{l-1}}\).
  Then we will show the following:
  \begin{claim}\label{claim.existence}
    There exists $j^* \in [n]$ such that $p_{j^*} \ne 0$ and
    \begin{equation} \label{eq.claim.application.roabp.core}
      \rank(\CMat{U}{V}(p)) \ge \rank(\CMat{U}{V}(p_{j^*})) + 1.
    \end{equation}
  \end{claim}
  Given \(j^*\) as in \expref{Claim}{claim.existence},
  we conclude by induction that
  \[
    \rank(\CMat{U}{V}(p))
    \ge \rank(\CMat{U}{V}(p_{j^*})) + 1
    \ge (\minsplitdeg-1) + 1 + 1
    = \minsplitdeg + 1.
  \]

  To see that \(p_{j^*}\) belongs to the vanishing ideal of
  \(\RFEparam{l-2}{l-1}\),
  we use \expref{Theorem}{thm.ideal.membership}. Note that \(p_{j^*}\) is homogeneous, just like $p$.
  \expref{Condition}{thm.ideal.membership.cond1} of
  \expref{Theorem}{thm.ideal.membership} is satisfied by \(p_{j^*}\)
  since it it is satisfied by \(p\),
  and all of $k$, $l$, $n$, and the degree of $p_{j^*}$ are one less.
  Given \(K\) and \(L\) as in
  \expref{condition}{thm.ideal.membership.cond2} of
  \expref{Theorem}{thm.ideal.membership},
  we have
  \begin{equation}\label{eq.application.roabp.core.induct}
    \PartialZero{L}{K}{p_{j^*}}
    =
    p_{d^*}
  \end{equation}
  where \(d^*\) is the degree pattern with domain \(K \cup L \cup \{j^*\}\)
  that has \(d^*_j = 1\) for \(j \in L\cup\{j^*\}\)
  and \(d^*_j = 0\) for \(j \in K\).
  Since \(p \in \VanIdeal{\RFElml}\),
  the contrapositive of the \hyperref[lem.multi-zsub-deriv]{Zoom Lemma} applied to \(p\) with
  \(K'=K\cup\{j^*\}\),
  \(L'=L\cup\{j^*\}\),
  \(d^*\),
  says that \eqref{eq.application.roabp.core.induct} is zero upon the substitution 
  \eqref{eq.thm.ideal.membership.sub}.
  So \(p_{j^*}\in\VanIdeal{\RFEparam{l-2}{l-1}}\) by
  \expref{Theorem}{thm.ideal.membership}.
  This concludes the proof of \expref{Lemma}{lem.application.roabp.core} modulo the proof of \expref{Claim}{claim.existence}.
  \end{proof}

  \begin{proof}[Proof of \expref{Claim}{claim.existence}]
  Let \(U'\subseteq U\) be the indices of variables \(x_i\) such that \(p\)
  depends on \(x_i\),
  and similarly define \(V'\subseteq V\).
  We first consider the possibility that
  \eqref{eq.claim.application.roabp.core} fails for every
  \(j^* \in V'\).
  We show that this can only happen when \(|V'| < |U'|\).
  A symmetric argument shows that if \eqref{eq.claim.application.roabp.core}
  fails for all \(j^* \in U'\),
  then it must be that \(|U'| < |V'|\).
  As both inequalities cannot simultaneously occur,
  this guarantees the existence of the desired \(j^*\).

  Suppose that \eqref{eq.claim.application.roabp.core} fails
  for each \(j^* \in V'\).
  Observe that the column of \(\CMat{U}{V}(p_{j^*})\) corresponding to a
  monomial $m$ equals the column of \(\CMat{U}{V}(x_{j^*} p_{j^*})\)
  corresponding to the monomial $x_{j^*} m$;
  all other columns of \(\CMat{U}{V}(x_{j^*} p_{j^*})\) are zero.
  The matrix \(\CMat{U}{V}(x_{j^*} p_{j^*})\) can also be formed from
  \(\CMat{U}{V}(p)\) by zeroing out all the columns indexed by subsets that do
  not contain $j^*$
  (corresponding to multilinear monomials not involving \(x_{j^*}\)).
  The failure of \eqref{eq.claim.application.roabp.core} for \(j^*\) implies
  that \(\CMat{U}{V}(p_{j^*})\) has the same rank as \(\CMat{U}{V}(p)\),
  which is to say that the columns of \(\CMat{U}{V}(p)\) indexed by subsets
  that contain $j^*$ span \emph{all} the columns of \(\CMat{U}{V}(p)\).
  Going block by block,
  this implies that for every block \(C_d\) of \(C = \CMat{U}{V}(p)\),
  the columns within \(C_d\) that are indexed by subsets containing $j^*$ span
  all the columns of \(C_d\).
  This goes for every \(j^* \in V'\),
  as we are assuming that \eqref{eq.claim.application.roabp.core} fails for
  all of them.

  Let \(d\) be minimal such that \(C_d\ne 0\),
  \ie, such that \(p\) has a monomial depending on exactly \(d\) variables
  indexed by \(U\).
  We have \(d \ge \minsplitdeg \ge 1\) and \(C_{d-1} = 0\).
  The entries of \(C_d\) appear in the linear equations
  \eqref{eq.application.roabp.coef-eqns} given in
  \expref{Proposition}{prop.application.roabp.coef-eqns},
  either with entries from \(C_{d-1}\) or from \(C_{d+1}\).
  Since \(C_{d-1}\) is zero,
  the equations involving \(C_{d-1}\) and \(C_d\) simplify to equations on
  \(C_d\) only.
  Namely, for every \(I\subseteq U\) with \(|I|=d-1\),
  every \(J\subseteq V\) with \(|J|=l-(d-1)\),
  and every \(j^* \in I\cup J\),
  equation \eqref{eq.application.roabp.coef-eqns} simplifies to
  \begin{equation}\label{eq.application.roabp.coef-eqn-spec-pre}
    \sum_{i\in U\setminus I} \frac{C_d(\{i\}\cup I, J)}{a_i - a_{j^*}} = 0.
  \end{equation}
  For any fixed $i \in U \setminus U'$,
  all entries of the form $C_d(\{i\}\cup I, J)$ are zero.
  Thus, we can restrict the range of $i$ in
  \eqref{eq.application.roabp.coef-eqn-spec-pre}
  from $U \setminus I$ to $U' \setminus I$:
  \begin{equation}\label{eq.application.roabp.coef-eqn-spec-mid}
    \sum_{i\in U'\setminus I} \frac{C_d(\{i\}\cup I, J)}{a_i - a_{j^*}} = 0.
  \end{equation}

  Since $C_d \ne 0$,
  there is at least one fixed $I$ for which not all entries of the form
  $C_d(\{i\} \cup I, J)$ are zero as \(i\) and \(J\) vary.
  Let $I^*$ be such an $I$,
  and let $C^*_d$ denote the submatrix of $C_d$
  that consists of all entries of the form $C_d(\{i\}\cup I^*,J)$
  as \(i\) and \(J\) vary.
  For every $J \subseteq V$ with $|J|=l-(d-1)$
  and every $j^* \in I^* \cup J$,
  we have
  \begin{equation}\label{eq.application.roabp.coef-eqn-spec}
    \sum_{i\in U'\setminus I^*}
      \frac{C^*_d(\{i\}\cup I^*, J)}{a_i - a_{j^*}}
    = 0.
  \end{equation}
  For each $j^* \in V'$,
  consider the equations \eqref{eq.application.roabp.coef-eqn-spec}
  where $J$ ranges over all subsets of $V$ of size $|J| = l-(d-1)$
  that contain $j^*$.
  Observe that the coefficients $\frac{1}{a_i - a_{j^*}}$ in
  \eqref{eq.application.roabp.coef-eqn-spec} are independent of the choice of
  $J$.
  We argued that the columns of $C_d$ indexed by subsets $J$ that contain
  $j^*$ span all columns of $C_d$.
  The same holds for $C^*_d$,
  as $C^*_d$ is obtained from $C_d$ by removing rows.
  It follows that \eqref{eq.application.roabp.coef-eqn-spec} holds for
  \emph{every} subset $J$ of $V$ of size $l-(d-1)$
  (not just the ones containing $j^*$).

  In particular, consider any one nonzero column of \(C^*_d\).
  The column represents a nontrivial solution to the homogeneous system
  \eqref{eq.application.roabp.coef-eqn-spec} of $|V'|$ linear equations
  (one for each choice of $j^* \in V'$) in $|U' \setminus I^*|$ unknowns
  (one for each $i \in U' \setminus I^*$).
  The coefficient matrix $[ \frac{1}{a_i - a_{j^*}} ]$ is a Cauchy matrix,
  which is well-known to have full rank.
  In order for there to be a nontrivial solution,
  the number of equations must be strictly less than the number of unknowns.
  In other words,
  we have \(|V'| < |U'\setminus I^*| \le |U'|\),
  as desired.
\end{proof}

\section{Alternating Algebra Representation}
\label{s.simplicial-repr}

In this section we present in greater detail the alternating
algebra-based representation of (multilinear) polynomials suited to studying the vanishing ideal of \(\RFE\).
\expref{Subsection}{s.simplicial-repr.base-case} expands the informal discussion from the overview,
describing the representation and characterization for
the setting when \(l=1\), \(k=0\), and degree $d=2$.
\expref{Subsection}{s.simplicial-repr.alternating-algebra} provides a brief
introduction to alternating algebra suited to our purpose. \expref{Subsection}{s.simplicial-repr.gen-case}
formalizes the discussion from \expref{Subsection}{s.simplicial-repr.base-case} and
extends it to the case of multilinear polynomials for general \(k\), \(l\), and \(d\).

\subsection{Basic case}
\label{s.simplicial-repr.base-case}

For the purposes of this subsection, we fix the parameters \(k=0\), \(l=1\), and \(d=2\).
That is to say, we are studying which degree-2 polynomials  belong to the vanishing ideal for \(\RFEparam{0}{1}\). 

In \expref{Theorem}{thm.ideal.generators},
we proved that the polynomials \(\VanGenparam{0}{1}[i_1, i_2, i_3]\) as
\(i_1,i_2,i_3\) range over \([n]\) generate \(\VanIdeal{\RFEparam{0}{1}}\).
As these generators are all homogeneous degree-2 polynomials,
a degree-2 polynomial $p$ is in the ideal if and only if it is a linear
combination of instantiations of \(\VanGenparam{0}{1}\). 

Consider the generator when expanded as a linear combination of monomials:
\begin{align*}
  \VanGenparam{0}{1}[i_1, i_2, i_3]
  &=
  \begin{vmatrix}
    a_{i_1} & 1 \\
    a_{i_2} & 1
  \end{vmatrix}x_{i_1} x_{i_2}
  +
  \begin{vmatrix}
    a_{i_3} & 1 \\
    a_{i_1} & 1
  \end{vmatrix}x_{i_3} x_{i_1}
  +
  \begin{vmatrix}
    a_{i_2} & 1 \\
    a_{i_3} & 1
  \end{vmatrix}x_{i_2} x_{i_3}.
\end{align*}
We represent it graphically by creating a vertex $v_i \in \Vertices$ for each variable $x_i$, an
undirected edge for each monomial, and assigning to each edge a weight
equal to the coefficient of that monomial:
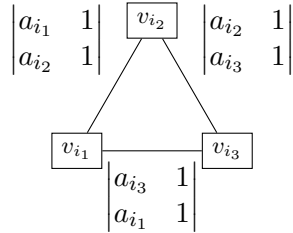
\begin{center}
  \begin{tikzpicture}
    \node[draw=black,rectangle] (x) at (0,0) {\footnotesize $v_{i_1}$};
    \node[draw=black,rectangle] (y) at (1,1.733) {\footnotesize $v_{i_2}$};
    \node[draw=black,rectangle] (z) at (2,0) {\footnotesize $v_{i_3}$};
    \draw[] (x) --
      node[auto]
        {\(\begin{vmatrix} a_{i_1} & 1 \\ a_{i_2} & 1\end{vmatrix}\)}
      (y);
    \draw[] (y) --
      node[auto]
        {\(\begin{vmatrix} a_{i_2} & 1 \\ a_{i_3} & 1\end{vmatrix}\)}
      (z);
    \draw[] (z) --
      node[auto]
        {\(\begin{vmatrix} a_{i_3} & 1 \\ a_{i_1} & 1\end{vmatrix}\)}
      (x);
  \end{tikzpicture}
\end{center}
Observe that the coefficient of \(x_{i_1}x_{i_2}\) has no dependence on
\(a_{i_3}\).
In particular, as \(i_3\) varies,
the coefficient of \(x_{i_1}x_{i_2}\) in \(\VanGenparam{0}{1}[i_1, i_2, i_3]\)
does not change.
In any other instantiation of \(\VanGenparam{0}{1}\) involving both
\(i_1\) and \(i_2\), the coefficient is either the same, or else differs by
a sign,
according to whether \(i_1\) or \(i_2\) precedes the other in the determinant.
A similar pattern holds with respect to all other monomials.
This suggests we can modify the graphical representation by rescaling the
weights on edges and suppress the dependence on the abscissas.
To capture the signs, we use oriented edges.
More precisely,
for each edge \(\{v_{i_1},v_{i_2}\}\),
we consider either of its two orientations,
say \(v_{i_1}\to v_{i_2}\),
and then divide its coefficient by
\(\begin{vmatrix}a_{i_1} & 1\\a_{i_2} & 1\end{vmatrix}\).
Note that considering the opposite orientation coincides with flipping the
sign of the scaling factor.
With these changes, \(\VanGenparam01[i_1, i_2, i_3]\) may be drawn
in any of the following ways
(among others).
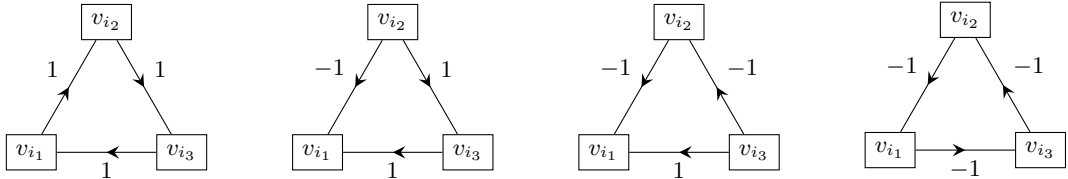
\begin{center}
  \begin{tikzpicture}
    \node[draw=black,rectangle] (x) at (0,0) {\footnotesize $v_{i_1}$};
    \node[draw=black,rectangle] (y) at (1,1.733) {\footnotesize $v_{i_2}$};
    \node[draw=black,rectangle] (z) at (2,0) {\footnotesize $v_{i_3}$};
    \draw[mid arrow] (x) -- node[auto] {\footnotesize $1$} (y);
    \draw[mid arrow] (y) -- node[auto] {\footnotesize $1$} (z);
    \draw[mid arrow] (z) -- node[auto] {\footnotesize $1$} (x);
  \end{tikzpicture}
  \hskip 1cm
  \begin{tikzpicture}
    \node[draw=black,rectangle] (x) at (0,0) {\footnotesize $v_{i_1}$};
    \node[draw=black,rectangle] (y) at (1,1.733) {\footnotesize $v_{i_2}$};
    \node[draw=black,rectangle] (z) at (2,0) {\footnotesize $v_{i_3}$};
    \draw[mid arrow] (y) -- node[auto,swap] {\footnotesize $-1$} (x);
    \draw[mid arrow] (y) -- node[auto] {\footnotesize $1$} (z);
    \draw[mid arrow] (z) -- node[auto] {\footnotesize $1$} (x);
  \end{tikzpicture}
  \hskip 1cm
  \begin{tikzpicture}
    \node[draw=black,rectangle] (x) at (0,0) {\footnotesize $v_{i_1}$};
    \node[draw=black,rectangle] (y) at (1,1.733) {\footnotesize $v_{i_2}$};
    \node[draw=black,rectangle] (z) at (2,0) {\footnotesize $v_{i_3}$};
    \draw[mid arrow] (y) -- node[auto,swap] {\footnotesize $-1$} (x);
    \draw[mid arrow] (z) -- node[auto,swap] {\footnotesize $-1$} (y);
    \draw[mid arrow] (z) -- node[auto] {\footnotesize $1$} (x);
  \end{tikzpicture}
  \hskip 1cm
  \begin{tikzpicture}
    \node[draw=black,rectangle] (x) at (0,0) {\footnotesize $v_{i_1}$};
    \node[draw=black,rectangle] (y) at (1,1.733) {\footnotesize $v_{i_2}$};
    \node[draw=black,rectangle] (z) at (2,0) {\footnotesize $v_{i_3}$};
    \draw[mid arrow] (y) -- node[auto,swap] {\footnotesize $-1$} (x);
    \draw[mid arrow] (z) -- node[auto,swap] {\footnotesize $-1$} (y);
    \draw[mid arrow] (x) -- node[auto,swap] {\footnotesize $-1$} (z);
  \end{tikzpicture}
\end{center}
While different choices of edge orientations lead to different illustrations,
any one illustration can be transformed into any other by
considering edges in opposite orientations as needed,
and flipping the sign of each associated coefficient.
By identifying each edge in one orientation with the negative of itself in the
opposite orientation,
we can view all the illustrations as renditions of the same underlying object.

In general, we can represent any degree-2 homogeneous multilinear
polynomial \(p\in\FF[x_1,\ldots,x_n]\) in a similar way:
For each monomial $x_{i_1}x_{i_2}$ create an oriented edge $v_{i_1}\to v_{i_2}$ and set the weight of the edge
to be the coefficient of \(x_{i_1}x_{i_2}\) in \(p\) divided by
\(\begin{vmatrix}a_{i_1} & 1\\ a_{i_2} & 1\end{vmatrix}\).
The representation determines the polynomial:
simply undo the scaling on each edge,
and read off a linear combination of monomials.
Note moreover that this graphical representation is linear in the polynomial:
adding or rescaling polynomials coincides with adding or rescaling
coefficients on the edges.

Observe that,
in every graphical representation of \(\VanGenparam01[i_1, i_2, i_3]\),
at every vertex, the sum of the coefficients on edges oriented out of
that vertex equals the sum of the coefficients on edges oriented in to
that vertex.
Indeed, we can interpret \(\VanGenparam01[i_1, i_2, i_3]\) as a
\emph{circulation} in which one unit of flow travels around a simple
3-cycle \(v_{i_1} \to v_{i_2}\to v_{i_3}\to v_{i_1}\).
The coefficient on an oriented edge \(v_1\to v_2\) measures how much
flow is traveling in the direction \(v_1\to v_2\),
with negatives representing flow in the opposite direction.
That the sum of coefficients on outgoing edges equals the sum of coefficients on incoming edges reflects the defining property of a circulation, namely that the
\emph{conservation law} holds at every vertex:
the total flow in equals the total flow out.

Conservation is maintained under linear combinations. 
Since every degree-2 polynomial $p$ in \(\VanIdeal{\RFEparam01}\) is a
linear combination of instantiations of \(\VanGenparam01\), the representation of $p$ also satisfies the conservation law at every vertex, \ie, the representation of $p$ is a circulation. 
Thus, conservation is a necessary condition for membership in
\(\VanIdeal{\RFEparam01}\).

Conservation is \emph{sufficient} for ideal membership, as well. By definition, conservation at every vertex means that the representation is a circulation. By the well-known flow decomposition theorem (see, \eg, \cite[p.~80-81]{AMO93}), every circulation can be decomposed into a superposition of circulations around simple cycles. A unit circulation
around a simple cycle can be decomposed into a sum of unit circulations around 3-cycles; this is depicted for a 5-cycle below, where each edge indicates unit flow:
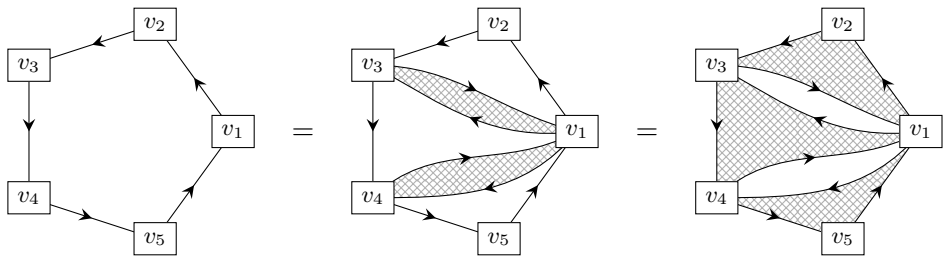
\[\vcenter{\hbox{\begin{tikzpicture}
    \coordinate (x1) at (  0:1.5);
    \coordinate (x2) at ( 72:1.5);
    \coordinate (x3) at (144:1.5);
    \coordinate (x4) at (216:1.5);
    \coordinate (x5) at (288:1.5);
    \draw[mid arrow] (x1) -- (x2);
    \draw[mid arrow] (x2) -- (x3);
    \draw[mid arrow] (x3) -- (x4);
    \draw[mid arrow] (x4) -- (x5);
    \draw[mid arrow] (x5) -- (x1);
    \node[draw=black,fill=white] at (x1) {\footnotesize \(v_1\)};
    \node[draw=black,fill=white] at (x2) {\footnotesize \(v_2\)};
    \node[draw=black,fill=white] at (x3) {\footnotesize \(v_3\)};
    \node[draw=black,fill=white] at (x4) {\footnotesize \(v_4\)};
    \node[draw=black,fill=white] at (x5) {\footnotesize \(v_5\)};
  \end{tikzpicture}}}
  \quad=\quad
  \vcenter{\hbox{\begin{tikzpicture}
    \coordinate (x1) at (  0:1.5);
    \coordinate (x2) at ( 72:1.5);
    \coordinate (x3) at (144:1.5);
    \coordinate (x4) at (216:1.5);
    \coordinate (x5) at (288:1.5);
    \fill[pattern=crosshatch,pattern color=black!30] (x1) to[out=187.2,in=330] (x3) to[out=0,in=165.6] (x1);
    \fill[pattern=crosshatch,pattern color=black!30] (x1) to[out=230.4,in=0] (x4) to[out=45,in=208.8] (x1);
    \draw[mid arrow] (x1) -- (x2);
    \draw[mid arrow] (x2) -- (x3);
    \draw[mid arrow] (x3) -- (x4);
    \draw[mid arrow] (x4) -- (x5);
    \draw[mid arrow] (x5) -- (x1);
    \draw (x1) edge [mid arrow,out=187.2,in=330] (x3);
    \draw (x3) edge [mid arrow,in=165.6,out=0] (x1);
    \draw (x1) edge [mid arrow,out=230.4,in=0] (x4);
    \draw (x4) edge [mid arrow,in=208.8,out=45] (x1);
    \node[draw=black,fill=white] at (x1) {\footnotesize \(v_1\)};
    \node[draw=black,fill=white] at (x2) {\footnotesize \(v_2\)};
    \node[draw=black,fill=white] at (x3) {\footnotesize \(v_3\)};
    \node[draw=black,fill=white] at (x4) {\footnotesize \(v_4\)};
    \node[draw=black,fill=white] at (x5) {\footnotesize \(v_5\)};
  \end{tikzpicture}}}
  \quad=\quad
  \vcenter{\hbox{\begin{tikzpicture}
    \coordinate (x1) at (  0:1.5);
    \coordinate (x2) at ( 72:1.5);
    \coordinate (x3) at (144:1.5);
    \coordinate (x4) at (216:1.5);
    \coordinate (x5) at (288:1.5);
    \fill[pattern=crosshatch,pattern color=black!30] (x1) to (x2) to (x3) to[out=0,in=165.6] (x1);
    \fill[pattern=crosshatch,pattern color=black!30] (x1) to[out=187.2,in=330] (x3) to (x4) to[out=45,in=208.8] (x1);
    \fill[pattern=crosshatch,pattern color=black!30] (x1) to[out=230.4,in=0] (x4) to (x5) to (x1);
    \draw[mid arrow] (x1) -- (x2);
    \draw[mid arrow] (x2) -- (x3);
    \draw[mid arrow] (x3) -- (x4);
    \draw[mid arrow] (x4) -- (x5);
    \draw[mid arrow] (x5) -- (x1);
    \draw (x1) edge [mid arrow,out=187.2,in=330] (x3);
    \draw (x3) edge [mid arrow,in=165.6,out=0] (x1);
    \draw (x1) edge [mid arrow,out=230.4,in=0] (x4);
    \draw (x4) edge [mid arrow,in=208.8,out=45] (x1);
    \node[draw=black,fill=white] at (x1) {\footnotesize \(v_1\)};
    \node[draw=black,fill=white] at (x2) {\footnotesize \(v_2\)};
    \node[draw=black,fill=white] at (x3) {\footnotesize \(v_3\)};
    \node[draw=black,fill=white] at (x4) {\footnotesize \(v_4\)};
    \node[draw=black,fill=white] at (x5) {\footnotesize \(v_5\)};
  \end{tikzpicture}}}
\]
The basis of the first equality in the above figure is that a unit flow \(v_1 \to v_3\)
cancels with a unit flow \(v_3\to v_1\), and similar for $v_4$ in lieu of $v_3$. Thus, conservation implies that we have a linear combination of unit circulations on 3-cycles, \ie, a linear combination of instantiations of \(\VanGenparam{0}{1}\). 

In summary, a multilinear homogeneous degree-2 polynomial is in
\(\VanIdeal{\RFEparam01}\) \emph{if and only if} its graphical representation satisfies the conservation law at every vertex. This is the representation and ideal membership characterization in the basic setting with \(k=0\), \(l=1\), and \(d=2\) for multilinear homogeneous polynomials. Note that, in this basic setting, the multilinear homogeneous degree-2 case represents the core of the problem. The remaining cases contain a univariate monomial, and are outside of $\RFEparam{0}{1}$ by \expref{Proposition}{prop.ideal.membership.degree-bound}.

\subsection{Alternating algebra}
\label{s.simplicial-repr.alternating-algebra}

In order to generalize \expref{Subsection}{s.simplicial-repr.base-case},
we need to be able to discuss higher-dimensional analogues of ``flow''
and ``circulation'',
as well as appropriately-generalized notions of ``conservation.''
Suited to this purpose is the language of \emph{alternating algebra}.
Alternating algebra was introduced in the 1800s by Hermann Grassmann
\cite{Grassmann1844,GrassmannKannenberg2000} and is the formalism
underlying differential geometry and its applications to physics.
We give a brief introduction to alternating algebra here,
tailored toward our purposes.

For each \(i\in [n]\), we create a fresh vertex $v_i \in \Vertices$, 
which corresponds to the variable $x_i$.
The alternating algebra provides a multiplication, denoted \(\wedge\),
that can be thought of as a constructor to make \emph{oriented simplices}
out of these vertices.
For example, the \(\wedge\)-product of \(v_1\) with \(v_2\),
written \(v_1\wedgec v_2\), encodes the simplex with vertices
\(v_1\) and \(v_2\) in a particular orientation; \(v_2\wedgec v_1\) encodes the same simplex with the opposite orientation. 
When \(v_1=v_2\), \(v_1\wedgec v_2\) is defined to be zero.
\(\wedge\)-multiplication is associative.
Rather than being commutative,
the \(\wedge\)-product is \emph{anti-commutative} in the sense that
\(v_1\wedgec v_2=-v_2\wedgec v_1\).
In this way the order of the vertices in the product encodes
an orientation.
There are only ever two orientations.
In a larger product such as
\(v_1\wedgec v_2\wedgec v_3\),
we have
\[\begin{array}{crcl}
      &\;   v_1\wedgec v_2\wedgec v_3
  \;&=&\; - v_1\wedgec v_3\wedgec v_2 \\
  \; =&\;   v_3\wedgec v_1\wedgec v_2
  \;&=&\; - v_3\wedgec v_2\wedgec v_1 \\
  \; =&\;   v_2\wedgec v_3\wedgec v_1
  \;&=&\; - v_2\wedgec v_1\wedgec v_3.
\end{array}\]
In general, permuting the vertices in a \(\wedge\)-product by an
even permutation has no effect,
while permuting by an odd permutation flips the sign.
Any \(\wedge\)-product that uses the same vertex more than once is
zero.

We can formally extend $\wedge$-multiplication to linear combinations of 
vertices in $\Vertices$. Denote $\SpanVertices$ to be the $\FF$-vector space
with basis $\Vertices$. The $\wedge$-multiplication extends to $\SpanVertices$ by
being \emph{distributive}. Overall, $\wedge$-multiplication 
has the following defining properties, for any $u_1,u_2,u_3\in \SpanVertices$:
\begin{itemize}
    \item \emph{Associativity:} $u_1\wedge (u_2\wedge u_3) = (u_1\wedge u_2)\wedge u_3$.
    \item \emph{Distributivity:} $u_1\wedge (u_2+u_3) = u_1\wedge u_2 + u_1\wedge u_3$.
    \item \emph{Alternation:} $u_1 \wedge u_1 = 0$.
\end{itemize}
The alternation property implies anti-commutativity\footnote{Alternation and anti-commutativity are equivalent provided the characteristic of the field differs from 2.}: $u_1\wedge u_2=-u_2\wedge u_1$.
The alternating algebra consists of all formal linear combinations of \(\wedge\)-products of vertices from $\Vertices$, or equivalently, of elements from $\SpanVertices$. 
We denote the underlying universe as follows.
\begin{definition}[space of oriented simplices]
  \label{def.simplex-space}
  For each \(t\in\N\),
  we let
   \[
    \Simplicest
    \doteq
    \linspan (u_1 \wedgecdots u_t : u_1,\ldots,u_t \in \SpanVertices)
  \]
  denote the space of linear combinations of \(t\)-vertex oriented simplices. For a set of indices $T$, we write $u^T \doteq \bigwedge_{i\in T} u_i$, with the convention that the indices are listed in increasing order. 
\end{definition}
The distributivity and anti-commutativity properties of $\wedge$ imply that 
\[\Simplicest = \linspan (u^{[t]} : u_1,\dots,u_t \text{ are distinct elements in } \Vertices),\] 
which justifies the reference to $t$-vertex simplices. The properties also imply that changing the order of the vertices in the wedge product yields the same element up to a sign, namely the sign of the underlying permutation. This justifies the reference to orientation, where there are two possible orientations. To emphasize, the \(t\) in \(\Simplicest\) counts the number of vertices in the simplices; this is one more than the usual notion of dimension of a simplex. For $t=0$, we have a distinct simplex corresponding to the empty product, denoted \(1\), which is an identity for \(\wedge\). Note that not every element in $\Simplicest$ can be expressed in the form $u_1\wedgecdots u_t$. 

To connect this with \expref{Subsection}{s.simplicial-repr.base-case},
recall the graphical depiction of
\(\VanGenparam01[{i_1}, {i_2}, {i_3}]\):
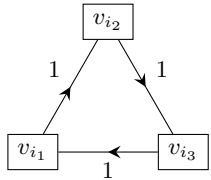
\begin{center}
  \begin{tikzpicture}
    \node[draw=black,rectangle] (x) at (0,0) {\footnotesize $v_{i_1}$};
    \node[draw=black,rectangle] (y) at (1,1.733) {\footnotesize $v_{i_2}$};
    \node[draw=black,rectangle] (z) at (2,0) {\footnotesize $v_{i_3}$};
    \draw[mid arrow] (x) -- node[auto] {\footnotesize $1$} (y);
    \draw[mid arrow] (y) -- node[auto] {\footnotesize $1$} (z);
    \draw[mid arrow] (z) -- node[auto] {\footnotesize $1$} (x);
  \end{tikzpicture}
\end{center}
Adopting the convention that an arrow \(v_1\to v_2\) is
\(v_1\wedgec v_2\)
(and so an arrow \(v_2\to v_1\) is
\(v_2\wedgec v_1 = -v_1\wedgec v_2\)),
we can alternatively express the above as
\[
  v_{i_1}\wedgec v_{i_2} \;+\; v_{i_2}\wedgec v_{i_3} \;+\; v_{i_3}\wedgec v_{i_1}.
\]
In general, the graphical representation of a homogeneous degree-2 multilinear
polynomial is some linear combination of 2-vertex oriented simplices.
When we go to higher-degree polynomials,
we make use of oriented simplices with more vertices.

To express conservation, we introduce \emph{boundary maps}, which are parametrized by a linear weight function \(w : \SpanVertices \to \FF\). The boundary map \(\bdry_w\) is a linear map that sends each simplex to a linear combination of its boundary faces (and the empty simplex to zero) according to a formula reminiscent of the minor expansion of a determinant along a column consisting of the values of $w$.
\begin{definition}[boundary map]
  \label{def.polybdry}
  For any linear function \(w : \SpanVertices \to\FF\),
  the boundary map with weight function \(w\) is the linear map 
  \(\bdry_w : \bigoplus_{t=0}^{n}\Simplicest \to \bigoplus_{t=0}^{n}\Simplicest\)
  realizing
  \begin{equation}\label{eq.def.boundary}
    u_1 \wedgecdots u_t
    \mapsto
      \sum_{i=1}^t (-1)^{i+1} w(u_i) 
      (u_1 \wedgecdots u_{i-1}\wedgec u_{i+1}\wedgecdots u_t)
  \end{equation}
  for all $u_1, \dots, u_t \in \SpanVertices$.
\end{definition}
The boundary map $\bdry_w$ is well-defined. To see this, note that the sign factor $(-1)^{i+1}$ in \eqref{eq.def.boundary} ensures well-definedness of the restriction to vertices, \ie, for $u_1, \dots, u_t \in \Vertices$. This is because changing the order of the vertices on the left-hand side results in the correct sign change on the right-hand side. The linearity of $w$ then guarantees that the linear extension of the restriction to vertices coincides with \eqref{eq.def.boundary}. 
For each \(t \ge 1\),
\(\bdry_w(\Simplicest) \subseteq \Simplicesparam{t-1}\),
while \(\bdry_w(\Simplicesparam{0}) = \{0\}\).

In the simplest case, \(w\) is the function that is 1 on every $v \in V$.
In this case, the boundary of some 2-vertex simplex is given by
\[
  \bdry_1( v_1 \wedgec v_2 ) = v_2 - v_1.
\]
In particular, \(v_1 \wedgec v_2\) contributes \(-1\) toward \(v_1\)
and \(+1\) toward \(v_2\).
This coincides with the contribution of the edge \(v_1\to v_2\) toward the
net flow into the vertices \(v_1\) and \(v_2\).
In exactly this way, conservation is identified with having a
\emph{vanishing boundary}. Note also that for this choice of weight function
\[ \bdry_1(v_1 \wedgec v_2 \wedgec v_3)
= {v_2}\wedgec {v_3} \;-\; {v_1}\wedgec {v_3} \;+\; {v_1}\wedgec {v_2} = {v_1}\wedgec {v_2} \;+\; {v_2}\wedgec {v_3} \;+\; {v_3}\wedgec {v_1}.\]
Thus, unit circulations on 3-cycles are in one-to-one and onto correspondence with the images under $\bdry_1$ of oriented 3-simplices on the vertices. By the decomposition discussed in the  
\expref{Subsection}{s.simplicial-repr.base-case}, it follows that circulations are in one-to-one and onto correspondence with the elements of $\bdry_1(\Simplicesparam{3})$. This means that $\bdry_1(\Simplicesparam{3}) = \ker(\bdry_1) \cap \Simplicesparam{2}$.

In general, for every linear $w: \SpanVertices \to \FF$
\begin{equation}\label{eq.key}
\im(\bdry_w) = \ker(\bdry_w),
\end{equation}
or equivalently, 
$\bdry_w(\Simplicesparam{t}) = \ker(\bdry_w) \cap \Simplicesparam{t-1}$ for every $t \in [n]$.
This key relationship implies that taking the same boundary multiple times always vanishes.
That is, for any \(w\), \(\bdry_w \circ \bdry_w = 0\), often written as $\bdry_w^2 = 0$.
Another property is that for any \(w, w'\) and \(\beta,\beta'\in\FF\),
\(\bdry_{\beta w + \beta'w'} = \beta\bdry_w + \beta'\bdry_{w'}\),
which is to say that the boundary maps themselves are linear in \(w\).
It follows from these that, for any \(w,w'\), \(\bdry_w \circ \bdry_{w'} =
-\bdry_{w'}\circ \bdry_w\).
This means that the boundary maps themselves behave like an
alternating algebra, with \(\circ\) as the multiplication rather than
\(\wedge\).
For any \(w_1,\ldots,w_{k+1}\), write \(\omega = w_1\wedgecdots w_{k+1}\),
and define
\(\bdry_\omega = \bdry_{w_{k+1}} \circ \cdots \circ \bdry_{w_1}\).
That is, \(w_1\wedgecdots w_{k+1}\) means apply \(\bdry_{w_1}\),
then \(\bdry_{w_2}\), and so on, up to \(\bdry_{w_{k+1}}\).
The result is well-defined, and we borrow the shorthand notation introduced in \expref{Definition}{def.simplex-space}: $w^T \doteq \bigwedge_{j\in T} w_j$, where $T \subseteq [k+1]$ and the indices in the wedge product are taken in increasing order.

The image-kernel relationship \eqref{eq.key} extends as follows: For any linearly independent $w_1, \dots, w_{k+1}$,
\begin{equation}
  \label{eq.bdry-kernel-image}
\im(\bdry_{w_1 \wedgecdots w_{k+1}}) = 
\bigcap_{r=1}^{k+1} \ker(\bdry_{w_r}).
\end{equation}
If $w_1, \dots w_{k+1}$ are linearly dependent, then $\bdry_{w_1 \wedgecdots w_{k+1}}$ vanishes. In fact, a further generalization holds and will be useful. We include a proof for completeness. \eqref{eq.bdry-kernel-image} corresponds to the special case $\Delta=0$.
\begin{proposition}[generalized image-kernel relationship]
    \label{prop.bdry-multiple-kernel-image}
    For $k, \Delta \in \N$ and any linearly independent linear functions $w_1,\dots,w_{k+\Delta+1}: \SpanVertices \to \FF$
    \begin{equation}\label{eq.bdry-kernel-image.general}
    \linspan_{\substack{B \subseteq [k+\Delta+1]\\ |B|=k+1}}\im(\bdry_{w^B}) = 
    \bigcap_{\substack{B \subseteq [k+\Delta+1] \\ |B|=\Delta+1}}
    \ker(\bdry_{w^B}). 
    \end{equation}
\end{proposition}
\begin{proof}
    Extend $w_1,\dots,w_{k+\Delta+1}$ to a basis $w_1,\dots,w_n$ of all linear functions $\SpanVertices\to \FF$. We can interpret $w_1,\dots,w_n$ as a basis of the dual space $\SpanVertices^*$, and the mapping $(w,u) \mapsto w(u)$ as a bilinear form $\SpanVertices^*\times \SpanVertices\to\FF$. This means we can construct a dual basis $u_1,\dots,u_n\in \SpanVertices$ such that for $i,j\in [n]$, $w_j(u_i)$ is 1 if $i=j$ and 0 if $i\neq j$.

    In this particular basis $u_1,\dots,u_n$, the boundary maps with weight functions $w_j$ take a very simple form: The only term in \eqref{eq.def.boundary} that remains for $w=w_j$ is the one with $i=j$. More generally,  
    for $B, T \subseteq [n]$,
    \begin{equation}\label{eq.bdry.special} 
    \bdry_{w^B} (u^T) =
    \begin{cases}
        \pm u^{T \setminus B} & \text{if $B\subseteq T$}\\
        0 & \text{otherwise.}
    \end{cases}
    \end{equation}
    With this characterization, we can see that both the span of the images and the intersection of the kernels 
    coincide with
    \[\linspan(u^S: S \subseteq [n] \text{ with } |S \cap [k+\Delta+1]| \le \Delta). \]
    
    We obtain $\pm u^S$ as $\bdry_{w^B} (u^T)$ if and only if $T = B \sqcup S$. Such a choice of $T$ and $B$ with $B \subseteq [k+\Delta+1]$ and $|B|=k+1$ exists if and only if there are at least $k+1$ elements in $[k+\Delta+1] \setminus S$, or equivalently, $|S \cap [k+\Delta+1]| \le (k+\Delta+1)-(k+1) = \Delta$. This proves the equality for the span of the images.
    
    On the other hand, $u^S$ falls within $\ker(\bdry_{w^B})$ if and only if $B \not\subseteq S$. This is case for every $B \subseteq [k+\Delta+1]$ with $|B|=\Delta+1$ if and only if $S$ contains at most $\Delta$ elements in $[k+\Delta+1]$. This proves the equality of the intersection of the kernels.
\end{proof}

In the following subsection, we will need an explicit formula for computing $\bdry_{\omega}(u^T)$ for generic $u_1, \dots, u_n \in \SpanVertices$ and $T \subseteq [n]$. From a concrete perspective, the effect of a single boundary map in \expref{Definition}{def.polybdry} resembles one level of determinant minor expansion, so composing boundary maps should produce a partially expanded determinant. We formalize that intuition with the following proposition, which characterizes the boundary of a $t$-simplex after applying $k$ weighted boundaries as a linear combination of $(t-k)$-simplices. Each $(t-k)$-simplex is indexed by a subset $J$ of $T$.

\begin{proposition}[composed boundary maps]
    \label{prop.composed-bdrys}
    Let $w_1,\dots, w_{k+1}:\SpanVertices\to\FF$ be linear functions, $T$ a set of indices, and $u_i \in \SpanVertices$ for $i \in T$.
    \begin{equation}\label{eq.prop.composed-bdrys}
    \bdry_{w^{[k+1]}}(u^T) = \sum_{\substack{I\sqcup J=T \\ |I|=k+1}} 
    (-1)^{\XInv(I,J)} \cdot 
    \det \begin{bmatrix}
        w_{r}(u_i)
    \end{bmatrix}_{i\in I}^{r\in [k+1]}
    \cdot u^J, 
    \end{equation}
    where in the determinant, the rows from top to bottom and the columns from left to right are in increasing order of index $i$ and $r$, respectively.
\end{proposition}
\begin{proof}
    Observe that $\bdry_{w^{[k+1]}}(u^T)\in \Simplicesparam{t-k-1}$ and, by \expref{Definition}{def.polybdry}, can be written as a linear combination of $u^J$ over all $J\subseteq T$ with $|J|=|T|-k-1$. It suffices to show that the coefficients of each $u^J$ match the ones given above.

    Without loss of generality, let $T = [t]$, as this does not change the relative order of any determinants or
    $\wedge$-products. Consider the terms formed by iteratively expanding $\bdry_{w^{[k+1]}}(u^T) \doteq \bdry_{w_{1}\wedgecdots w_{k+1}}(u^T)$ by \expref{Definition}{def.polybdry}. Each term is in one-to-one correspondence with the choices of $i$ we make in the expansions of \expref{Definition}{def.polybdry}. In particular, the terms that yield $u^J$ correspond to the bijections $\sigma:[k+1]\to I$, where $I=T \setminus J$. For a given $\sigma$, the corresponding coefficient is equal to
    \[ (-1)^{(\sum_{i\in I} i) + k+1 - |\{r,r'\in[k+1]: r'<r,\sigma(r')<\sigma(r) \}|} \prod_{r\in [k+1]}w_{r}(u_{\sigma(r)}). \]
    The $|\{r'<r,\sigma(r')<\sigma(r) \}|$ term accounts for the fact that, when each $r$ is selected, some terms of $T$ may have been previously removed, shifting the relative rank of $r$. Since for any distinct $r,r'\in [k+1]$, either $\sigma(r')>\sigma(r)$ or $\sigma(r')<\sigma(r)$, we can rewrite $|\{r'<r,\sigma(r')<\sigma(r) \}| = \binom{k+1}{2} - |\{r'<r,\sigma(r')>\sigma(r) \}|$.
    
    As for the term $\sum_{i\in I}i$, writing $i$ as $i = |\{i' \in I: i' \le i\}| + |\{j \in J: j < i\}|$ and summing over all $i \in I$, we have that
    $\sum_{i \in I} i = \sum_{r=1}^{k+1} r + \XInv(I,J) 
    = \binom{k+2}{2} + \XInv(I,J)$. 
    
    As $\binom{k+2}{2} = \binom{k+1}{2} + k+1$, we get that the coefficient of $u^J$ equals
    \[ \sum_{\sigma:[k+1]\to I} (-1)^{\XInv(I,J)+|\{r'<r,\sigma(r')>\sigma(r) \}| + 2k+2} \prod_{r\in [k+1]} w_{r}(u_{\sigma(r)}), \]
    and by definition of the determinant and simplifying, this is equal to
\[ (-1)^{\XInv(I,J)}\det \begin{bmatrix}
    w_{r}(u_i)
      \end{bmatrix}_{i\in I}^{r\in [k+1]}\,.    \qedhere
\]
\end{proof}

In the next subsection we will apply \expref{Proposition}{prop.composed-bdrys} with $u_i=v_i$. For the choice of $u_i$ in the proof of \expref{Proposition}{prop.bdry-multiple-kernel-image}, the matrix in \eqref{eq.prop.composed-bdrys} is the identity matrix and thus has determinant $1$, which results in \eqref{eq.bdry.special}.

\subsection{General case}
\label{s.simplicial-repr.gen-case}

With the notation of alternating algebra in hand,
we turn now to generalizing the characterization of $\VanIdeal{\RFEkl}$ based on the representation of polynomials that we introduced in \expref{Subsection}{s.simplicial-repr.base-case}, henceforth the \emph{simplicial representation}. We focus on multilinear polynomials, but the parameters \(k,l\in\N\) may be arbitrary. For starters, we still restrict to degree $d=l+1$. We then generalize to multilinear polynomials of arbitrary degree and present an alternate proof to \expref{Theorem}{thm.ideal.membership}. We end with some thoughts about the non-multilinear case. 

As before, we associate each variable \(x_i\) with a distinct vertex \(v_i \in V\), where $U \doteq \linspan(V)$ denotes an underlying vector space over $\FF$. We view a polynomial as a linear combination of monomials and represent each
degree-$t$ multilinear monomial as an oriented simplex with $t$ vertices. The representation makes use of the Vandermonde determinants $\Adet{T}$ for $T \subseteq [n]$, where $\Amat{T}$ refers to the notation that we introduced in \eqref{eq.vandermonde-rep} for the Vandermonde matrix built from the abscissas $a_i$ for $i \in T$ in increasing order. The Vandermonde determinant $\Adet{T}$ can be written as the product of pairwise differences:
\begin{equation}\label{eq.vandermonde-rep.differences}
\Adet{T} = \prod_{i,j \in T, i<j} (a_i-a_j).
\end{equation}
In particular, as the abscissas are distinct, $\Adet{T}$ is always nonzero.

Let $v^T\doteq \bigwedge_{i\in T} v_i$, where the indices are listed in increasing order. We represent the monomial $x^T \doteq \prod_{i \in T} x_i$ for $T \subseteq [n]$ by the element $v^T/\Adet{T}$. Formally, we define the following ``decoder map,'' which maps a simplicial representation to the polynomial it represents.
\begin{definition}[representation]
  \label{def.Srepr-prim}
  \(\Srepr : \bigoplus_{t=0}^{n}\Simplicest \to \FF[x_1,\ldots,x_n]\) is the linear
  map extending
  \begin{equation}\label{eq.decoder-mapping}
    v^T
    \;\mapsto\;
    \Adet{T}\,\cdot\,x^T
  \end{equation}
  for every $T \subseteq [n]$.
\end{definition}
Note that \eqref{eq.decoder-mapping} holds irrespective of the order of the indices, as long as the same order is used for both $v^T$ and $\Adet{T}$. This is because exchanging any two indices changes the sign of both the left-hand side and the determinant on the right-hand side. The mapping \(\Srepr\) induces a vector space isomorphism between
\(\Simplicesparam{l+1}\) and the space of multilinear homogeneous degree-\((l+1)\) polynomials.

\medskip

The strategy for our membership test in \(\VanIdeal{\RFE}\) consists of two steps: First express \(\VanIdeal{\RFE}\) in terms of $\Srepr$ and the image of the boundary maps $\bdry_w$, and then apply the (generalized) image-kernel relationship from alternating algebra. In \expref{Definition}{def.polybdry}, $w$ is taken to be a linear function from $\SpanVertices$ to $\FF$. As a linear function, $w$ is completely defined by its values on the basis $\Vertices$. By Lagrange interpolation, every function $w$ from $\Vertices$ to $\FF$ can be viewed as a univariate polynomial of degree less than $n \doteq |\Vertices|$ restricted to the abscissas, namely the polynomial interpolating $a_i \mapsto w(v_i)$ for $i \in [n]$. 
\begin{definition}[degree of boundary map]
  \label{def.deg.polybdry}
  Let $w: \SpanVertices \to \FF$ be linear and $a_1, \dots, a_n$ be distinct elements of $\FF$. We say that $w$ is \emph{interpolated} by $q\in \FF[\alpha]$ if $w(v_i) = q(a_i)$ for $i\in [n]$. We say that $w$  is of degree $d$ if $w$ is interpolated by a degree-$d$ polynomial $q$.
\end{definition}

Furthermore, given fixed $a_1,\dots,a_n$, the correspondence between a weight function $w$ and its interpolating polynomial $q$ forms an isomorphism; if $w_1,w_2$ are interpolated by $q_1,q_2$, then $w_1+w_2$ is interpolated by $q_1+q_2$, and $cw_1$ is interpolated by $cq_1$. From now on, we directly refer to a weight function by the polynomial in $\FF[\alpha]$ that interpolates it. We will be interested in the boundaries that are weighted by low-degree polynomials.

\paragraph{Multilinear case for degree $d=l+1$.}
In the case of degree $d=l+1$, the first step of our approach boils down to finding a simplicial representation for the generators $\VanGenkl$. We do so using composed boundary maps of degree at most $k$.
\begin{lemma}
  \label{lem.simplicial-vangen}
  For any \(k,l \in \N\) and \(S\subseteq [n]\), $|S|=k+l+2$,
  \begin{equation} \label{eq.simplicial-vangen}
    \VanGenkl[S]
    =
    \Srepr \left( 
      \bdry_{\alpha^k\wedgecdots\alpha^0}\left(
        v^S
      \right)
     \right).
  \end{equation}
\end{lemma}
\noindent
That is, \(\VanGenkl\) is the polynomial formed from a given
\((k+l+2)\)-vertex simplex by iteratively applying to it the \(k+1\) boundaries weighted by \(\alpha^k, \alpha^{k-1}, \ldots, \alpha^0\) respectively, where \(\alpha^r\) stands for the weight function interpolated by the polynomial \(\alpha^r\).
\begin{proof}
Using our notation, the explicit expression \eqref{eq.EVC.expanded} in \expref{Proposition}{prop.vangen.basic} can be rewritten as
  \begin{equation}
    \label{eq.simplicial-vangen.expanded}
      \VanGenkl[S] = \sum_{\substack{K\sqcup L= S \\ |K|=k+1}}
    (-1)^{\XInv(K,L)}\cdot \Adet{K}\cdot \Adet{L}\cdot x^L.
  \end{equation}
  For the right-hand side, we use \expref{Proposition}{prop.composed-bdrys} to get:
  \begin{align*}
      \Srepr \left(
      \bdry_{\alpha^k\wedgecdots\alpha^0}\left(
        v^S
      \right)
    \right) 
    &= 
    \sum_{\substack{K\sqcup L= S \\ |K|=k+1}} 
    (-1)^{\XInv(K,L)} \cdot
    \Adet{K}\cdot
    \Srepr \left( v^L \right) \\
    &= 
    \sum_{\substack{K\sqcup L= S \\ |K|=k+1}} 
    (-1)^{\XInv(K,L)} \cdot
    \Adet{K}\cdot
    \Adet{L}\cdot
    x^L.
  \end{align*}
  The sum is identical to \eqref{eq.simplicial-vangen.expanded}.
\end{proof}

\expref{Lemma}{lem.simplicial-vangen} yields the following characterization of the part of $\VanIdeal{\RFEkl}$ of degree $l+1$. We state it in a format to which we can directly apply the image-kernel relationship \eqref{eq.bdry-kernel-image}.

\begin{corollary}
  \label{cor.simplicial-vangen}
  For any \(k,l \in \N\), the set of polynomials of degree $l+1$ in \(\VanIdeal{\RFEkl}\) is given by
  \[\Srepr(\bdry_{\alpha^k\wedgecdots\alpha^0}(\Simplicesparam{k+l+2})). \]
\end{corollary}
\begin{proof}
Since every degree-\((l+1)\) polynomial $p$ in \(\VanIdeal{\RFEkl}\) is a linear combination of instantiations of \(\VanGenkl\), \expref{Lemma}{lem.simplicial-vangen} allows to us to express the subset in \(\VanIdeal{\RFEkl}\) as
\[\linspan_{\substack{S \subseteq [n]\\ |S|=k+l+2}} \Srepr(\bdry_{\alpha^k\wedgecdots\alpha^0}(v^S)).\]
The result follows by linearity and the fact that $U = \linspan(V)$. 
\end{proof}

The image-kernel relationship \eqref{eq.bdry-kernel-image} then leads to the following membership test. Recall that $\Srepr$ induces an isomorphism from the space of $(l+1)$-vertex oriented simplices $\Simplicesparam{l+1}$ to the set of multilinear polynomials of degree $l+1$, so $\Srepr^{-1}$ is well-defined on multilinear polynomials.
\begin{theorem}
  \label{thm.vanideal-char.base}
  Let \(k,l \in \N\).
  For any multilinear polynomial \(p\in\FF[x_1,\ldots,x_n]\) of
  degree \(l+1\),
  \(p(\RFEkl) = 0\) if and only if $p$ is homogeneous of degree $l+1$ and
  \[
    \bdry_w( \Srepr^{-1}(p) ) = 0
  \]
  for every weight function \(w\) of degree at most \(k\).
\end{theorem}
\begin{proof}
The criterion in \expref{Corollary}{cor.simplicial-vangen} can be rewritten as 
\[ \Srepr^{-1}(p)\in \bdry_{\alpha^k\wedgecdots\alpha^0}(\Simplicesparam{k+l+2}). \]
By \expref{Proposition}{prop.bdry-multiple-kernel-image}, this is equivalent to
\[ \Srepr^{-1}(p)\in \left(\bigcap_{r=0}^{k} \ker(\bdry_{\alpha^r})\right)\cap \Simplicesparam{l+1}. 
\]
The intersection with $\Simplicesparam{l+1}$ means that $p$ is homogeneous of degree $l+1$. For such polynomials $p$, we have that \(p(\RFEkl) = 0\) if and only if  \( \bdry_{\alpha^r}(\Srepr^{-1}(p))=0 \) for $r=0,\dots,k$, which by linearity is equivalent to \( \bdry_{w}(\Srepr^{-1}(p))=0 \) for all weight functions \(w\) of degree at most \(k\).
\end{proof}
\expref{Theorem}{thm.vanideal-char.base} states that a multilinear polynomial $p$ of degree \(l+1\) is in the vanishing ideal of \(\RFEkl\) if and only if it is homogeneous of degree $l+1$ and the simplicial representation of $p$ satisfies conservation with respect to all degree-\(k\) boundaries. This is the representation and ideal membership characterization for such polynomials for general \(k\) and \(l\) in the special case of degree \(d = l+1\). As we will argue in \expref{Proposition}{prop.ideal.membership.altalg}, the characterization coincides with the membership test from \expref{Theorem}{thm.ideal.membership} for multilinear polynomials of degree \(l+1\). 

In~\expref{Section}{s.simplicial-repr.base-case} we considered the special case with $k=0$ and $l=1$. In that basic setting, the only weight functions of degree $k$ are the constant functions, and only $w \equiv 1$ needs to be considered in \expref{Theorem}{thm.vanideal-char.base}. The resulting criterion is exactly the conservation criterion that we developed in \expref{Section}{s.simplicial-repr.base-case}.

Note that the restriction in \expref{Theorem}{thm.vanideal-char.base} to \emph{multilinear} polynomials $p$ is just to ensure that $\Srepr^{-1}(p)$ is well-defined. For polynomials of degree $l+1$ that are not multilinear, one could interpret the non-existence of $\Srepr^{-1}(p)$ as not satisfying the criterion. This is consistent with \expref{Proposition}{prop.ideal.membership.degree-bound}, which implies that polynomials of degree $l+1$ that are not multilinear are automatically outside \(\VanIdeal{\RFEkl}\) since they necessarily have a monomial supported on $l$ or fewer variables.

Through \expref{Lemma}{lem.simplicial-vangen}, the property that $\bdry_w(\Srepr^{-1}(\VanGenkl)) = 0$ for every weight function $w$ of degree at most $k$ can be viewed as an application of $\bdry_{w \wedge \alpha^k \wedgecdots \alpha^0} = 0$ to $\Simplicest$ with $t = k+l+2$. The equations \eqref{eq.generator.dependencies} follow in a similar way from an application with $t = k+l+3$.

\paragraph{Multilinear case of arbitrary degree.}
The two-step approach underlying \expref{Theorem}{thm.vanideal-char.base} extends to multilinear polynomials of higher degrees. Whereas in the special case of degree $d=l+1$ we only needed simplicial representations for $\VanGenkl[S]$ in the first step, we now need them for polynomials of the more general form $\VanGenkl[S]\cdot x^M$ where $M \subseteq [n]$ is disjoint from $S$. We can handle the additional term $x^M$ in \expref{Lemma}{lem.simplicial-vangen} by including a multiplicative factor 
\begin{equation}\label{eq.mu}
\mu_M(\alpha) \doteq \prod_{j \in M}(\alpha - a_j)
\end{equation}
in each of the weight functions. The extra factor acts as a masking term and ensures that in the expansions of \eqref{eq.def.boundary} the terms with $i \in M$ vanish, so under $\Srepr$ the factor $x^M$ remains. 

\begin{lemma}
  \label{lem.simplicial-vangen-high-degree}
  For any \(k,l \in \N\), \(S\sqcup M \subseteq [n]\) with $|S|=k+l+2$, and $\mu_M(\alpha) \doteq \prod_{j \in M}(\alpha - a_j)$, 
  \begin{equation} \label{eq.simplicial-vangen-high-degree}
    \VanGenkl[S]\cdot x^M = 
    \frac{\Adet{S}}{\Adet{S \sqcup M}} \cdot \Srepr(\bdry_{\mu_M(\alpha)\alpha^k\wedgecdots \mu_M(\alpha)\alpha^0}(v^{S \sqcup M})).
  \end{equation}
\end{lemma}

\begin{proof}
Expand $\bdry_{\mu_M(\alpha)\alpha^k\wedgecdots \mu_M(\alpha)\alpha^0} (v^{S \sqcup M})$ by \expref{Proposition}{prop.composed-bdrys}. Notice that the only nonzero terms in the expansion correspond to subsets $J$ that contain $M$. Substituting $I \gets K$ and $J \gets L \sqcup M$, and factoring out the $\mu_M(a_i)$ terms from the determinant, we can write
\[ \bdry_{\mu_M(\alpha)\alpha^k\wedgecdots \mu_M(\alpha)\alpha^0}(v^{S \sqcup M}) = 
\sum_{\substack{K\sqcup L=S\\ |K|=k+1}} (-1)^{\XInv(K,L\sqcup M)} \cdot \left(\prod_{i\in K} \mu_M(a_{i})\right) \cdot \Adet{K} \cdot v^{L \sqcup M}. 
\]
Applying $\Srepr$ yields
\begin{equation}\label{eq.vanideal-char.multilinear.eq4}
\Srepr(\bdry_{\mu_M(\alpha)\alpha^k\wedgecdots \mu_M(\alpha)\alpha^0}(v^{S \sqcup M})) = 
\sum_{\substack{K\sqcup L=S\\ |K|=k+1}} (-1)^{\XInv(K,L\sqcup M)} \cdot \left(\prod_{i\in K} \mu_M(a_{i})\right) 
    \cdot \Adet{K} \cdot \Adet{L \sqcup M} \cdot x^{L \sqcup M}. 
\end{equation}
Applying \eqref{eq.vandermonde-rep.differences} to $T=L \sqcup M$, $T=L$, and $T=M$, rearranging terms, and remembering that $A_T$ takes rows in increasing index, we obtain
\begin{equation}
\label{eq.vandermonde-rep.two-sets}
\Adet{L \sqcup M} = (-1)^{\XInv(L,M)}\cdot \left( \prod_{i\in L,j\in M} (a_{i}-a_{j})\right) \cdot \Adet{L} \cdot \Adet{M}.
\end{equation}
We can expand $(-1)^{\XInv(K,L\sqcup M)}$ as the product $(-1)^{\XInv(K,L)} (-1)^{\XInv(K,M)}$ because
$\XInv(K,L \sqcup M)$ equals the sum $\XInv(K,L)+\XInv(K,M)$.
By the definition of $\mu_M$, we can expand $\prod_{i \in K} \mu_M(a_i)$ as $\prod_{i\in K, j \in M} (a_i-a_j)$. Those expansions and \eqref{eq.vandermonde-rep.two-sets} allow us to write the summand on the right-hand side of \eqref{eq.vanideal-char.multilinear.eq4} as
\[ (-1)^{\XInv(K,L)} (-1)^{\XInv(K,M)} (-1)^{\XInv(L,M)} \cdot \left( \prod_{i\in K \sqcup L, j \in M} (a_i-a_j) \right)  \cdot \Adet{K} \cdot \Adet{L} \cdot \Adet{M} \cdot x^{L \sqcup M} \]
Using the similar fact as above that  
$(-1)^{\XInv(K \sqcup L,M)} = (-1)^{\XInv(K,M)} (-1)^{\XInv(L,M)}$, recalling that $K \sqcup L = S$, and pulling out the terms independent of the choice of $K$, we obtain
\begin{align*}
&\Srepr(\bdry_{\mu_M(\alpha)\alpha^k\wedgecdots \mu_M(\alpha)\alpha^0}(v^{S \sqcup M})) \\
&= (-1)^{\XInv(S,M)} \cdot \left( \prod_{i\in S, j \in M} (a_i-a_j) \right)
\cdot \Adet{M} \cdot x^M \sum_{\substack{K\sqcup L=S\\ |K|=k+1}} (-1)^{\XInv(K,L)} \cdot \Adet{K} \cdot \Adet{L} \cdot x^L \\
&= \frac{\Adet{S \sqcup M}}{\Adet{S}} \cdot x^M \sum_{\substack{K\sqcup L=S\\ |K|=k+1}} (-1)^{\XInv(K,L)} \cdot \Adet{K} \cdot \Adet{L} \cdot x^L,
\end{align*}
where the last step applies \eqref{eq.vandermonde-rep.two-sets} with $L \gets S$. By \expref{Proposition}{prop.vangen.basic}, this establishes the result.
\end{proof}

The multilinear elements in \(\VanIdeal{\RFEkl}\) are exactly the linear combinations of terms of the form \eqref{eq.simplicial-vangen-high-degree} where $S \subseteq [n]$ ranges over subsets of size $k+l+2$ and $M \subseteq [n]$ over subsets disjoint with $S$. In order to obtain a simpler characterization of the same type, as well as one to which we can apply the generalized image-kernel relationship, we show that we can replace the weight functions on the right-hand side of \eqref{eq.simplicial-vangen-high-degree} by generic weight functions of the same degree or by Lagrange interpolants with respect to a subset of abscissas of size one more. 

\begin{proposition}\label{prop.simplicial-vangen-high-degree}
Let \(k+1,m,t \in \N\) with $t \ge k+1$, $\nu \in \Simplicest$, and $N \subseteq [n]$ with $|N|=k+m+1$. Let $L_{N,j}$ for $j \in N$ denote the Lagrange interpolants for the subset of abscissas $\{a_i\}_{i \in N}$, \ie, $L_{N,j}$ denotes the unique univariate polynomial of degree at most $|N|-1$ satisfying $L_{N,j}(a_i)=1$ for $i=j$ and $L_{N,j}(a_i) = 0$ for $i \in N \setminus \{j\}$.
For all weight functions $w_1, \dots, w_{k+1}$ of degree at most $k+m$, 
\begin{equation}\label{eq.prop.simplicial-vangen-high-degree}
\linspan_{\substack{M \subseteq N \\ |M|=m}} \bdry_{\mu_M(\alpha)\alpha^k\wedgecdots\mu_M(\alpha)\alpha^0}(\nu) \; \;
= 
\linspan_{\substack{w_1, \dots, w_{k+1} \in \FF[\alpha] \\ \deg(w_1), \dots \deg(w_{k+1}) \le k+m}}
\bdry_{w^{[k+1]}}(\nu) \;
=
\linspan_{\substack{B \subseteq N \\ |B|=k+1}} \bdry_{L_N^B}(\nu).
\end{equation}
\end{proposition}
Some explanation of the compact notation on the right-hand side of \eqref{eq.prop.simplicial-vangen-high-degree} is in order. First, we use $L_{N,j}$ to differentiate with the notation $L_j$ for Lagrange interpolants that we introduced in \expref{Definition}{def.sv}, where $L_j$ corresponds to $L_{[n],j}$. Second, for a subset $B \subseteq N$, we write $L_N^B$ as a shorthand for $\bigwedge_{j\in B} L_{N,j}$, where the indices in the wedge product are taken in increasing order. Finally, in the composed boundary operator $\bdry_{L_N^B}$, the Lagrange interpolant $L_{N,j}$ represents the weight function interpolated by $L_{N,j}$ as in 
\expref{Definition}{def.deg.polybdry}. 

\begin{proof}
The inclusion $\subseteq$ of the first equality in \eqref{eq.prop.simplicial-vangen-high-degree} follows because the weight functions $\mu_M(\alpha) \alpha^r$ for $r \in \{0,\dots,k\}$ have degree at most $k+|M|=k+m$. 

To argue the inclusion $\subseteq$ of the second equality in \eqref{eq.prop.simplicial-vangen-high-degree}, note that the Lagrange interpolants $L_{N,j}$ for $j \in N$ are linearly independent and that there are as many of them as the dimension of the space of polynomials of degree at most $|N|-1=k+m$, so they form a basis for that space. In particular, we can write all weight functions $w_1,\dots,w_{k+1}$ of degree at most $k+m$ as linear combinations of the Lagrange interpolants $L_{N,j}$, $j \in N$. By the distributivity and antisymmetry of the wedge product, this implies that
\begin{equation*}
\bdry_{w^{[k+1]}}(\nu)  \in 
\linspan_{\substack{B \subseteq N \\ |B|=k+1}} \bdry_{L_N^B}(\nu).
\end{equation*}

It remains to argue that the right-most side of \eqref{eq.prop.simplicial-vangen-high-degree} is included in the left-most side. Fix a subset $B \subseteq N$ of size $|B|=k+1$. Since the polynomials $L_{N,j}$ for $j \in B$ individually have roots in all but one element of $\{a_i\}_{i \in N}$, they collectively have common roots among exactly $|N|-|B|=m$ of these abscissas, which form a set $M \subseteq N$. Each $L_{N,j}$ can therefore be written as the product of $\mu_M$ and a polynomial of degree at most $k$, or equivalently, as a linear combination of $\mu_{M}(\alpha)\alpha^k,\dots,\mu_{M}(\alpha)\alpha^0$. Once again, by the distributivity and antisymmetry of the wedge product, we have that 
\begin{equation*}
\bdry_{L_N^B}(\nu) \in 
\linspan_{\substack{M \subseteq N \\ |M|=m}}
\bdry_{\mu_M(\alpha)\alpha^k\wedgecdots\mu_M(\alpha)\alpha^0}(\nu).
\qedhere
\end{equation*}
\end{proof}

The first equality in \eqref{eq.prop.simplicial-vangen-high-degree} connects weight functions as on the right-hand side of \eqref{eq.simplicial-vangen-high-degree} with generic ones of the same degree. This leads to the following simple characterization of the multilinear part of $\VanIdeal{\RFEkl}$ in terms of $\Srepr$ and the image of composed boundary maps. The characterization naturally decomposes into separate ones for the homogeneous components of the various degrees $d$. 

\begin{corollary}
  \label{cor.simplicial-vangen-high-degree}
  For any \(k,l \in \N\), the set of multilinear polynomials $p \in \FF[x_1,\dots,x_n]$ in \(\VanIdeal{\RFEkl}\) is given by the direct sum $\oplus_{d=0}^{n-k-1} H_d$ of homogeneous components of degree $d \in \{0,\dots,n-k-1\}$ given by 
  \begin{equation}\label{eq.cor.simplicial-vangen-high-degree}
  H_d \doteq \linspan \Srepr (\bdry_{w^{[k+1]}}(\Simplicesparam{k+d+1})), 
  \end{equation}
  where $w_1, \dots, w_{k+1}$ range over all weight functions of degree at most $k+d-l-1$.
\end{corollary}
For $d \le l$, the only possible choices for the weight functions $w_1, \dots w_{k+1}$ in \expref{Corollary}{cor.simplicial-vangen-high-degree} are linearly dependent, which implies that $\bdry_{w^{[k+1]}}$ vanishes and therefore $H_d$ only contains the zero polynomial. This is consistent with \expref{Proposition}{prop.ideal.membership.degree-bound}, as is the restriction $d \le n-k-1$.
\begin{proof}
By \expref{Theorem}{thm.ideal.generators} and the fact that all the instantiations $\VanGenkl$ are homogeneous of degree $l+1$, the multilinear elements in \(\VanIdeal{\RFEkl}\) are exactly the linear combinations of terms of the form \eqref{eq.simplicial-vangen-high-degree} where $S \subseteq [n]$ ranges over subsets of size $k+l+2$ and $M \subseteq [n]$ over subsets disjoint with $S$. The homogeneous component of degree $d$ equals the contributions of the combinations $(S,M)$ where $|M|=m \doteq d-l-1$. Since $S \sqcup M \subseteq [n]$ and $|S|+|M|=k+d+1$, it follows that $d \le n-k-1$.

Since the weight functions on the right-hand side of \eqref{eq.simplicial-vangen-high-degree} are of degree at most $|M|+k=k+d-l-1$, the homogeneous component of degree $d$ falls inside $H_d$. For the other inclusion, consider $\nu=v^T$ for $T \subseteq [n]$ with $|T|=t \doteq k+d+1$. The first equality in \eqref{eq.prop.simplicial-vangen-high-degree} applies for any $N \subseteq [n]$ with $|N|=k+m+1=t-l-1$. If we pick $N \subseteq T$, we have that $M \subseteq N \subseteq T$ and we can write $T$ as $T = S \sqcup M$ where $|S| = |T|-|M| = k+l+2$. Thus, each term on the left-most side of \eqref{eq.prop.simplicial-vangen-high-degree} is of the form of the boundary expression on the right-hand side of 
\eqref{eq.simplicial-vangen-high-degree}. By \expref{Lemma}{lem.simplicial-vangen-high-degree} and linearity, it follows that all of $H_d$ can be realized as homogeneous components of degree $d$ of polynomials in $\VanIdeal{\RFEkl}$.
\end{proof}
For $d=l+1$, up to constant factors, there is only one nontrivial composed boundary map $\bdry_{w^{[k+1]}}$ up to scalar multiplication, namely the map $\bdry_{\alpha^k\wedgecdots\alpha^0}$ from \expref{Corollary}{cor.simplicial-vangen}. Thus, \expref{Corollary}{cor.simplicial-vangen} represents the special case of \expref{Corollary}{cor.simplicial-vangen-high-degree} for degree $d=l+1$. 

The second equality in \eqref{eq.prop.simplicial-vangen-high-degree} from \expref{Proposition}{prop.simplicial-vangen-high-degree} leads to another characterization of the multilinear part of $\VanIdeal{\RFEkl}$ in terms of $\Srepr$ and composed boundary maps, one that is more technical but to which we can directly apply the generalized image-kernel relationship. This leads to the following test for membership of multilinear polynomials in $\VanIdeal{\RFEkl}$. Consistent with the characterization in \expref{Corollary}{cor.simplicial-vangen-high-degree} and with \expref{Proposition}{prop.rfe-homog}, the test decomposes into independent ones for each of the homogeneous components.

\begin{theorem}
  \label{thm.vanideal-char.multilinear}
  Let \(k,l \in \N\).
  For any multilinear polynomial
  \(p\in\FF[x_1,\ldots,x_n]\),
  \(p(\RFEkl) = 0\) if and only if the homogeneous components $p^{(d)}$ of $p$ for all degrees $d$ satisfy the following requirements:
  \begin{enumerate}
      \item $p^{(d)}=0$ if $d \le l$ or $d \ge n-k$. \label{cond.thm.vanideal-char.multilinear-1}
      \item For all $d = l+\Delta+1$ with $\Delta \in \{0, \dots, n-k-l-2\}$ and all weight functions \(w_1,\ldots,w_{\Delta+1}\) of degree at most \(k+\Delta\) \label{cond.thm.vanideal-char.multilinear-2},
      \begin{equation}\label{eq.thm.vanideal-char.multilinear}
      \bdry_{w_1\wedgecdots w_{\Delta+1}}( \Srepr^{-1}(p^{(d)})  ) = 0.
      \end{equation}
  \end{enumerate}
\end{theorem}

\begin{proof}
Consider the characterization \eqref{eq.cor.simplicial-vangen-high-degree} of the homogeneous components $H_d$ in \expref{Corollary}{cor.simplicial-vangen-high-degree}. We already argued that $H_d$ only contains the zero polynomial for $d \le l$ and that there are no terms for $d \ge n-k-1$. This gives us \expref{condition}{cond.thm.vanideal-char.multilinear-1}.

In the remainder of the proof we consider the requirements for $d \in \{l+1, \dots, n-k-1\}$. For multilinear $p$, $\Srepr^{-1}(p)$ is well-defined. Applying $\Srepr^{-1}$ and the second equality in \eqref{eq.prop.simplicial-vangen-high-degree}, we can alternately write \eqref{eq.cor.simplicial-vangen-high-degree} as
\begin{equation}\label{eq.thm.vanideal-char.multilinear-1}
\Srepr^{-1}(H_d) = \linspan_{\substack{B \subseteq N \\ |B|=k+1}} \bdry_{L_N^B}(\Simplicesparam{k+d+1}),
\end{equation}
where $N \subseteq [n]$ can be any fixed subset of size $|N|=k+d-l$, namely by setting $m=d-l-1$, which we know is non-negative. For easier notation, we pick $N = [k+d-l]$. By \expref{Proposition}{prop.bdry-multiple-kernel-image} with $w_j \doteq L_{N,j}$ and $\Delta \doteq d-l-1$, we can further rewrite the right-hand side of \eqref{eq.thm.vanideal-char.multilinear-1} as
\[
\Srepr^{-1}(H_d) = \bigcap_{\substack{B \subseteq [k+\Delta+1] \\ |B|=\Delta+1}} \ker(\bdry_{L_N^B}) \cap \Simplicesparam{d}.
\]
Thus $p^{(d)} \in H_d$ if and only if 
\begin{equation}\label{eq.thm.vanideal-char.multilinear-2}
(\forall B \subseteq [k+\Delta+1] \textrm{ with }|B|=\Delta+1) \; \;  
\bdry_{L_N^B}(\Srepr^{-1}(p^{(d)}))=0.
\end{equation}
Another application of the second part of \expref{Proposition}{prop.simplicial-vangen-high-degree}, this time with $k \gets \Delta$, $m \gets k$, $t \gets d$, $\nu=\Srepr^{-1}(p^{(d)})$, and $N=[k+\Delta+1]$, shows that if \eqref{eq.thm.vanideal-char.multilinear-2} holds for the particular choice of weight functions $w_j = L_{N,j}$, then \eqref{eq.thm.vanideal-char.multilinear-2} holds for all choices of weight functions $w_j$ of degree at most $k+\Delta$. The statement follows.
\end{proof}

As we will argue in more detail below, by another application of the first part of \expref{Proposition}{prop.simplicial-vangen-high-degree}, it suffices in \expref{condition}{cond.thm.vanideal-char.multilinear-2} of \expref{Theorem}{thm.vanideal-char.multilinear} to consider weight functions of the form $w_j(\alpha)=\mu_K(\alpha) \alpha^{\Delta-j+1}$ for $j \in [\Delta+1]$, where $K$ ranges over all subsets of size $k$ of some fixed $N \subseteq [n]$ with $|N|=k+\Delta+1$. In this case, \eqref{eq.thm.vanideal-char.multilinear} becomes
\begin{equation}\label{eq.thm.vanideal-char.multilinear-3}
\bdry_{\mu_K(\alpha)\alpha^{\Delta} \wedgecdots \mu_K(\alpha)\alpha^0}(  \Srepr^{-1}(p^{(d)}) ) = 0.
\end{equation}
The left-hand side of \eqref{eq.thm.vanideal-char.multilinear-3} lives in $\Simplicesparam{l}$, and the condition is equivalent to requiring that the coefficient of $v^L$ vanishes for every subset $L \subseteq [n]\setminus K$ of size $|L|=l$. Those coefficients can be expressed in terms of evaluations of $\PartialZero{L}{K}{p^{(d)}}$, where we take the partial derivative with respect to the variables $x_i$ for $i \in L$ and set the variables $x_i$ for $i \in K$ to zero. Intuitively, whereas in \expref{Lemma}{lem.simplicial-vangen-high-degree} the effect of the masking factors $\mu_M$ was to retain only contributions of monomials that contain $x_i$ for every $i \in M$, in this dual setting the effect of $\mu_K$ is to cancel the contributions of monomials that contain $x_i$ for at least one $i \in K$. 

\begin{proposition}
    \label{prop.ideal.membership.altalg}
    Let $p \in \FF[x_1,\dots,x_n]$ be a multilinear polynomial, let $K,L \subseteq [n]$ be disjoint subsets with $|K|=k$ and $|L|=l$, and $\Delta \in \N$. Let $c_{K,L}$ denote the coefficient of $v^L$ in $\bdry_{\mu_K(\alpha)\alpha^{\Delta} \wedgecdots \mu_K(\alpha)\alpha^0}(\Srepr^{-1}(p))$, and $e_{K,L}$ denote the value of $\PartialZero{L}{K}{p}$
    upon the substitution $x_i\gets \mu_K(a_i)/\mu_L(a_i)$ for $i \in [n]\setminus (K \sqcup L)$. Then $c_{K,L} = e_{K,L} / \Adet{L}$. 
\end{proposition}

\begin{proof}
    By linearity, it suffices to establish the result for monomials $p=x^T$ where $T \subseteq [n]$. In such a case $\Srepr^{-1}(p)=v^T/\Adet{T}$. 

    If $L \not\subseteq T$, then $c_{K,L}$ vanishes because boundary maps can only remove components from a wedge product, not insert new components (see \eqref{eq.bdry.special}). On the other hand,  $\PartialZero{L}{K}{x^T}$ is identically zero because we are taking a partial derivative with respect to a variable that does not appear, so $e_{K,L}$ vanishes and the equality holds. 
    
    If $L \subseteq T$, then by applying \expref{Proposition}{prop.composed-bdrys} to $v^T$ and scaling, 
    \[ c_{K,L} = (-1)^{\XInv(M,L)} \prod_{i\in M} \mu_{K}(a_i)\cdot \Adet{M} / \Adet{T}, \]
    where $M \doteq T \setminus L$. 
    Note that if $K\cap M \ne \emptyset$ then the term $\prod_{i\in M} \mu_{K}(a_i)$ vanishes, hence $c_L$ vanishes. On the other hand, $\PartialZero{L}{K}{x^T}$ is identically zero because we are setting a variable to zero that appears in the monomial $x^T$. So, $e_{K,L}$ vanishes and the equality holds. 
    
    The remaining cases are those where $L \subseteq T$ and $K\cap M = \emptyset$. By \eqref{eq.vandermonde-rep.two-sets} 
    \[ \Adet{T} = \Adet{M \sqcup L} = (-1)^{\XInv(M,L)} \cdot \left( \prod_{i\in M}\prod_{j\in L}(a_i-a_j) \right) \cdot \Adet{M}\cdot \Adet{L}. \]
    Combining this with the notation $\mu_L(a_i) \doteq  \prod_{j\in L}(a_i-a_j)$, we can rewrite the expression for $c_{K,L}$ as
    \begin{equation}\label{eq.prop.ideal.membership.altalg-1}
    c_{K,L} = \left( \prod_{i\in M} \frac{\mu_K(a_i)}{\mu_L(a_i)}\right) \cdot \frac{1}{\Adet{L}}.
    \end{equation}
    On the other hand, we have that $\PartialZero{L}{K}{x^T} = x^M$, and the value upon the substitution $x_i\gets \mu_K(a_i)/\mu_L(a_i)$ for $i \in [n]\setminus (K \sqcup L)$ equals
    \begin{equation}\label{eq.prop.ideal.membership.altalg-2} e_{K,L} = \prod_{i\in M} \frac{\mu_K(a_i)}{\mu_L(a_i)}. 
    \end{equation}
    The result follows by comparing \eqref{eq.prop.ideal.membership.altalg-1} and \eqref{eq.prop.ideal.membership.altalg-2}. 
\end{proof}

In combination with \expref{Theorem}{thm.vanideal-char.multilinear}, the connection in \expref{Proposition}{prop.ideal.membership.altalg} yields an alternate proof of \expref{Theorem}{thm.ideal.membership}. It provides a membership test for the ideal generated by the instantiations of $\VanGenkl$ that, beyond the machinery of alternating algebra developed in this section, only requires the elementary properties of $\VanGenkl$ stated in \expref{Proposition}{prop.vangen.basic}. In particular, it does not make use of the \hyperref[lem.multi-zsub-deriv]{Zoom Lemma}, which we developed as a tool to obviate the need for alternating algebra after we had obtained our results. Note that the alternate approach to \expref{Theorem}{thm.ideal.membership} still relies on the Zoom Lemma for the connection to $\RFE$, namely in the argument that the ideal generated by the instantiations of $\VanGenkl$ includes all of $\VanIdeal{\RFEkl}$.

\begin{proof}[Alternate proof of \expref{Theorem}{thm.ideal.membership}]
Consider the membership test given by \expref{Theorem}{thm.vanideal-char.multilinear}. \expref{Condition}{cond.thm.vanideal-char.multilinear-1} is equivalent to 
\expref{condition}{thm.ideal.membership.cond1} in \expref{Theorem}{thm.ideal.membership}. It remains to argue that \expref{condition}{cond.thm.vanideal-char.multilinear-2} is equivalent to \expref{condition}{thm.ideal.membership.cond2} in \expref{Theorem}{thm.ideal.membership}. 

Fix $\Delta \in \{0,\dots,n-k-l-2\}$ and consider  \expref{Proposition}{prop.simplicial-vangen-high-degree} with $k \gets \Delta$, $m \gets k$, $t \gets d \doteq l+\Delta+1$, and $\nu = \Srepr^{-1}(p^{(d)})$. Set $N \subseteq [n]$ to be an arbitrary subset of size $N=k+\Delta+1$ and rename the set $M$ as $K$. The application of the first equality in \expref{Proposition}{prop.simplicial-vangen-high-degree} tells us that the combined requirements \eqref{eq.thm.vanideal-char.multilinear} over all choices of weight functions \(w_1,\ldots,w_{\Delta+1}\) of degree at most \(k+\Delta\) are equivalent to the combined requirements
\eqref{eq.thm.vanideal-char.multilinear-3} over all subsets $K \subseteq N$ of size $k$, or, because of the arbitrariness of $N$, over all subsets $K \subseteq [n]$ of size $k$. The left-hand side of \eqref{eq.thm.vanideal-char.multilinear-3} is a linear combination of terms of the form $v^L$, where $L \subseteq [n]$ is a subset of size $|L|=d-\Delta-1=l$ and is disjoint from $K$ because of the masking factor $\mu_K$ in all weight functions. Thus, \eqref{eq.thm.vanideal-char.multilinear-3} holds if and only if the coefficient $c_{K,L,d}$ of $v^L$ on the left-hand side vanishes for every such $L$. By \expref{Proposition}{prop.ideal.membership.altalg}, $c_{K,L,d}=0$ is equivalent to $e_{K,L,d}=0$, where $e_{K,L,d}$ denotes the value of $\PartialZero{L}{K}{p^{(d)}}$ upon the substitution $x_i\gets \mu_K(a_i)/\mu_L(a_i)$ for $i \in [n]\setminus (K \sqcup L)$.

In summary, \expref{condition}{cond.thm.vanideal-char.multilinear-2} in
\expref{Theorem}{thm.vanideal-char.multilinear} stipulates that for all disjoint subsets $K,L \subseteq [n]$ with $|K|=k$ and $|L|=l$,
\begin{equation}\label{eq.alt}
(\forall d \in \{l+1,\dots,n-k-1\}) \; e_{K,L,d} = 0.
\end{equation}
The value $e_{K,L,d}$ is also the coefficient of degree $d-l$ of the univariate polynomial in $z$ obtained from $\PartialZero{L}{K}{p}$ after the substitution  \eqref{eq.thm.ideal.membership.sub} from \expref{condition}{thm.ideal.membership.cond2} in \expref{Theorem}{thm.ideal.membership}. Since the range of $d$ in \eqref{eq.alt} covers all terms of this univariate polynomial in $z$, \eqref{eq.alt} is equivalent to the polynomial being zero, which is exactly \expref{condition}{thm.ideal.membership.cond2} in \expref{Theorem}{thm.ideal.membership}.
\end{proof}

\paragraph{Beyond multilinearity.}
\expref{Theorem}{thm.vanideal-char.multilinear}
does well for understanding the multilinear elements of the vanishing ideal.
For non-multilinear elements, one may do the following.
Let \(\ExtSimplicest\) be \(\Simplicest\) except that coefficients may be
arbitrary polynomials in \(\FF[x_1,\ldots,x_n]\) rather than just scalars in
\(\FF\).
The decoder map \(\Srepr\) and boundary maps \(\bdry_w\) carry over
to \(\ExtSimplicest\) directly,
though now \(\Srepr\) is no longer injective.
The following variation of \expref{Theorem}{thm.vanideal-char.base}
characterizes ideal membership for arbitrary polynomials.
\begin{proposition}
  \label{prop.vanideal-char.base-lift}
  Let \(k,l \in \N\).
  For any polynomial \(p\in\FF[x_1,\ldots,x_n]\),
  \(p(\RFEkl) = 0\) if and only if
  there exists \(\eta\in\ExtSimplicesparam{l+1}\)
  with \(\Srepr(\eta) = p\) such that,
  for every weight function $w$ of degree at most \(k\),
  \[
    \bdry_w(\eta) = 0.
  \]
\end{proposition}

\begin{proof}
    For the forward direction, we consider polynomials of the form $p= \VanGenkl[S] \cdot m$, where $S \subseteq [n]$ with $|S|=k+l+2$ and $m$ is a (not necessarily multilinear) monomial in $\FF[x_1,\dots,x_n]$. One choice of $\eta\in \ExtSimplicesparam{l+1}$ for which $\Srepr(\eta)=p$ is $\eta=(-1)^{(k+1)(l+1)} \bdry_{\alpha^k\wedgecdots\alpha^0}(v^S) \cdot m$, by \expref{Lemma}{lem.simplicial-vangen}. For this choice of $\eta$ and any weight function $w$ of degree at most $k$, $\bdry_{w}(\eta)=0$. The forward direction follows since every polynomial $p$ for which $p(\RFEkl)=0$ can be expressed as a linear combination of polynomials of the described form.

    For the backward direction, suppose there exists $\eta\in \ExtSimplicesparam{l+1}$ such that $\Srepr(\eta)=p$ and $\bdry_w(\eta)=0$ for all weight functions $w$ of degree at most $k$. We can write $\eta = \sum_m \omega_m 
    \cdot m$ as a linear combination of monomials $m\in \FF[x_1,\dots,x_n]$, each with coefficient $\omega_m\in \Simplicesparam{l+1}$. Since $\bdry_w$ does not affect polynomial coefficients by nonconstant factors, we have that for each $m$, $\bdry_w(\omega_m)$ vanishes for all $w$ of degree at most $k$. \expref{Theorem}{thm.vanideal-char.base} implies that $\Srepr(\bdry_w(\omega_m))$ is not hit by $\RFEkl$. By linearity, $\Srepr(\eta)$ is not hit by $\RFEkl$.
\end{proof}

While \expref{Proposition}{prop.vanideal-char.base-lift} applies to a broader class of
polynomials,
it has the drawback that representing polynomials with
\(\ExtSimplicesparam{l+1}\) is too redundant.
Specifically,
whenever \(p\) has a representation in \(\ExtSimplicesparam{l+1}\),
there are many \(\eta\in\ExtSimplicesparam{l+1}\) that represent \(p\),
and most of them do \emph{not} satisfy the boundary conditions,
even when \(p\) belongs to the vanishing ideal.
This erodes the utility of the characterization.
\expref{Theorems}{thm.vanideal-char.base}~and~\ref{thm.vanideal-char.multilinear} yield
straightforward tests:
Given \(p\),
form the unique \(\eta\) with \(\rho(\eta)=p\),
and then check whether the boundary conditions hold for \(\eta\).
\expref{Proposition}{prop.vanideal-char.base-lift}, on the other hand, leaves \(\eta\) underspecified.

\section*{Acknowledgements}
We are grateful to Herv\'{e} Fournier and Arpita Korwar for their presentation at WACT'18 in Paris \cite{FournierKorwar2018}. 
We are indebted to Gautam Prakriya for helpful discussions and detailed feedback. 
We also thank Amir Shpilka and Michael Forbes for comments and encouragement, the anonymous reviewers for their careful proofreading and interesting suggestions, and the ToC editors for their thorough work.  
Finally, we appreciate the partial support for this research by the U.S.\ National Science Foundation under Grants No.\ 1838434, 2137424, and 2312540. Any opinions, findings, and conclusions or recommendations expressed in this material are those of the authors and do not necessarily reflect the views of the National Science Foundation.

\bibliographystyle{tocplain}   %

\bibliography{bibstrings,v020a001,bibtail}

\newcommand{\prelim}{Preliminary version in\ }\newcommand{\prelims}{Preliminary
  versions in\ }
\providecommand{\bibhead}[1]{}
\expandafter\ifx\csname pdfbookmark\endcsname\relax%
  \providecommand{\tocrefpdfbookmark}{}
\else\providecommand{\tocrefpdfbookmark}{%
   \hypertarget{tocreferences}{}%
   \pdfbookmark[1]{References}{tocreferences}}%
\fi

\tocrefpdfbookmark
\begin{thebibliography}{10}

\bibitem{AL1994}\bibhead{AL1994}
{\sc William Adams and Philippe Loustaunau}: {\em An Introduction to
  Gr\"{o}bner Bases}.
\newblock Amer. Math. Soc., 1994.
\newblock [\epfmtdoi{10.1090/gsm/003}]

\bibitem{Agrawal2005}\bibhead{Agrawal2005}
{\sc Manindra Agrawal}: Proving lower bounds via pseudo-random generators.
\newblock In {\em Proc. 25th Found. Softw. Techn. Theoret. Comp. Sci. Conf.
  (FSTTCS'05)}, pp. 92--105. Springer, 2005.
\newblock [\epfmtdoi{10.1007/11590156\_6}]

\bibitem{AGKS2014}\bibhead{AGKS2014}
{\sc Manindra Agrawal, Rohit Gurjar, Arpita Korwar, and Nitin Saxena}:
  Hitting-sets for {ROABP} and sum of set-multilinear circuits.
\newblock {\em SIAM J. Comput.}, 44(3):669--697, 2015.
\newblock [\epfmtdoi{10.1137/140975103}]

\bibitem{ASS2013}\bibhead{ASS2013}
{\sc Manindra Agrawal, Chandan Saha, and Nitin Saxena}: Quasi-polynomial
  hitting-set for set-depth-{$\Delta$} formulas.
\newblock In {\em Proc. 45th STOC}, pp. 321--330. ACM Press, 2013.
\newblock [\epfmtdoi{10.1145/2488608.2488649}]

\bibitem{AMO93}\bibhead{AMO93}
{\sc Ravindra Ahuja, Thomas Magnanti, and James Orlin}: {\em Network Flows:
  Theory, Algorithms, and Applications}.
\newblock Pearson, 1993.
\newblock [\epfmtdoi{10.5555/137406}]

\bibitem{AFSSV2015}\bibhead{AFSSV2015}
{\sc Matthew Anderson, Michael~A. Forbes, Ramprasad Saptharishi, Amir Shpilka,
  and Ben~Lee Volk}: Identity testing and lower bounds for read-$k$ oblivious
  algebraic branching programs.
\newblock {\em ACM Trans. Comput. Theory}, 10(1/3):1--30, 2018.
\newblock [\epfmtdoi{10.1145/3170709}]

\bibitem{AvMV2011}\bibhead{AvMV2011}
{\sc Matthew Anderson, Dieter van Melkebeek, and Ilya Volkovich}: Deterministic
  polynomial identity tests for multilinear bounded-read formulae.
\newblock {\em Comput. Complexity}, 24(4):695--776, 2015.
\newblock [\epfmtdoi{10.1007/s00037-015-0097-4}]

\bibitem{AndrewsForbes2022}\bibhead{AndrewsForbes2022}
{\sc Robert Andrews and Michael~A. Forbes}: Ideals, determinants, and
  straightening: proving and using lower bounds for polynomial ideals.
\newblock In {\em Proc. 54th STOC}, pp. 389--402. ACM Press, 2022.
\newblock [\epfmtdoi{10.1145/3519935.3520025}]

\bibitem{arvind19}\bibhead{arvind19}
{\sc Vikraman Arvind, Pushkar~S. Joglekar, Partha Mukhopadhyay, and S.~Raja}:
  Randomized polynomial-time identity testing for noncommutative circuits.
\newblock {\em Theory of Computing}, 15(7):1--36, 2019.
\newblock [\epfmtdoi{10.4086/toc.2019.v015a007}]

\bibitem{BhargavaGhosh2021}\bibhead{BhargavaGhosh2021}
{\sc Vishwas Bhargava and Sumanta Ghosh}: Improved hitting set for orbit of
  {ROABPs}.
\newblock {\em Comput. Complexity}, 31(15), 2022.
\newblock [\epfmtdoi{10.1007/s00037-022-00230-9}]

\bibitem{CT2023}\bibhead{CT2023}
{\sc Prerona Chatterjee and Anamay Tengse}: On annihilators of explicit
  polynomial maps, 2023.
\newblock [\epfmt{arxiv}{2309.07612}]

\bibitem{CLO2013}\bibhead{CLO2013}
{\sc David~A. Cox, John Little, and Donal O'Shea}: {\em Ideals, Varieties, and
  Algorithms: An Introduction to Computational Algebraic Geometry and
  Commutative Algebra}.
\newblock Springer, 2013.
\newblock [\epfmtdoi{10.1007/978-3-319-16721-3}]

\bibitem{DeMilloLipton1978}\bibhead{DeMilloLipton1978}
{\sc Richard~A. Demillo and Richard~J. Lipton}: A probabilistic remark on
  algebraic program testing.
\newblock {\em Inform. Process. Lett.}, 7(4):193--195, 1978.
\newblock [\epfmtdoi{10.1016/0020-0190(78)90067-4}]

\bibitem{Forbes2015}\bibhead{Forbes2015}
{\sc Michael~A. Forbes}: Deterministic divisibility testing via shifted partial
  derivatives.
\newblock In {\em Proc. 56th FOCS}, pp. 451--465. IEEE Comp. Soc., 2015.
\newblock [\epfmtdoi{10.1109/FOCS.2015.35}]

\bibitem{FSS2014}\bibhead{FSS2014}
{\sc Michael~A. Forbes, Ramprasad Saptharishi, and Amir Shpilka}: Hitting sets
  for multilinear read-once algebraic branching programs, in any order.
\newblock In {\em Proc. 46th STOC}, pp. 867--875. ACM Press, 2014.
\newblock Full version \href{https://arxiv.org/abs/1309.5668}{arXiv:1309.5668}.
\newblock [\epfmtdoi{10.1145/2591796.2591816}]

\bibitem{FS2012}\bibhead{FS2012}
{\sc Michael~A. Forbes and Amir Shpilka}: On identity testing of tensors,
  low-rank recovery and compressed sensing.
\newblock In {\em Proc. 44th STOC}, p. 163–172. ACM Press, 2012.
\newblock [\epfmtdoi{10.1145/2213977.2213995}]

\bibitem{ForbesShpilka2013}\bibhead{ForbesShpilka2013}
{\sc Michael~A. Forbes and Amir Shpilka}: Quasipolynomial-time identity testing
  of non-commutative and read-once oblivious algebraic branching programs.
\newblock In {\em Proc. 54th FOCS}, pp. 243--252. IEEE Comp. Soc., 2013.
\newblock [\epfmtdoi{10.1109/FOCS.2013.34}]

\bibitem{FSTW21}\bibhead{FSTW21}
{\sc Michael~A. Forbes, Amir Shpilka, Iddo Tzameret, and Avi Wigderson}: Proof
  complexity lower bounds from algebraic circuit complexity.
\newblock {\em Theory of Computing}, 17(10):1--88, 2021.
\newblock [\epfmtdoi{10.4086/toc.2021.v017a010}]

\bibitem{FSV2018}\bibhead{FSV2018}
{\sc Michael~A. Forbes, Amir Shpilka, and Ben~Lee Volk}: Succinct hitting sets
  and barriers to proving lower bounds for algebraic circuits.
\newblock {\em Theory of Computing}, 14(18):1--45, 2018.
\newblock \prelim \href{https://doi.org/10.1145/3055399.3055496}{STOC'17}.
\newblock [\epfmtdoi{10.4086/toc.2018.v014a018}]

\bibitem{FournierKorwar2018}\bibhead{FournierKorwar2018}
{\sc Herv\'{e} Fournier and Arpita Korwar}: Limitations of the
  {Shpilka--Volkovich} generator.
\newblock In {\em Workshop on Algebraic Complexity Theory {(WACT)}, Paris},
  2018.
\newblock \href{http://wact.imj-prg.fr}{CONF} and
  \href{http://wact.imj-prg.fr/pdfs/slides-korwar.pdf}{SLIDES}.

\bibitem{GR2008}\bibhead{GR2008}
{\sc Ariel Gabizon and Ran Raz}: Deterministic extractors for affine sources
  over large fields.
\newblock {\em Combinatorica}, 28(4):415–440, 2008.
\newblock [\epfmtdoi{10.1007/s00493-008-2259-3}]

\bibitem{Grassmann1844}\bibhead{Grassmann1844}
{\sc Hermann Grassmann}: {\em {Die lineale Ausdehnungslehre, ein neuer Zweig
  der Mathematik: dargestellt und durch Anwendungen auf die {\"u}brigen Zweige
  der Mathematik, wie auch auf die Statik, Mechanik, die Lehre vom Magnetismus
  und die Krystallonomie erl{\"a}utert}}.
\newblock Otto Wigand, Leipzig, 1844.
\newblock Available at
  \href{https://babel.hathitrust.org/cgi/pt?id=nyp.33433017485941&seq=11}{HathiTrust}.

\bibitem{GrassmannKannenberg2000}\bibhead{GrassmannKannenberg2000}
{\sc Hermann Grassmann}: {\em Extension theory}.
\newblock Amer. Math. Soc., 2000.
\newblock Translation of \cite{Grassmann1844}, translated by {\sc{Lloyd C.
  Kannenberg}}.

\bibitem{GuoGurjar2020}\bibhead{GuoGurjar2020}
{\sc Zeyu Guo and Rohit Gurjar}: Improved explicit hitting-sets for {ROABPs}.
\newblock In {\em Proc. 24th Internat. Conf. on Randomization and Computation
  (RANDOM'20)}, pp. 4:1--16. Schloss Dagstuhl--Leibniz-Zentrum fuer Informatik,
  2020.
\newblock [\epfmtdoi{10.4230/LIPIcs.APPROX/RANDOM.2020.4}]

\bibitem{GKS2016}\bibhead{GKS2016}
{\sc Rohit Gurjar, Arpita Korwar, and Nitin Saxena}: Identity testing for
  constant-width, and any-order, read-once oblivious arithmetic branching
  programs.
\newblock {\em Theory of Computing}, 13(2):1--21, 2017.
\newblock [\epfmtdoi{10.4086/toc.2017.v013a002}]

\bibitem{GKST2015}\bibhead{GKST2015}
{\sc Rohit Gurjar, Arpita Korwar, Nitin Saxena, and Thomas Thierauf}:
  Deterministic identity testing for sum of read-once oblivious arithmetic
  branching programs.
\newblock {\em Comput. Complexity}, 26(4):835--880, 2017.
\newblock [\epfmtdoi{10.1007/s00037-016-0141-z}]

\bibitem{HeintzSchnorr1980}\bibhead{HeintzSchnorr1980}
{\sc Joos Heintz and Claus-Peter Schnorr}: Testing polynomials which are easy
  to compute.
\newblock In {\em Proc. 12th STOC}, pp. 262--272. ACM Press, 1980.
\newblock [\epfmtdoi{10.1145/800141.804674}]

\bibitem{ImpagliazzoWigderson1997}\bibhead{ImpagliazzoWigderson1997}
{\sc Russell Impagliazzo and Avi Wigderson}: {P=BPP} if {E} requires
  exponential circuits: Derandomizing the {XOR} lemma.
\newblock In {\em Proc. 29th STOC}, pp. 220--229. ACM Press, 1997.
\newblock [\epfmtdoi{10.1145/258533.258590}]

\bibitem{JQS2009}\bibhead{JQS2009}
{\sc Maurice Jansen, Youming Qiao, and Jayalal Sarma M.~N.}: Deterministic
  identity testing of read-once algebraic branching programs, 2009.
\newblock [\epfmt{arxiv}{0912.2565}]

\bibitem{JQS2010}\bibhead{JQS2010}
{\sc Maurice Jansen, Youming Qiao, and Jayalal Sarma M.~N.}: Deterministic
  black-box identity testing {$\pi$}-ordered algebraic branching programs.
\newblock In {\em Proc. 30th Found. Softw. Techn. Theoret. Comp. Sci. Conf.
  (FSTTCS'10)}, pp. 296--307. Schloss Dagstuhl--Leibniz-Zentrum fuer
  Informatik, 2010.
\newblock [\epfmtdoi{10.4230/LIPIcs.FSTTCS.2010.296}]

\bibitem{KabanetsImpagliazzo2004}\bibhead{KabanetsImpagliazzo2004}
{\sc Valentine Kabanets and Russell Impagliazzo}: Derandomizing polynomial
  identity tests means proving circuit lower bounds.
\newblock {\em Comput. Complexity}, 13(1--2):1--46, 2004.
\newblock [\epfmtdoi{10.1007/s00037-004-0182-6}]

\bibitem{KMSV2009}\bibhead{KMSV2009}
{\sc Zohar~S. Karnin, Partha Mukhopadhyay, Amir Shpilka, and Ilya Volkovich}:
  Deterministic identity testing of depth-4 multilinear circuits with bounded
  top fan-in.
\newblock {\em SIAM J. Comput.}, 42(6):2114--2131, 2013.
\newblock [\epfmtdoi{10.1137/110824516}]

\bibitem{KS2011}\bibhead{KS2011}
{\sc Zohar~S. Karnin and Amir Shpilka}: Black box polynomial identity testing
  of generalized depth-3 arithmetic circuits with bounded top fan-in.
\newblock {\em Combinatorica}, 31(3):333–364, 2011.
\newblock [\epfmtdoi{10.1007/s00493-011-2537-3}]

\bibitem{KS2001}\bibhead{KS2001}
{\sc Adam~R. Klivans and Daniel Spielman}: Randomness efficient identity
  testing of multivariate polynomials.
\newblock In {\em Proc. 33rd STOC}, p. 216–223. ACM Press, 2001.
\newblock [\epfmtdoi{10.1145/380752.380801}]

\bibitem{Korwar2021}\bibhead{Korwar2021}
{\sc Arpita Korwar}: Personal communication, 2021.

\bibitem{KuhnleM1996}\bibhead{KuhnleM1996}
{\sc Klaus K{\"{u}}hnle and Ernst~W. Mayr}: Exponential space computation of
  {Gr{\"{o}}bner} bases.
\newblock In {\em Proc. 21st Internat. Symp. Symbolic and Algebraic Computation
  (ISSAC'96)}, pp. 63--71. ACM Press, 1996.
\newblock [\epfmtdoi{10.1145/236869.236900}]

\bibitem{KumarSaraf2016}\bibhead{KumarSaraf2016}
{\sc Mrinal Kumar and Shubhangi Saraf}: Arithmetic circuits with locally low
  algebraic rank.
\newblock {\em Theory of Computing}, 13(6):1--33, 2017.
\newblock [\epfmtdoi{10.4086/toc.2017.v013a006}]

\bibitem{Mayr1997}\bibhead{Mayr1997}
{\sc Ernst~W. Mayr}: Some complexity results for polynomial ideals.
\newblock {\em J. Complexity}, 13(3):303--325, 1997.
\newblock [\epfmtdoi{10.1006/jcom.1997.0447}]

\bibitem{MediniShpilka2021}\bibhead{MediniShpilka2021}
{\sc Dori Medini and Amir Shpilka}: Hitting sets and reconstruction for dense
  orbits in $\mathrm{VP}_{e}$ and ${\Sigma\Pi\Sigma}$ circuits.
\newblock In {\em Proc. 36th Comput. Complexity Conf. (CCC'21)}, pp. 19:1--27.
  Schloss Dagstuhl--Leibniz-Zentrum fuer Informatik, 2021.
\newblock [\epfmtdoi{10.4230/LIPIcs.CCC.2021.19}]

\bibitem{conf-version}\bibhead{conf-version}
{\sc Dieter~van Melkebeek and Andrew Morgan}: Polynomial identity testing via
  evaluation of rational functions.
\newblock In {\em Proc. 13th Innovations in Theoret. Comp. Sci. Conf.
  (ITCS'22)}, pp. 119:1--24. Schloss Dagstuhl--Leibniz-Zentrum fuer Informatik,
  2022.
\newblock [\epfmtdoi{10.4230/LIPIcs.ITCS.2022.119}]

\bibitem{MinahanVolkovich2017}\bibhead{MinahanVolkovich2017}
{\sc Daniel Minahan and Ilya Volkovich}: Complete derandomization of identity
  testing and reconstruction of read-once formulas.
\newblock {\em ACM Trans. Comput. Theory}, 10(3/10):1--11, 2018.
\newblock [\epfmtdoi{10.1145/3196836}]

\bibitem{Nisan1991}\bibhead{Nisan1991}
{\sc Noam Nisan}: Lower bounds for non-commutative computation.
\newblock In {\em Proc. 23rd STOC}, pp. 410--418. ACM Press, 1991.
\newblock [\epfmtdoi{10.1145/103418.103462}]

\bibitem{NisanWigderson1994}\bibhead{NisanWigderson1994}
{\sc Noam Nisan and Avi Wigderson}: Hardness vs randomness.
\newblock {\em J. Comput. System Sci.}, 49(2):149--167, 1994.
\newblock \prelim \href{https://doi.org/10.1109/SFCS.1988.21916}{FOCS'88}.
\newblock [\epfmtdoi{10.1016/S0022-0000(05)80043-1}]

\bibitem{Ore1922}\bibhead{Ore1922}
{\sc {\O}ystein Ore}: {\"U}ber {h}\"ohere {K}ongruenzen.
\newblock {\em Norsk Mat. Forenings Skrifter, Ser. I}, 1(7):1--15, 1922.

\bibitem{RazShpilka2005}\bibhead{RazShpilka2005}
{\sc Ran Raz and Amir Shpilka}: Deterministic polynomial identity testing in
  non-commut\-ative models.
\newblock {\em Comput. Complexity}, 14(1):1--19, 2005.
\newblock [\epfmtdoi{10.1007/s00037-005-0188-8}]

\bibitem{SahaThankey2021}\bibhead{SahaThankey2021}
{\sc Chandan Saha and Bhargav Thankey}: Hitting sets for orbits of circuit
  classes and polynomial families.
\newblock In {\em Proc. 25th Internat. Conf. on Randomization and Computation
  (RANDOM'21)}, pp. 50:1--26. Schloss Dagstuhl--Leibniz-Zentrum fuer
  Informatik, 2021.
\newblock [\epfmtdoi{10.4230/LIPIcs.APPROX/RANDOM.2021.50},
  \epfmt{eccc}{TR21--015}]

\bibitem{Schwartz1980}\bibhead{Schwartz1980}
{\sc Jacob~T. Schwartz}: Fast probabilistic algorithms for verification of
  polynomial identities.
\newblock {\em J. ACM}, 27(4):701--717, 1980.
\newblock [\epfmtdoi{10.1145/322217.322225}]

\bibitem{ShpilkaVolkovich2008}\bibhead{ShpilkaVolkovich2008}
{\sc Amir Shpilka and Ilya Volkovich}: Read-once polynomial identity testing.
\newblock {\em Comput. Complexity}, 24(3):477--532, 2015.
\newblock [\epfmtdoi{10.1007/s00037-015-0105-8}]

\bibitem{ShpilkaYehudayoff2010}\bibhead{ShpilkaYehudayoff2010}
{\sc Amir Shpilka and Amir Yehudayoff}: Arithmetic circuits: A survey of recent
  results and open questions.
\newblock {\em Found. Trends Theor. Comp. Sci.}, 5(3--4), 2010.
\newblock [\epfmtdoi{10.1561/0400000039}]

\bibitem{Zippel1979}\bibhead{Zippel1979}
{\sc Richard Zippel}: Probabilistic algorithms for sparse polynomials.
\newblock In {\em Proc. Internat. Symp. Symbolic and Algebraic Computation
  (EUROSAM'79)}, pp. 216--226. ACM Press, 1979.
\newblock [\epfmtdoi{10.1007/3-540-09519-5\_73}]

\end{thebibliography}

\begin{tocauthors}
\begin{tocinfo}[hu]
 Ivan Hu\\
 \phd\ student\\
 Department of Computer Science\\
 University of Wisconsin -- Madison\\
 Madison, Wisconsin, USA\\
 ilhu\tocat{}wisc\tocdot{}edu\\           %
 \url{hhtps://pages.cs.wisc.edu/~ihu/}    %
\end{tocinfo}
\begin{tocinfo}[vanmelkebeek]
 Dieter van Melkebeek\\
 Professor\\
 Department of Computer Science\\
 University of Wisconsin -- Madison\\
 Madison, Wisconsin, USA\\
 dieter\tocat{}cs\tocdot{}wisc\tocdot{}edu \\   %
 \url{https://pages.cs.wisc.edu/~dieter/}      %
\end{tocinfo}
\begin{tocinfo}[morgan]
 Andrew Morgan\\
 Software engineer \\                 %
 Google \\                            %
 amorgan\tocat{}cs\tocdot{}wisc\tocdot{}edu \\
 \url{https://pages.cs.wisc.edu/~amorgan/}    %
\end{tocinfo}
\end{tocauthors}

\begin{tocaboutauthors}
\begin{tocabout}[hu]  %
   \textsc{Ivan Hu} is a second year \phd\ student at the
  \href{https://www.wisc.edu/}{University of Wisconsin–Madison}
  under the supervision of Dieter van Melkebeek.
  He is studying complexity theory, with interests in algebraic complexity and pseudorandomness. He is currently an NSF Graduate Research Fellow. %
\end{tocabout}
\begin{tocabout}[vanmelkebeek]
\textsc{Dieter van Melkebeek} received his \phd\ from the 
\href{http://www.uchicago.edu}{University of Chicago},
under the supervision of 
\href{http://lance.fortnow.com}{Lance Fortnow}. 
His thesis was awarded the 
\href{https://awards.acm.org/award-recipients/vanmelkebeek_7183705}{ACM 
Doctoral Dissertation Award}. 
After postdocs at 
\href{http://dimacs.rutgers.edu}{DIMACS} and the
\href{http://www.ias.edu}{Institute for Advanced Study},
he joined the faculty at the
\href{http://www.wisc.edu}{University of Wisconsin-Madison}, where he currently is a full professor. His research interests include the power of randomness, lower bounds for NP-complete problems, and connections between derandomization and lower bounds. 
\end{tocabout}
\begin{tocabout}[morgan]
    \textsc{Andrew Morgan} received his \phd\ in 2022 from the 
    \href{https://www.wisc.edu/}{University of Wisconsin–Madison} 
    under the supervision of Dieter van Melkebeek. After
    graduating, he became a software engineer at Google.
\end{tocabout}
\end{tocaboutauthors}

\end{document}